\documentclass[11pt,a4paper]{article}

\def\bea#1\eea{\begin{eqnarray}#1\end{eqnarray}}
\def\be#1\ee{\begin{equation}#1\end{equation}}
\def\ba#1\ea{\begin{align}#1\end{align}}

\usepackage{subcaption}

\usepackage[usenames,dvipsnames]{xcolor}

\usepackage[utf8]{inputenc}
\usepackage{jheppub}
\usepackage{amsmath}
\usepackage{multicol}
\usepackage{bbm}
\usepackage{enumerate}
\usepackage[inline]{enumitem}
\usepackage{bibentry}
\usepackage{comment}
\usepackage{amsthm}
\usepackage{mathrsfs}
\usepackage{upgreek}
\usepackage{amssymb}
\usepackage{bm}
\usepackage{setspace}
\usepackage{array,multirow,arydshln}
\usepackage{bigdelim}
\usepackage{scalerel}
\usepackage{diagbox}
\usepackage{caption}
\usepackage{tabularx}
\usepackage{empheq}

\usepackage{tikz}
\usepackage[compat=1.1.0]{tikz-feynman}

\usetikzlibrary{shapes.geometric,arrows,arrows.meta,decorations.pathmorphing,decorations.markings,patterns}

\newcommand*{\halfway}{0.5*\pgfdecoratedpathlength+4pt}
\tikzset{vertex/.style={inner sep=0,minimum size=3pt,circle,fill}}

\def\<{\langle}
\def\>{\rangle}

\newtheorem{lemma}{Lemma}

\preprint{UUITP-27/19, \ LCTP-20-19}

\title{Open associahedra and scattering forms}

\author[a]{Aidan Herderschee}
\emailAdd{aidanh@umich.edu}

\author[b]{and Fei Teng}
\emailAdd{fei.teng@physics.uu.se}

\affiliation[a]{Leinweber Center for Theoretical Physics, \\
Randall Laboratory of Physics, Department of Physics, \\
University of Michigan, Ann Arbor, MI 48109, USA}

\affiliation[b]{Department of Physics and Astronomy, Uppsala University, 75108 Uppsala, Sweden}

\date{\today}

\abstract{We continue the study of open associahedra associated with bi-color scattering amplitudes initiated in Ref.~\cite{Herderschee:2019wtl}. We focus on the facet geometries of the open associahedra, uncovering many new phenomena such as fiber-product geometries. We then provide novel recursion procedures for calculating the canonical form of open associahedra, generalizing recursion relations for bounded polytopes to unbounded polytopes. }

\begin{document}

\maketitle
\flushbottom

\section{Introduction}

Scattering amplitudes are one of the most fundamental observables in modern high energy physics. They have applications in a wide variety of fields, ranging from collider physics to gravity waves. Beyond their physical applications, the surprising simplicity of scattering amplitudes in certain theories, such as $\mathcal{N}=4$ super-Yang Mills, offers hints of hidden structures. It is important we understand why such simplifications occur and whether these hidden structures generalize for amplitudes in more realistic theories. With this broader objective in mind, this paper focuses on further developing the connection between scattering amplitudes and positive geometry \cite{Arkani-Hamed:2017tmz}. This program has been largely restricted to theories with adjoint states, such as bi-adjoint scalar theories and Yang-Mills. We instead focus on studying the positive geometry of scalar amplitudes with (anti-)fundamental states in addition to adjoint states. \\
\indent The first concrete connection between geometry and perturbative amplitudes was the identification of super-amplitudes in $\mathcal{N}=4$ super Yang-Mills (SYM) with the volumes of polytopes in the complex projective space $\mathbb{CP}^{4}$ \cite{Hodges:2009hk}. This connection was quickly generalized to all tree level amplitudes and loop integrands in $\mathcal{N}=4$ SYM by identifying their dual geometry, the amplituhedron \cite{Arkani-Hamed:2013jha}. Since the initial discovery of the amplituhedron, the connection between geometry and amplitudes has been refined and generalized. Amplitudes in a wide variety of theories have been interpreted as canonical forms on the space of external kinematic data~\cite{Arkani-Hamed:2017vfh, Arkani-Hamed:2017mur,Banerjee:2018tun,Aneesh:2019ddi,Raman:2019utu, Kalyanapuram:2020vil}. These canonical forms are volume forms of a polytope (or more generally a positive geometry)  on a subspace of external kinematic data~\cite{Arkani-Hamed:2017tmz}. Recently, novel recursive relations have been proposed for efficient computation of various canonical forms~\cite{He:2018svj, Salvatori:2019phs, Yang:2019esm, Kojima:2020tox, John:2020jww}. Like the original BCFW recursion~\cite{Britto:2005fq}, they involve complex shifts of kinematic data, and correspond to different triangulations of the positive geometry.
In addition, the canonical form is intimately related to color-kinematics duality, with differentials on kinematic space being dual to color factors \cite{Arkani-Hamed:2017mur}. Since the initiation of this program, the positive geometry dual of a wide class of observables have been found~\cite{Eden:2017fow, Salvatori:2018aha, Arkani-Hamed:2017fdk, Arkani-Hamed:2019vag}.  \\
\indent In Ref.~\cite{Herderschee:2019wtl}, the authors, together with He and Zhang, initiated a systematic study into the positive geometry of bi-color scalar theory \cite{Naculich_2014}, which is an (anti-)fundamental extension of the bi-adjoint scalar theory. It involves states that transform in the (anti-)fundamental, not just adjoint, representation. The bi-color theory is intimately related to the double copy procedure of theories with (anti-)fundamental states~\cite{Brown_2016,Brown:2018wss,Johansson:2019dnu}. Motivated by the application of an ``inverse soft construction ansatz'' previously used to construct subspaces associated with Cayley polytopes \cite{He:2018pue}, an inverse soft construction ansatz was used to recursively construct the positive geometry of the $n$-point bi-color amplitude in terms of the lower point ones~\cite{Herderschee:2019wtl}. The bi-color amplitudes were shown to correspond to a class of unbounded geometries called \emph{open associahedra}. Due to its simplicity as a purely scalar theory and importance in the (anti-)fundamental double copy, finding the positive geometry of the bi-color theory signals the wide applicability of the positive geometry picture. Many novel features appear in the bi-color positive geometry. For example, the positive geometries of the bi-color amplitudes are not bounded, which manifests as the associated canonical forms not being projective in kinematic space. Furthermore, the facet geometries of the open associahedra are not simply direct product geometries, but fiber-product geometries.\footnote{In Ref.~\cite{Herderschee:2019wtl} they are called semi-direct product geometries. In this work, we reserve the name for a (possibly) more general class of geometries.} The analysis of Ref.~\cite{Herderschee:2019wtl} was just an initial step in this program, while many open questions are yet to be answered. \\
\indent In this note, we give a more in-depth analysis of the open associahedra. We begin by studying the facet structure of our open associahedra. In the case of bi-adjoint theory, the geometry of any facet of closed associahedra is simply the direct product of two lower degree associahedra,

\begin{equation}
\mathcal{A}\big|_{\textrm{arbitrary facet}}=\mathcal{A}_{L}\times \mathcal{A}_{R}\,.    
\end{equation}

\noindent This is the geometrical manifestation of locality: how tree-level amplitudes must factorize on poles. In the case of open associahedra associated with bi-color amplitudes, we instead find locality manifests as a fiber of lower point geometries, which generalizes the product geometries of facets in Ref.~\cite{Arkani-Hamed:2017mur}. Once we analyze the facet geometries themselves, we find an extended equivalence class of positive geometries which yield the canonical form with the same functional dependence on facet variables. A small subset of this extended class can be derived using a minor generalization of the inverse soft construction. Using our analysis of the factorization channels, we propose a set of novel recursion relations for bi-color amplitudes in section~\ref{recursionsection}. These recursion relations are the natural generalization of the recursion procedures for bounded polytopes to unbounded polytopes \cite{Salvatori:2019phs,Yang:2019esm}. 
Using our analysis of the factorization channels, a complete proof of the pullback conjecture of Ref.~\cite{Herderschee:2019wtl} is provided for bi-color amplitudes. The pullback conjecture is key for proving color-kinematics duality of theories with (anti-)fundamental states. We conclude with a short summary and list some future directions. The organization of the paper is as follows:

\begin{itemize}
    \item Section \ref{sec:invsoftconstruction} reviews necessary backgrounds and the construction of open associahedra structure discussed in Ref.~\cite{Herderschee:2019wtl}. 
    \item Section \ref{sec:propfactorchannelss} discusses new geometrical structures that appear on the facets of open associahedra. 
    \item Section \ref{recursionsection} gives a succinct review of recursion procedures for closed polytopes before deriving a new recursion relation for open associahedra which is applicable to bi-color amplitudes.
    \item Section \ref{sec:proofpullback} proves the pullback conjecture of Ref.~\cite{Herderschee:2019wtl} using a systematic study of certain facet geometries in appendix \ref{sec:factorizationchannel}. 
    \item Section \ref{sec:conclusion} concludes with a short summary and some open questions.
    \item A succinct review canonical forms is given in appendix~\ref{sec:reviewcanonicaofrm}.
\end{itemize}

\paragraph{Notation:} Throughout the paper, we abbreviate ``fundamental" to $\mathfrak{f}$, ``anti-fundamental'' to $\mathfrak{af}$, and ``adjoint" to $\mathfrak{adj}$ for convenience. We use normal font letters to denote subspaces and the corresponding boldface letters to denote the constraints. For example, $P=Q\cap R$ is the intersection of the subspaces $Q$ and $R$, and it is given by the constraints $\pmb{P}=\pmb{Q}\cup\pmb{R}$.

\section{ABHY construction for bi-color scalar amplitudes}
\label{sec:invsoftconstruction}

In this section, we review the ABHY construction for the open associahedra associated to bi-color scalar amplitudes. The crucial objects of study in bi-color theory are \textit{double partial} amplitudes, $m_n[\alpha|\beta]$, a sum of all Feynman diagrams consistent with both $\alpha$ and $\beta$ orderings. $m_n[\alpha|\beta]$ are extremely important for double copy computations and intimately related to underlying geometry of the string worldsheet \cite{Mizera:2017cqs, Arkani-Hamed:2017mur, Frost:2019fjn}.  For example, if we interpret $m_n[\alpha|\beta]$ as a matrix of amplitudes and isolate a sub-matrix of full rank, the inverse of this sub-matrix defines field theory KLT relations \cite{Bern:2010ue,Brown:2018wss}. \\
\indent An interesting framework for understanding the underlying structure of $m_n[\alpha|\beta]$ is positive geometry. In both the bi-color and bi-adjoint case, $m_n[\alpha|\beta]$ can be interpreted as the canonical rational function of a polytope in kinematic space. This is very non-trivial as this property cannot hold for arbitrary collections of Feynman diagrams. Even more miraculous than the fact that the bi-color is dual to a geometry is that the geometry can be constructed recursively.  Although the resulting recursion is difficult to state in closed form for generic orderings, it follows from a simple ansatz, the inverse soft construction \cite{Herderschee:2019wtl}. 

\subsection{Bi-color scalar theory}
The matter content of the bi-color scalar theory is a real scalar $\phi^{aA}$ that transforms in the adjoint representation of $U(N)\times U(N')$, and a family of $n_f$ complex scalars $(\varphi_r)_{\underline{iI}}$ that transform in the fundamental representation of $U(N)\times U(N')$. The Lagrangian is\footnote{We use $[T^a,T^b]=\tilde{f}^{abc}T^c$ and $\textrm{tr}(T^aT^b)=\delta^{ab}$ as our Lie algebra normalization.}
\begin{align}\label{eq:bicolor}
\mathcal{L}_{\phi^3}=&\;\frac{1}{2}\partial_\mu\phi^{aA}\partial^{\mu}\phi^{aA}+\frac{\lambda}{6}\tilde{f}^{abc}\tilde{f}^{ABC}\phi^{aA}\phi^{bB}\phi^{cC}\nonumber\\
&+\sum_{r=1}^{n_f}\left[\partial_\mu(\varphi_r)_{\underline{iI}}\partial^\mu(\varphi^*_r)^{\overline{iI}}+\lambda\phi^{aA}(\varphi^*_r)^{\overline{iI}}(T^a)_{\underline{i}}^{\overline{j}}(T^A)_{\underline{I}}^{\overline{J}}(\varphi_r)_{\underline{jJ}}\right],
\end{align}
where $\lambda$ is an arbitrary coupling constant. There is a global flavor symmetry acting on the flavor index $r$. We further assume that all the $\mathfrak{f}$-$\mathfrak{af}$ scalar pairs have distinct flavor assignments in our amplitudes. To determine which factorization channels are allowed by flavor symmetry, we define a non-numeric flavor symbol to each particle, $f(\varphi_r)=f_r$, $f(\varphi^*_r)=-f_r$ and $f(\phi)=0$, such that a factorization channel $s_I$ is allowed by flavor symmetry if and only if $\vartheta_I=1$~\cite{Johansson:2019dnu}, where
\begin{equation}\label{eq:thetaI}
\vartheta_{I}=\left\{\begin{matrix}
1 & \quad & \textrm{ if }\sum_{s\in I} f(s)=0\textrm{ or a single term }\pm f_r\\ 
0 & \quad & \textrm{ otherwise }
\end{matrix}\right. \,.
\end{equation}
For instance, if we have two $\mathfrak{f}$-$\mathfrak{af}$ pairs $\{\underline{1},\overline{2}\}$ and $\{\underline{3},\overline{4}\}$, the propagator $1/s_{1,2}$ is allowed but $1/s_{2,3}$ is not. The amplitude with only a single flavor can be reproduced by averaging over all possible flavor assignments.

Decomposition of the global color factors leads to a matrix of color-ordered (double partial) amplitudes $m_n[\alpha|\beta]$. In this work, we will mainly focus on the diagonal components $A_n[\alpha]\equiv m_n[\alpha|\alpha]$, which consists of the planar Feynman diagrams under the ordering $\alpha$ that respect the flavor conservation. Next, we review the minimal basis for $A_n[\alpha]$ under the color decomposition.

\subsection{Melia basis}

The minimal basis for the color ordering $\alpha$ is the Melia basis~\cite{Melia:2013bta,Melia:2013epa,Johansson:2015oia}. If we represent each $\mathfrak{f}$-$\mathfrak{af}$ pair by a pair of parentheses, $\underline{i}\rightarrow (i$ and $\overline{j}\rightarrow j)$, the Melia basis is then given by all the valid ways of adding the parentheses to a word in which the position of a specific $\mathfrak{f}$-$\mathfrak{af}$ pair, say $(1,2)$, is fixed.\footnote{More specifically, a left parenthesis must first be closed by the right parenthesis with the same flavor before another right parenthesis can appear. Such arrangements of parentheses are also called Dyck words.} For example, the five-point Melia basis for one $\mathfrak{adj}$ particle $5$ and two $\mathfrak{f}$-$\mathfrak{af}$ pairs $(1,2)$ and $(3,4)$ consists of the color-ordered amplitudes
\begin{align}
    \Big\{A_5[(1,2),(3,4),5]\,,A_5[(1,2),5,(3,4)]\,,A_5[(1,2),(3,5,4)]\Big\}\,.
\end{align}

However, factorizations of Melia-basis amplitudes will not automatically land on lower-point Melia bases in general. For example, the $s_{1,2}\rightarrow 0$ channel of $A_5[(1,2),(3,5,4)]$ gives
\begin{align}
    A_5[(1,2),(3,5,4)]\rightarrow A_3[(1,2),q]\times\frac{1}{q^2}\times A_4[-q,(3,5,4)]\,,
\end{align}
where $A_4[-q,(3,5,4)]$ is not in the four-point Melia basis automatically. One can of course settle it into the Melia basis using certain amplitude relations, but we find it more convenient to consider more generic color orderings when studying factorizations. To start with, we define a \emph{block} to be the structure enclosed by an $\mathfrak{f}$-$\mathfrak{af}$ pair,
\begin{align}
    \mathsf{B}_{i}\equiv(l_i,\mathsf{B}_{i_1},\mathsf{B}_{i_2},\ldots,\mathsf{B}_{i_s},r_i)\,,
\end{align}
where $\mathsf{B}_{i_\ell}\in\text{sub}[\mathsf{B}_i]\equiv\{\mathsf{B}_{i_1},\mathsf{B}_{i_2},\ldots,\mathsf{B}_{i_s}\}$ are sub-blocks defined recursively. The simplest blocks, which terminate the recursive definition, are either a single $\mathfrak{adj}$ particle, $\mathsf{B}_i=g_i$, or a single $\mathfrak{f}$-$\mathfrak{af}$ pair, $\mathsf{B}_i=(l_i,r_i)$. In this work, we are interested in the color orderings of the following form,
\begin{align}\label{eq:genericOrder}
\alpha=[\mathsf{B}_1,\mathsf{B}_2,\mathsf{B}_3,\ldots,\mathsf{B}_m]\,,\text{ where either }\mathsf{B}_1=(l_1,r_1)\text{ or }\mathsf{B}_1=g_1\,,
\end{align}
The Melia basis corresponds to fix $\mathsf{B}_1=(1,2)$. In fact, all the color-ordered amplitudes for bi-color scalars can be represented by these $\alpha$'s after flipping certain parentheses and using cyclicity, since the trivial kinematic numerators are insensitive to particles and anti-particles.

\subsection{The subspace construction}\label{subspaceconstructionssss}

The kinematic space $\mathcal{K}_n$ for $n$ massless particles is spanned by $n(n-3)/2$ planar Mandelstam variables $X_{i,j}=(p_i+\ldots+p_{j-1})^2$. The positive geometry that corresponds to the bi-adjoint amplitudes is the ABHY associahedron $\mathcal{A}_n$~\cite{Arkani-Hamed:2017mur}. It is obtained by intersecting the positive cone of $\mathcal{K}_n$ with a subspace $H_n$. The planar amplitude is given by the canonical form of $\mathcal{A}_n$. It can be constructed by the pullback of the planar scattering form from $\mathcal{K}_n$ to $\mathcal{A}_n$. At $n=4$, the associaheron is the intersection of the positive cone $\{s>0,t>0\}$ and the subspace $s+t=c>0$, which is a line segment. The canonical form is obtained by the pullback $\left(\frac{ds}{s}-\frac{dt}{t}\right)\big|_{s+t=c}=\left(\frac{1}{s}+\frac{1}{t}\right)ds$, which also gives the amplitude. We refer the readers to Ref.~\cite{Arkani-Hamed:2017mur} for more details.

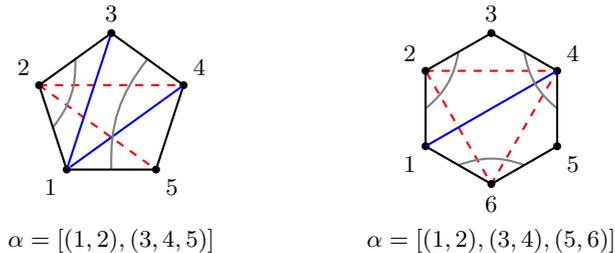
\begin{figure}[t]
\centering
\begin{tikzpicture}[every node/.style={font=\footnotesize,},dir/.style={decoration={markings, mark=at position \halfway with {\arrow{Latex}}},postaction={decorate}}]
\coordinate (p2) at (162:1);
\coordinate (p4) at (18:1);
\coordinate (p5) at (-54:1);
\draw [thick,red,dashed] (p2) -- (p5) (p2) -- (p4);
\node (p1) at (-126:1) [vertex,label={[label distance=-2pt]-126:{$1$}}] {};
\node at (p2) [vertex,label={[label distance=-2pt]162:{$2$}}] {};
\node (p3) at (0,1)  [vertex,label={[label distance=-1pt]90:{$3$}}] {};
\node at (p4) [vertex,label={[label distance=-2pt]18:{$4$}}] {};
\node at (p5) [vertex,label={[label distance=-2pt]-54:{$5$}}] {};
\draw [thick,blue] (p1.center) -- (p3.center) (p1) -- (p4.center);
\path (p1.center) -- (p2.center) node (f1) [pos=0.5] {};
\path (p2.center) -- (p3.center) node (f2) [pos=0.5] {};
\path (p3.center) -- (p4.center) node (f3) [pos=0.5] {};
\path (p5.center) -- (p1.center) node (f5) [pos=0.5] {};
\draw [thick,gray] (f1.center) to [bend right=20] (f2.center) (f3.center) to [bend right=20] (f5.center);
\draw [thick] (p1.center) -- (p2.center) -- (p3.center) -- (p4.center) -- (p5.center) -- cycle;
\node at (0,-1.75) [] {$\alpha=[(1,2),(3,4,5)]$};
\begin{scope}[xshift=5cm]
\coordinate (p1) at (-150:1);
\coordinate (p2) at (150:1);
\coordinate (p3) at (0,1);
\coordinate (p4) at (30:1);
\coordinate (p5) at (-30:1);
\coordinate (p6) at (0,-1);
\draw [red,thick,dashed] (p2) -- (p4) -- (p6) -- cycle;
\node at (p1)  [vertex,label={[label distance=-2pt]-150:{$1$}}] {};
\node at (p2)  [vertex,label={[label distance=-2pt]150:{$2$}}] {};
\node at (p3) [vertex,label={[label distance=-1pt]90:{$3$}}] {};
\node at (p4) [vertex,label={[label distance=-2pt]30:{$4$}}] {};
\node at (p5) [vertex,label={[label distance=-2pt]-30:{$5$}}] {};
\node at (p6) [vertex,label={[label distance=-1pt]-90:{$6$}}] {};
\draw [blue,thick] (p1.center) -- (p4.center);
\path (p1.center) -- (p2.center) node (f1) [pos=0.5] {};
\path (p2.center) -- (p3.center) node (f2) [pos=0.5] {};
\path (p3.center) -- (p4.center) node (f3) [pos=0.5] {};
\path (p4.center) -- (p5.center) node (f4) [pos=0.5] {};
\path (p5.center) -- (p6.center) node (f5) [pos=0.5] {};
\path (p6.center) -- (p1.center) node (f6) [pos=0.5] {};
\draw [thick,gray] (f1.center) to [bend right=20] (f2.center) (f3.center) to [bend right=20] (f4.center) (f5.center) to [bend right=20] (f6.center);
\draw [thick] (p1) -- (p2) -- (p3) -- (p4) -- (p5) -- (p6) -- cycle;
\node at (0,-1.75) [] {$\alpha=[(1,2),(3,4),(5,6)]$};
\end{scope}
\end{tikzpicture}
\caption{Example of flavor lines taken from Ref.~\cite{Herderschee:2019wtl}. The gray curves are the flavor lines associated with the $\mathfrak{f}$-$\mathfrak{af}$ pairs. The red dashed lines are the planar variables forbidden by flavor symmetry. The blue solid lines are a few representatives of allowed planar variables.}
\label{fig:forbiddenPoles}
\end{figure}

Together with He and Zhang, we have shown in Ref.~\cite{Herderschee:2019wtl} that the positive geometry for the bi-color scalar amplitude $A_n[\alpha]$ is an $(n{-}3)$-dimensional open associahedron $\mathcal{A}_n[\alpha]=\Delta_n[\alpha]\cap H_n[\alpha]$, given by the intersection of a positive cone $\Delta_n[\alpha]$ and a subspace $H_n[\alpha]$, just as the ABHY associahedron for bi-adjoint scalar amplitudes~\cite{Arkani-Hamed:2017mur}. We denote the set of constraints that give $\Delta_n$ and $H_n$ as $\pmb{\Delta}_n$ and $\pmb{H}_n$. Starting from the full kinematic space $\mathcal{K}_n$ spanned by $n(n-3)/2$ planar Mandelstam variables, we first remove those forbidden by the flavor conservation, which can be done by setting
\begin{align}
    \pmb{F}_n[\alpha]=\Big\{X_{\alpha(i),\alpha(j)}=b_{\alpha(i),\alpha(j)}>0\text{ if }\vartheta_{\alpha(i),\ldots,\alpha(j-1)}=0 \Big\}.
\end{align}
For $k$ pairs of $\mathfrak{f}$-$\mathfrak{af}$ particles, these constraints restrict the physical kinematic space to $\mathcal{K}^k_n[\alpha]$ with the dimension\footnote{The first equality is reached only when $n=4$ and $k=2$, while the second one is reached when all the $\mathfrak{f}$ particles are adjacent to their $\mathfrak{af}$ partners.}
\begin{align}
    n-3\leqslant\textrm{dim}\,\mathcal{K}_n^k[\alpha]\leqslant\frac{n(n-3)}{2}-\frac{k(k-1)}{2}\,.
\end{align}
The cone $\Delta_n[\alpha]$ is just the positive region of $\mathcal{K}_n^k[\alpha]$, obtained by imposing the constraints $\pmb{\Delta}_n[\alpha]=\pmb{F}_n[\alpha]\cup\pmb{P}_n[\alpha]$, where
\begin{align}
    \pmb{P}_n[\alpha]=\left\{X_{\alpha(i),\alpha(j)}\geqslant 0\text{ for all }1\leqslant i<j\leqslant n\right\}
\end{align}
imposes the positivity. A nice way to visualize which factorization channels are allowed by flavor symmetry is to use \emph{flavor lines} in the polygon dual to Feynman diagrams. If $\big(\alpha(i),\alpha(j)\big)$ is an $\mathfrak{f}$-$\mathfrak{af}$ pair, its flavor line connects the edge $E_{\alpha(i),\alpha(i+1)}$ and $E_{\alpha(j),\alpha(j+1)}$ in the dual polygon. \emph{The planar variables allowed by flavor symmetry can only cross at most one flavor line.} Two examples are given in figure~\ref{fig:forbiddenPoles}. \\

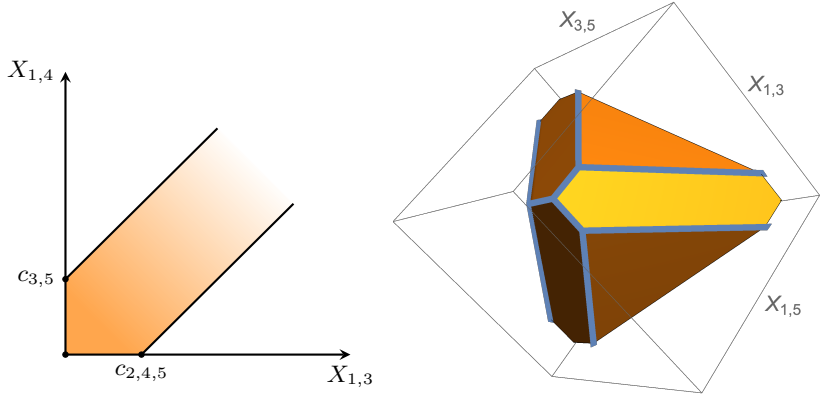
\begin{figure}[t]
\centering
  \begin{tikzpicture}[every node/.style={font=\footnotesize}]
	\shade [left color=orange!70!white,right color=white,shading angle=135] (1,0) ++(2,2) -- (1,0) -- (0,0) -- (0,1) -- ++(2,2);
	\draw [thick,-stealth] (0,0) -- (3.75,0) node[below=0pt]{$X_{1,3}$};
	\draw [thick,-stealth] (0,0) -- (0,3.75) node[left=0pt]{$X_{1,4}$};
	\draw [thick] (1,0) -- ++(2,2) (0,1) -- ++(2,2);
	\filldraw (0,0) circle (1pt) (1,0) circle (1pt) node[below=0pt]{$c_{2,4,5}$} (0,1) circle (1pt);
	\node at (0,1) [left=0pt]{$c_{3,5}$};
	\end{tikzpicture}
  \includegraphics[scale=0.4]{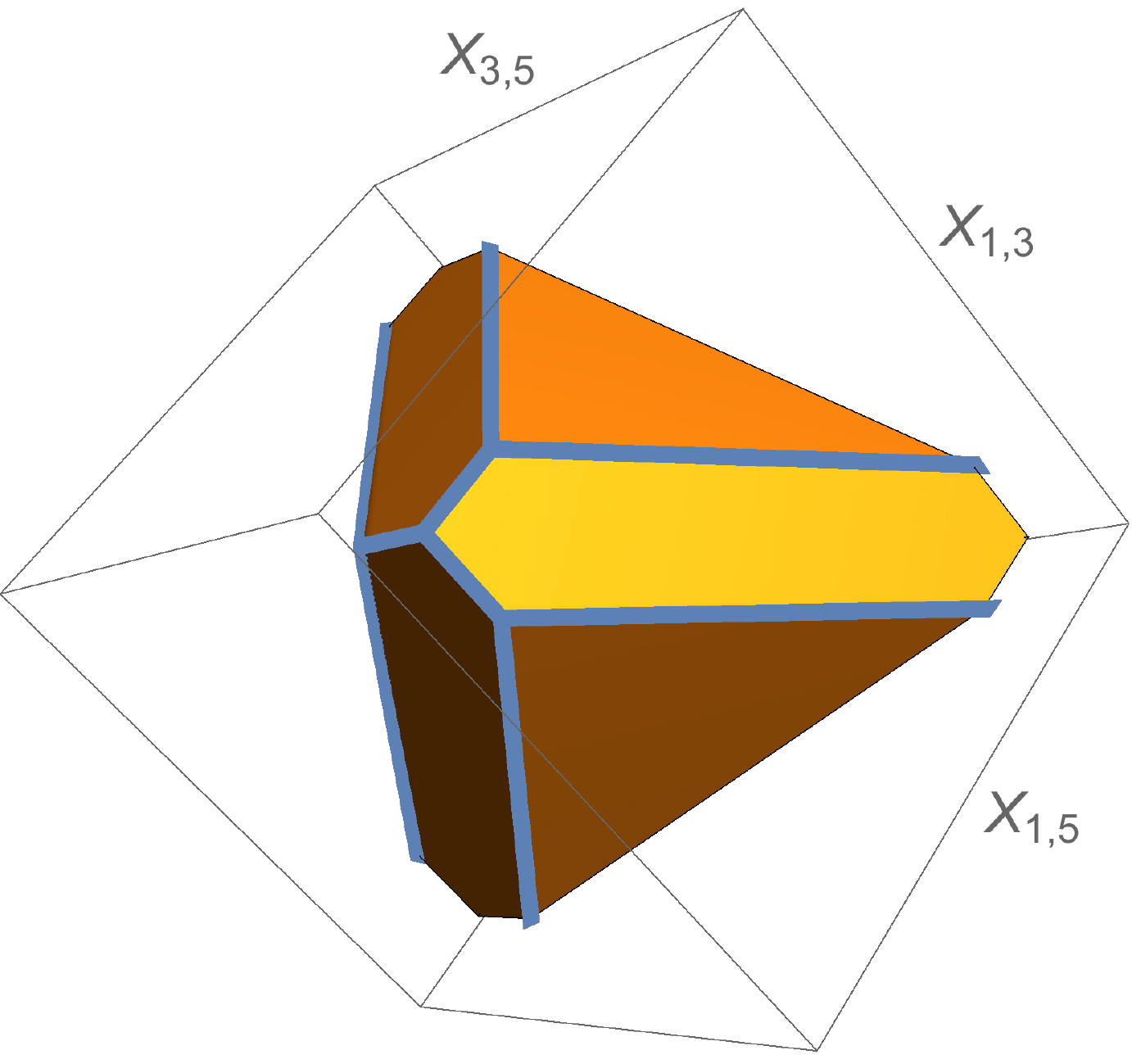}
  \caption{The left and right geometries are the visualizations of $\mathcal{A}_5[(1,2),3,(4,5)]$ and $\mathcal{A}_6[(1,2),(3,4),(5,6)]$ respectively,  taken from Ref.~\cite{Herderschee:2019wtl}. Both geometries are unbounded, and  $\mathcal{A}_5[(1,2),3,(4,5)]$ appears as a facet of $\mathcal{A}_6[(1,2),(3,4),(5,6)]$.} 
  \label{fig:6pointexample}
\end{figure}
\indent The open associahedron $\mathcal{A}_n[\alpha]$ is obtained by intersecting the cone $\Delta_n[\alpha]$ with another subspace $H_n[\alpha]$ that is given by $\textrm{dim}\,\mathcal{K}_n^k[\alpha]-(n{-}3)$ constraints $\pmb{H}_n[\alpha]$, 
\begin{align}
    \mathcal{A}_n[\alpha]=H_n[\alpha]\cap\Delta_n[\alpha]\,.
\end{align}
A five-point and six-point example for such geometries are provided in figure~\ref{fig:6pointexample}. There exists an recursive construction for $\pmb{H}_n[\alpha]$ through an inverse soft/factorization analysis~\cite{Herderschee:2019wtl}. Suppose we already know the constraints for a lower-point ordering $\beta=[\mathsf{B}_1,\ldots,\mathsf{B}_{m-1}]$, the constraints $\pmb{H}_n[\alpha]=\pmb{H}_n[\beta,\mathsf{B}_m]$ can be obtained by the following procedures:
\paragraph{Case 1:}
If the last block $\mathsf{B}_m=n$ is a single $\mathfrak{adj}$ particle, we have
\begin{align}\label{eq:rec_n}
    \pmb{H}_n[\alpha]=\pmb{H}_n[\beta,n]=\pmb{H}_{n-1}[\beta]\cup\pmb{C}_1[\beta,n]\,,
\end{align}
where the set $\pmb{C}_1$ is given by
\begin{align}\label{eq:C1_n}
    \pmb{C}_1[\beta,q]=\left.\bigcup_{i=1}^{m-1}\bigcup_{\mathsf{I}\in\text{sub}[\mathsf{B}_i]}
    \left\{\begin{array}{c}-s_{q,l_i}=c_{q,l_i}>0 \\
    -s_{q,r_i}=c_{q,r_i}>0 \\
    -s_{q,\mathsf{I}}=c_{q,\mathsf{I}}>0
    \end{array}
    \right\}\middle\backslash\{s_{q,l_1},s_{q,r_{m-1}}\}\right..
\end{align}
Here $q=p_n$ is on-shell. However, this equation also applies to off-shell $\mathfrak{adj}$ particles.\footnote{In a color ordering $\alpha$, we label off-shell particles (if they exist) by their momenta.}
\paragraph{Case 2:}
If the last block $\mathsf{B}_m=(l_m,r_m)$ is a single adjacent $\mathfrak{f}$-$\mathfrak{af}$ pair, we have
\begin{align}\label{eq:rec_nbif}
        \pmb{H}_n[\alpha]=\pmb{H}_n[\beta,(l_m,r_m)]&=\pmb{H}_{n-1}[\beta,P_{l_m,r_m}]\cup\pmb{C}_2[\beta,(l_m,r_m)]\nonumber\\
        &=\pmb{H}_{n-2}[\beta]\cup\pmb{C}_1[\beta,P_{l_m,r_m}]\cup\pmb{C}_2[\beta,(l_m,r_m)]
\end{align}
where $P_{l_m,r_m}=p_{l_m}+p_{r_m}$, and $\pmb{H}_{n-1}[\beta,P_{l_m,r_m}]$ can be further recursively constructed using Eq.~\eqref{eq:rec_n}. The set $\pmb{C}_2$ is given by
\begin{align}\label{eq:C2_n}
    \pmb{C}_2[\beta,(l_m,r_m)]=\bigcup_{i=2}^{m-1}\{-s_{\mathsf{B}_i,r_m}=c_{\mathsf{B}_i,r_m}>0\}\,.
\end{align} 
Here, $P_{l_m,r_m}$ is an off-shell $\mathfrak{adj}$ particle, which is exchanged between the block $(l_m,r_m)$ and the other particles.
\paragraph{Case 3:}
Finally, if $\mathsf{B}_m$ is a generic block with substructures, we have
\begin{align}\label{eq:Hngeneric}
    \pmb{H}_n[\alpha]=\pmb{H}_n[\beta,\mathsf{B}_m]=&\;\pmb{H}_{n-|\mathsf{B}_m|+2}[\beta,(L_m,r_m)]\cup\pmb{C}_3[\beta,\mathsf{B}_m]\nonumber\\
    &\cup\pmb{H}_{|\mathsf{B}_m|}[(-L_m,l_m),\mathsf{B}_{m_1},\ldots,\mathsf{B}_{m_s}]\nonumber\\
    =&\;\pmb{H}_{n-|\mathsf{B}_m|}[\beta]\cup\pmb{C}_1[\beta,P_{L_m,r_m}]\cup\pmb{C}_2[\beta,(L_m,r_m)]\cup\pmb{C}_3[\beta,\mathsf{B}_m]\nonumber\\
    &\cup\pmb{H}_{|\mathsf{B}_m|}[(-L_m,l_m),\mathsf{B}_{m_1},\ldots,\mathsf{B}_{m_s}]
\end{align}
where $L_m$ is the total momentum of $\mathsf{B}_m$ except for $r_m$, and $\mathsf{B}_{m_1},\ldots,\mathsf{B}_{m_s}$ are the sub-blocks contained in $\mathsf{B}_m$. The set $\pmb{C}_3$ has two parts,
\begin{align}\label{eq:ra2}
    \pmb{C}_3[\beta,\mathsf{B}_m]=\pmb{C}_3^a[\mathsf{B}_m]\cup\pmb{C}_3^b[\beta,\mathsf{B}_m]\,,
\end{align}
which are given by
\begin{gather}\label{eq:C3_n}
    \pmb{C}_3^a[\mathsf{B}_m]=\{-s_{l_m,r_m}=c_{l_m,r_m}>0\}\bigcup_{\mathsf{I}\in\text{sub}[\mathsf{B}_m]\backslash\mathsf{B}_{m_s}}\left\{-s_{\mathsf{I},r_m}=c_{\mathsf{I},r_m}>0\right\},\nonumber\\
    \pmb{C}_3^b[\beta,\mathsf{B}_m]=\bigcup_{i=2}^{m-1}\bigcup_{\mathsf{I}\in\text{sub}[\mathsf{B}_m]}\left\{
    -s_{\mathsf{B}_i,\mathsf{I}}=c_{\mathsf{B}_i,\mathsf{I}}>0\right\}\,.
\end{gather}
On the other hand, the set $\pmb{H}_{n-|\mathsf{B}_m|+2}$ and $\pmb{H}_{|\mathsf{B}_m|}$ can be obtained recursively.
We note that this recursion will terminate at $\pmb{H}_3=\emptyset$. \\
\indent As an example, we will now recursively construct $\pmb{H}_{6}[(1,2),(3,4),(5,6)]$. Since the last block is a $(\mathfrak{a})\mathfrak{f}$ pair, we apply eqs. (\ref{eq:rec_nbif}) and (\ref{eq:C2_n}):
\begin{equation}
\begin{split}
\pmb{H}_{6}[(1,2),(3,4),(5,6)]&=\pmb{H}_{5}[(1,2),(3,4),P_{5,6}]\cup \pmb{C}_{2}[(1,2),(3,4),(5,6)] \ ,\\
&=\pmb{H}_{5}[(1,2),(3,4),P_{5,6}]\cup \{ c_{3,4,6}=-s_{3,4,6}\} \ . \\
\end{split}
\label{eq:firstlinederiv}
\end{equation}
\noindent We now apply the same procedure to $\pmb{H}_{5}[(1,2),(3,4),P_{5,6}]$. Since $P_{5,6}$ is an $\mathfrak{adj}$ state, we instead apply eqs. (\ref{eq:rec_n}) and (\ref{eq:C1_n}):
\begin{equation}
\begin{split}
\pmb{H}_{5}[(1,2),(3,4),P_{5,6}]&= \pmb{H}_{4}[(1,2),(3,4)]\cup \pmb{C}_{1}[(1,2),(3,4),P_{5,6}] \ , \\
&=\pmb{H}_{4}[(1,2),(3,4)] \cup \{ c_{2,5,6}=-s_{2,5,6}, \ c_{3,5,6}=-s_{3,5,6} \}\ 
\end{split}  
\label{eq:secondlinederiv}
\end{equation}
\noindent Finally, we apply the procedure to $\pmb{H}_4[(1,2),(3,4)]$, finding 
\begin{equation}
\begin{split}
\pmb{H}_{4}[(1,2),(3,4)]&=\pmb{H}_{4}[(1,2),P_{3,4}]\cup \pmb{C}_{2}[(1,2),(3,4)] \\
&=\pmb{H}_{4}[(1,2),P_{3,4}]=\emptyset\,. \\
\end{split} 
\label{eq:thirdlinederiv}
\end{equation}
\noindent Combining eqs. (\ref{eq:firstlinederiv}), (\ref{eq:secondlinederiv}), and (\ref{eq:thirdlinederiv}), the restriction equations for $\mathcal{A}_6[(1,2),(3,4),(5,6)]$ take the form 
\begin{equation}
\pmb{H}_{6}[(1,2),(3,4),(5,6)]=\{ \ c_{3,4,6}=-s_{3,4,6} ,\ c_{2,5,6}=-s_{2,5,6}, \ c_{3,5,6}=-s_{3,5,6} \ \}  \ .
\label{eq:restirctionsexam}
\end{equation}
\noindent In terms of planar variables, $X_{i,j}>0$, the restrictions in eq. (\ref{eq:restirctionsexam}) can be written as 
\begin{equation}
\begin{split}
c_{3,4,6}&=X_{3,6}+X_{1,5}-X_{1,3}-X_{3,5}\, ,\\
c_{2,5,6}&=X_{2,5}+X_{1,3}-X_{3,6}-X_{1,5}\, , \\
c_{3,5,6}&=X_{3,5}+X_{1,4}-X_{1,5}-X_{1,3}\, . \\
\end{split}   
\label{allrestrictionss}
\end{equation}
\noindent Solving for all $X_{i,j}$ in terms of $X_{3,5}$, $X_{1,3}$ and $X_{1,5}$ using eq. (\ref{allrestrictionss}) and imposing that all $X_{i,j}>0$, one finds the unbounded polytope in figure~\ref{fig:6pointexample}. One can show that the restriction equations for generic adjacent-pair configuration take the form
\begin{align}\label{eq:adjacentGeneric}
    \pmb{H}_n[(1,2),(3,4),\ldots,(n-1,n)]=\left\{
    \begin{array}{c}
    s_{j,2i-1,2i}=-c_{j,2i-1,2i} \\
    s_{2i-1,2i,2k}=-c_{2i-1,2i,2k}
    \end{array} \middle|
    \begin{array}{c}
    2\leqslant i\leqslant n/2 \\
    2\leqslant j\leqslant 2i-3 \\
    i+1\leqslant k\leqslant n/2
    \end{array}\right\}\,.
\end{align}

\indent As a further consistency check, we show that the number of constraints in $\pmb{H}_n[\alpha]$ is indeed $\textrm{dim}\,\mathcal{K}_n^k[\alpha]-(n{-}3)$. However, while both $\textrm{dim}\,\mathcal{K}_n^k[\alpha]$ and $\pmb{H}_n[\alpha]$ have rather complicated dependence on the block structures, it is more convenient to define $\delta_i$ as the difference between the new planar variables and constraints introduced by the block $\mathsf{B}_i$ in the above recursive construction, and prove that
\begin{align}\label{eq:Nconstraints}
    \sum_{i=1}^{m}\delta_i=n-3\,.
\end{align}
According to eq.~\eqref{eq:Hngeneric}, the new constraints coming with a new block $\mathsf{B}_i$ are in $\pmb{C}_{1,2,3}$ and $\pmb{H}_{|\mathsf{B}_i|}$. We note that $\pmb{H}_{|\mathsf{B}_i|}$ are completely localized within the block $\mathsf{B}_i$. As the inductive assumption, the difference between the new planar variables and constraints localized with $\mathsf{B}_i$ is $|\mathsf{B}_i|-3$ for $|\mathsf{B}_i|\geqslant 3$ and zero otherwise. On the other hand, the constraints in $\pmb{C}_{1,2,3}$ are in one-to-one correspondence with the new planar variables that are incompatible with either $X_{l_1,l_i}$ or $X_{l_i,r_i}$~\cite{Herderschee:2019wtl}. In fact, now we are left with at most three new planar variables that are not covered by either $\pmb{C}_{1,2,3}$ or $\pmb{H}_{|\mathsf{B}_i|}$,
\begin{align}
    & X_{l_1,l_i}\,,  & & \text{present if } & & |\mathsf{B}_i|\geqslant 2\,, \nonumber\\
    & X_{l_1,r_{i-1}}\,,  & & \text{present if }& & i>2\,, \\
    & X_{l_i,r_i}\,,  & & \text{present if } & & |\mathsf{B}_i|\geqslant 3\,. \nonumber
\end{align}
For the Melia basis, in which $\mathsf{B}_1=(1,2)$ , a direct counting gives
\begin{align}\label{eq:deltai}
    \delta_i=\left\{\begin{array}{lcc}
    0 & \quad & i=1 \\
    |\mathsf{B}_i|-1 & \quad &i=2 \\
    |\mathsf{B}_i| & &i>2 
    \end{array}\right.\,.
\end{align}
This result holds for a generic block $\mathsf{B}_i$. 
Then using $\sum_{i=1}^{m}|\mathsf{B}_i|=n$, one can easily see that eq.~\eqref{eq:Nconstraints} holds. The derivation for $\mathsf{B}_1$ being an $\mathfrak{adj}$ particle only differs slightly by some technical details, which will not be repeated here.

\section{Factorization channels and facet geometries of open associahedra}
\label{sec:propfactorchannelss}

\indent For the ABHY associahedra, a facet is characterized by a planar variable $X$ reaching zero, which also defines a partition $L\cup R$ of the particles. In practice, one can show that certain constraints in $\pmb{H}_n$ are satisfied automatically and thus can be dropped when $X\rightarrow 0$ due to the fact that incompatible planar variables cannot reach zero simultaneously. On the other hand, the rest constraints correspond exactly to those of the left and right sub-polytope,
\begin{align}\label{eq:directH}
    \pmb{H}_n\xrightarrow{X=0}\pmb{H}^L\cup\pmb{H}^R\,,
\end{align}
namely, the constraints in $\pmb{H}^{L/R}$ are localized on the planar variables in $L/R$ respectively. Each facet is thus a direct product of two lower dimensional associahedra,
\begin{align}\label{eq:directA}
    \mathcal{A}\big|_{X=0}\cong \mathcal{A}^L\times \mathcal{A}^R\,.
\end{align}
The residue of canonical form factorizes on the facet, 
\begin{align}
    \text{Res}_{X=0}\Omega(\mathcal{A})=\Omega(\mathcal{A}^L)\wedge\Omega(\mathcal{A}^R)\,,
\end{align}
which implies the proper factorization of the amplitude. We refer the readers to section~4.1 of Ref.~\cite{Arkani-Hamed:2017mur} for more details.

The same strategy applies to the open associahedra. First, both the facet $X_{l_1,l_m}=0$ and $X_{l_m,r_m}=0$ ($m$ labels the last block) have the direct product geometry~\eqref{eq:directH} and~\eqref{eq:directA} by construction, while both $\pmb{H}^{L/R}$ are given by the same procedure as in the last section. A simple but slightly nontrivial example is the $X_{3,5}=q^2=0$ facet of $\mathcal{A}_6[(1,2),(3,4),(5,6)]$. In the subspace constraints,
\begin{align}
    \pmb{H}_6[(1,2),(3,4),(5,6)]=\big\{s_{2,5,6}=-c_{2,5,6},s_{3,5,6}=-c_{3,5,6},s_{3,4,6}=-c_{3,4,6}\big\}\,,
\end{align}
the last one becomes $X_{1,4}=c_{3,5,6}+X_{1,3}+X_{1,5}$ when $X_{3,5}=0$. It is gives no constraints on $X_{1,3}$ (identified as $X_{1,q}$) and $X_{1,5}$ since $X_{1,4}$ is strictly positive. The other two constraints can be directly identified as the subspace $\pmb{H}_5[(1,2),q,(5,6)]$,
\begin{align}
    -s_{2,5,6}=c_{2,5,6} & & \longrightarrow & & & X_{2,5}+X_{1,q}-c_{2,5,6}=X_{1,5}\,, \nonumber\\
    -s_{3,4,6}=c_{3,4,6} & & \longrightarrow & & & X_{1,5}+X_{q,6}-c_{3,4,6}=X_{1,q}\,.
\end{align}
We have thus showed that $\mathcal{A}_{6}[(1,2),(3,4),(5,6)]\big|_{X_{3,5}=0}\cong \mathcal{A}_5[(1,2),q,(5,6)]$. It is a direct product geometry since the three-point associahedron $\mathcal{A}_3[-q,(3,4)]$ is zero dimensional.

The main purpose of this section is to characterize generic facets of the open associahedra, which contain more geometric structures than the direct product.

\subsection{Fiber-product geometries}

For a positive geometry to be dual to scattering amplitudes, the residue of the canonical form on a facet $X=0$ must factorize,
\begin{equation}
\textrm{Res}_{X=0}\Omega[\mathcal{A}]=\Omega[\mathcal{A}_L]\wedge\Omega[\mathcal{A}_R]\,.
\label{reuiqrementss}
\end{equation}
This is a geometric manifestation of cluster decomposition. For the amplituhedron, corresponding to amplitudes in $\mathcal{N}=4$ SYM~\cite{Arkani-Hamed:2017vfh} and cluster polytopes, corresponding to amplitudes in bi-adjoint theory~\cite{Arkani-Hamed:2017mur,Arkani-Hamed:2019vag}, the facets are simply \emph{direct product geometries}:
\begin{equation}
    \mathcal{A}|_{X=0} \cong \mathcal{A}_{L}\times \mathcal{A}_{R} \,,
\label{factorization}
\end{equation}
which is sufficient for eq.~\eqref{reuiqrementss} to hold. For example, a square is a direct product of two line segments,
\begin{equation}\label{expgeocan}
\begin{split}
&\mathcal{A}_{\textrm{square}}= \{ 0<x<1\}\times\{0<y<1 \}  \, ,  \\
&\Omega[\mathcal{A}_\textrm{square}] = \left ( \frac{1}{x}-\frac{1}{x-1}\right )dx\wedge \left (\frac{1}{y}-\frac{1}{y-1} \right )dy \,,
\end{split}
\end{equation}
whose canonical form factorizes into those of the line segments. However, for the polytope dual to bi-color amplitudes, we find that not every facet is a direct product geometry, although the residue of the canonical form still factorizes as eq.~\eqref{reuiqrementss}. Instead of direct product geometries, these facets correspond to \textit{fiber-product geometries}, which still obey eq.~\eqref{reuiqrementss} but are more general than eq.~\eqref{factorization}. The simplest example of a fiber-product geometry is a trapezoid,
\begin{equation}
\begin{split}
&\mathcal{A}_{\textrm{trapezoid}}=\{ 0<x<1+y \,, 0<y<1\} \, ,  \\
&\Omega[\mathcal{A}_{\text{trapezoid}}] = \left ( \frac{1}{x}-\frac{1}{x-1-y}\right )dx\wedge \left (\frac{1}{y}-\frac{1}{y-1} \right )dy \ .
\end{split}
\label{expgeocanss}
\end{equation}
The trapezoid is clearly not a direct product geometry. Nevertheless, the canonical form still factorizes into those of line segments. We can think of the trapezoid as one line segment being fibered through out the other while the length of the former depends linearly on the coordinates of the latter.

Now we move on to generic cases. Consider two polytopes $\mathcal{A}_{L/R}$ living in the projective space $\mathbb{P}^m$ and $\mathbb{P}^n$ that are bounded by $M$ and $N$ facets respectively. The facets are specified by vectors in the dual space,
\begin{align}\label{eq:ALAR}
    & \mathcal{A}_L(X)=\big\{X\in\mathbb{P}^m\,\big|\,X\cdot W_i^L\geqslant 0\,, 1\leqslant i\leqslant M\big\}\,,& & W_i^L=(C^L_i,w_{i_1}^L,w_{i_2}^L,\ldots,w_{i_m}^L)\,,\nonumber\\
    & \mathcal{A}_R(Y)=\big\{Y\in\mathbb{P}^n\,\big|\,Y\cdot W_i^R\geqslant 0\,, 1\leqslant i\leqslant N\big\}\,,& & W_i^R=(C^R_i,w_{i_1}^R,w_{i_2}^R,\ldots,w_{i_n}^R)\,,   
\end{align}
where $X=(1,x_1,x_2,\ldots,x_m)$ and $Y=(1,y_1,y_2,\ldots,y_m)$ are homogeneous coordinates. 
We can embed the direct product $\mathcal{A}=\mathcal{A}_L\times\mathcal{A}_R$ into the space $\mathbb{P}^{m+n}$ with homogeneous coordinate $Z=(1,x_1,\ldots,x_m,y_1,\ldots,y_n)$,
\begin{align}
    \mathcal{A}(Z)=\mathcal{A}_L\times\mathcal{A}_R=\big\{Z\in\mathbb{P}^{m+n}\,\big|\,Z\cdot\mathcal{W}_i\geqslant 0\,,1\leqslant i\leqslant M+N\big\}\,,
\end{align}
which is bounded by $M{+}N$ facets,
\begin{subequations}
\label{eq:dpembed}
\begin{align}
    & W_i^L\rightarrow\mathcal{W}_i=(C_i^L,w_{i_2}^L,\ldots,w_{i_m}^L,\underbrace{0,0,\ldots,0}_{n\text{ zeros}})\,,\\
    \label{eq:dpembed2}
    & W_i^R\rightarrow\mathcal{W}_{M+i}=(C_i^R,\underbrace{0,0,\ldots,0}_{m\text{ zeros}},w_{i_1}^R,w_{i_2}^R,\ldots,w_{i_n}^R)\,.
\end{align}
\end{subequations}
It is easy to show that the canonical form of $\mathcal{A}(Z)$ factorizes,
\begin{equation}
\Omega[\mathcal{A}]=\Omega[\mathcal{A}_{L}] \wedge \Omega[\mathcal{A}_{R}] \,.
\label{productss}
\end{equation}
For example, the line segments appeared in eq.~\eqref{expgeocan} are bounded by
\begin{align}\label{eq:square}
    W_1^L=W_1^R=(1,0)\,,& & W_2^L=W_2^R=(-1,1)\,.
\end{align}
After we embed them into $\mathbb{P}^2$, they form the four boundary components of the square,
\begin{align}
    & W_1^L\rightarrow\mathcal{W}_1=(1,0,0)\,,& & W_2^L\rightarrow\mathcal{W}_2=(-1,1,0)\,,\nonumber\\
    & W_1^R\rightarrow\mathcal{W}_3=(1,0,1)\,,& & W_2^R\rightarrow\mathcal{W}_4=(-1,0,1)\,.
\end{align}

Next, to construct a fiber-product geometry $\mathcal{A}_L\ltimes\mathcal{A}_R$, we first fiber $\mathcal{A}_L$ through $\mathcal{A}_R$ by making the shape of $\mathcal{A}_L$ depend linearly on the coordinates of $\mathcal{A}_R$. This can be done by deforming the facet vectors as
\begin{align}
    W_i^L\rightarrow\widetilde{W}_i^L=\Big(C_i^L+\sum_{j=1}^{n}\alpha_jy_j,,w_{i_1}^L,w_{i_2}^L,\ldots,w_{i_m}^L\Big),
\end{align}
where $\alpha_j$'s are a set of constants such that the topology of $\mathcal{A}_L$ does not change within $\mathcal{A}_R$. We denote this deformed polytope as $\mathcal{A}_L(X;Y)$. We can trivially obtain the canonical form $\Omega[\mathcal{A}_L(X;Y)]$ from $\Omega[\mathcal{A}_L(X)]$ by shifting the $C_i^L$'s correspondingly. In our trapezoid example~\eqref{expgeocanss}, only one facet receive such a shift,
\begin{align}
    W_2^L=(-1,1)\;\rightarrow\;\widetilde{W}_2^L=(-1+y,1)\,.
\end{align}
while the rest remain the same as eq.~\eqref{eq:square}.
Then similar to the direct product, $\mathcal{A}_L\ltimes\mathcal{A}_R$ is the polytope bounded by $M{+}N$ facets $\{\widetilde{W}_i^L,W_i^R\}$ after they are embedded into $\mathbb{P}^{m+n}$. While the embedding $W_i^R$ still follows eq.~\eqref{eq:dpembed2}, for $\widetilde{W}_i^L$ we have
\begin{align}
    \widetilde{W}_i^L\rightarrow\mathcal{W}_i=(C_i^L,w_1^i,w_2^i,\ldots,w_m^i,\alpha_1,\alpha_2,\ldots,\alpha_n)\,.
\end{align}
Back to the trapezoid example, this means $\widetilde{W}_2^L\rightarrow\mathcal{W}_2=(-1,1,1)$, while the other facets follow the procedure of eq.~\eqref{eq:dpembed}. As we will show with more details in appendix~\ref{sec:semidirpro}, the canonical forms of the fiber-product geometries factorize as
\begin{align}
    \Omega[\mathcal{A}_L\ltimes\mathcal{A}_R]=\Omega[\mathcal{A}_L(X;Y)]\wedge\Omega[\mathcal{A}_R(Y)]\,.
\end{align}

\indent The first example of a fiber-product geometry for open associahedra is the $X_{1,4}\rightarrow 0$ facet  of $\mathcal{A}[(1,2),3,(4,5),6]$. The associated constraint equations are 

\begin{equation}
\begin{split}
\pmb{H}_{6}[(1,2),3,(4,5),6]=\left\{\begin{array}{c}
s_{3,5},\ s_{2,4,5},\ s_{2,6},\ s_{3,6},\ s_{4,6}\\
\text{ equal to negative constants}
\end{array}\right\}\, \ .
\end{split}  
\label{6pointexrest}
\end{equation}

\noindent There are only three incompatible planar variables, which can be written in a manifestly positive form using three of the five constraints in eq.~\eqref{6pointexrest},

\begin{equation}
\begin{split}
X_{3,5} \big|_{X_{1,4}=0} &=c_{3,6}+c_{3,5}+X_{1,3} \ ,\\
X_{2,6} \big|_{X_{1,4}=0} &= c_{3,6}+c_{2,6}+X_{4,6}\ , \\
X_{3,6} \big|_{X_{1,4}=0} &= c_{3,6}+X_{1,3}+X_{4,6} \ , \\
\end{split}    
\end{equation}

\noindent The remaining two constraints fall into two groups associated with the left and right amplitudes on the factorization channel respectively. For the right group, the constants become the standard restriction equations associated with the ordering $\alpha_{R}=[p,(4,5),6]$:

\begin{equation}
X_{4,6}+X_{1,5}-c_{4,6}=0\;\rightarrow\; c_{4,6}=-s_{4,6} \ .
\label{rightrestrequations}
\end{equation}

\noindent However, the restrictions of the left amplitude, with the ordering $\alpha_{L}=[(1,2),3,q]$, gain a dependence on $X_{R}$ variables:

\begin{equation}
X_{2,3}+X_{1,3}-(c_{2,4,5}+c_{2,6}+X_{4,6})=0 \;\rightarrow\; c_{2,q}+X_{4,6}=-s_{2,q} \ ,
\label{leftrestrequations}
\end{equation}

\noindent where $c_{2,q}=c_{2,4,5}+c_{2,6}$ and $q=-p$. In particular, the constraint equation is parametrized by $X_{4,6}\in X_R$. A visualization is provided in figure~\ref{fig:semidirprod}. Therefore, we find the facet geometry associated with eqs.~\eqref{leftrestrequations} and \eqref{rightrestrequations} is a fiber geometry instead of a direct product:

\begin{equation}
\mathcal{A}_{6}[(1,2),3,(4,5),6]\big|_{X_{1,4}\rightarrow 0}=\mathcal{A}_{4}[(1,2),3,q]\ltimes  \mathcal{A}_{4}[p,(4,5),6] \ .
\end{equation}

\begin{figure}[t]
	\centering
	\begin{tikzpicture}[every node/.style={font=\footnotesize}]
	\filldraw [color=orange!30] (2,0) -- (0,0) -- (0,2) -- (4,2) -- (2,0);
	\draw [thick,-stealth] (0,0) -- (3.75,0) node[below=0pt]{$X_{2,q}$};
	\draw [thick,-stealth] (0,0) -- (0,3.75) node[left=0pt]{$X_{4,6}$};
	\draw [thick] (0,2) -- (4,2) -- (2,0);
	\filldraw (0,0) circle (1pt);
	\filldraw (0,2) circle (1pt) node[left=0pt]{$c_{4,6}$};
    \filldraw (2,0) circle (1pt) node[below=0pt]{$c_{2,q}$};
    \filldraw (4,2) circle (1pt);
    \begin{scope}[xshift=5.5cm]
	\draw [thick] (0,0.5) node[below=0pt]{$X_{4,6}$} -- (0,2.5) node[above=0pt]{$X_{p,5}$};
	\draw (1,1.5) node{$\ltimes$};
	\draw [thick] (2.3,1.5) node[left=0pt]{$X_{2,q}$} -- (4.3,1.5) node[right=0pt]{$X_{1,3}$};
	\draw (-1,1.5) node{$=$};
	\end{scope}
	\end{tikzpicture}
	\caption{The facet geometry of $X_{1,4}$ in $\mathcal{A}[(1,2),3,(4,5),6]$. The geometry is the semi-direct product of two bounded lines: $\mathcal{A}[(1,2),3,q]\ltimes  \mathcal{A}[p,(4,5),6] $. If it were a direct product, the geometry would be a square. }
	\label{fig:semidirprod}
\end{figure}
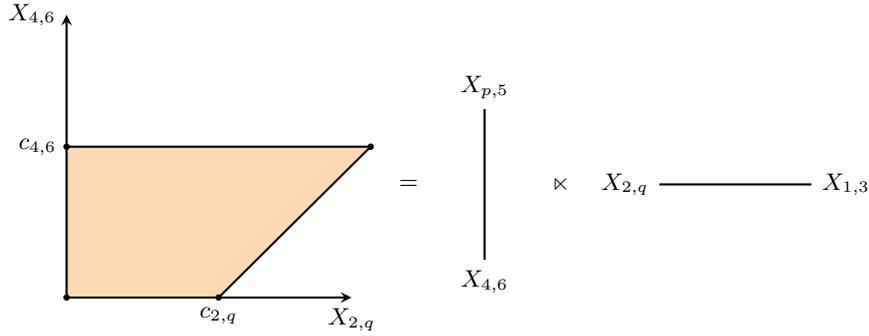

In principle, one can further generalize the fiber-product geometry by also deforming $\mathcal{A}_R$ with the coordinate of $\mathcal{A}_L$, namely,
\begin{align}\label{eq:Cshift}
    & W_i^L\rightarrow\widetilde{W}_i^L=\Big(C_i^L+\sum_{j=1}^{n}\alpha_jy_j,,w_{i_1}^L,w_{i_2}^L,\ldots,w_{i_m}^L\Big)\,,\nonumber\\
    & W_i^R\rightarrow\widetilde{W}_i^R=\Big(C_i^R+\sum_{j=1}^{m}\beta_jx_j,w_{i_1}^R,w_{i_2}^R,\ldots,w_{i_n}^R\Big)\,.
\end{align}
After being embedded into $\mathbb{P}^{m+n}$, we call the geometry bounded by the above $M{+}N$ facets a \emph{semi-direct product} $\mathcal{A}_L\bowtie\mathcal{A}_R$. However, the canonical form $\Omega(\mathcal{A}_L\bowtie\mathcal{A}_R)$ factorizes only if either $\alpha_i=0$ or $\beta_i=0$. Namely, it is a fiber product. In other words, for a general closed cluster polytope, each facet should either be a direct product or a fiber-product geometry. The same conclusion also applies to those open polytopes that are obtained by sending certain facets of some closed polytopes to infinity. Details of the derivation will be given in appendix~\ref{sec:semidirpro}.

\subsection{Extended equivalence classes}
\label{sec:extendeequiclass}

\indent The fiber geometry is not the only new phenomena that appears on the factorization channels of open associahedra. There is also a larger, continuous equivalence class of subspaces. In this section, we study this large equivalence class of subspaces associated with open associahedra. We first find by direct computation a new class of open associahedra geometry which appears at $n=6$. Motivated by the appearance of such a subspace, we attempt a direct construction all subspaces at $n=6$ before making an all $n$ conjecture. \\
\indent To see the appearance of this larger equivalence class on the facets of open associahedra, consider the $X_{2,4}=q^2\rightarrow 0$ factorization channel of $\mathcal{A}_6[1,(2,3),4,(5,6)]$. The restriction equations are

\begin{equation}
\begin{split}
\pmb{H}_{6}[1,(2,3),4,(5,6)]=\left\{\begin{array}{c}
s_{2,4}, \ s_{4,6},\ s_{2,3,6},\ s_{2,5,6},\ s_{3,5,6}\\
\text{ equal to negative constants}
\end{array}\right\}\, \ .
\end{split}  
\label{6pointexresttwo}
\end{equation}

\noindent There are now two incompatible planar variables $\{X_{1,3},X_{3,5}\}$, which can be written in a manifestly positive form using two of the five constraints in eq.~\eqref{6pointexresttwo},

\begin{align}
X_{3,5}\big|_{X_{2,4}=0}=c_{2,4}+X_{2,5}\,,& &  
X_{1,3}\big|_{X_{2,4}=0}=c_{2,4}+c_{2,5,6}+X_{1,5}\,.
\label{ireeleconst}
\end{align}

\noindent The other three constraints organize themselves into $\pmb{H}^R_{5}[1,q,4,(5,6)]$ as $\pmb{H}^L_{3}[-q,(2,3)]$ does not contribute any constraints. The constraint $s_{4,6}=-c_{4,6}$ remains unchanged and on the support of this factorization $s_{2,3,6}=s_{q,6}$ such that the constraint $s_{2,3,5}=-c_{2,3,6}$ naturally reduces to $s_{q,6}=-c_{q,6}$ with $c_{q,6}=c_{2,3,6}$. However, the constraint $s_{3,5,6}=-c_{3,5,6}$ has to be combined linearly with $s_{2,5,6}=-c_{2,5,6}$ to produce the final element of $\pmb{H}^R_{5}[1,q,4,(5,6)]$,
\begin{align} 
s_{3,5,6}+s_{2,5,6}=-c_{3,5,6}-c_{2,5,6}\;\rightarrow\; s_{q,5,6}+X_{1,5}=-c_{q,5,6}
\end{align}
where $c_{q,5,6}=c_{2,5,6}+c_{3,5,6}$. Therefore, we have
\begin{align}\label{newrestrictionss}
    \pmb{H}_5^{R}[1,q,4,(5,6)]=\big\{s_{4,6}=-c_{4,6}\,,s_{q,6}=-c_{q,6}\,,s_{q,5,6}+X_{1,5}=-c_{q,5,6}\big\}\,.
\end{align}
A visualization of the resultant polytope $\mathcal{A}_5^R[1,q,4,(5,6)]$ is provided in figure~\ref{fig:extendedequivalenclass}, after being relabeled to $[1,2,3,(4,5)]$. We now see that eq.~\eqref{newrestrictionss} is a deformed restriction, that does not correspond to any restrictions from section~\ref{sec:invsoftconstruction}. Note that the amplitude $A_5[1,2,3,(4,5)]$ is equal to $A_5[1,2,3,4,5]$, the $\mathfrak{adj}$ amplitude, as no planar variables are forbidden. Therefore, this deformation also corresponds to a new equivalence class of subspaces for $\mathfrak{adj}$ amplitudes.

\begin{figure}[t]
	\centering
	\begin{tikzpicture}[every node/.style={font=\footnotesize}]
	\filldraw [color=orange!30] (0,0) -- (0,1) -- (1,2) -- (5,4) -- (1,0) -- (0,0);
	\draw [thick,-stealth] (0,0) -- (3.75,0) node[below=0pt]{$X_{1,3}$};
	\draw [thick,-stealth] (0,0) -- (0,3.75) node[left=0pt]{$X_{3,5}$};
	\draw [thick] (0,0) -- (0,1) -- (1,2) -- (5,4) -- (1,0) -- (0,0);
	\filldraw (0,1) circle (1pt) node[left=0pt]{$c_{3,5}$};
    \filldraw (1,0) circle (1pt) node[below=0pt]{$c_{2,5}$};
    \filldraw (0,0) circle (1pt) (1,2) circle (1pt) (5,4) circle (1pt);
    \draw [dashed] (1,2) -- (0,2) node[left=0pt]{$c_{2,4,5,}{+}c_{3,5}$};
    \begin{scope}[xshift=6.5cm]
    \filldraw [color=orange!30] (0,0) -- (0,1) -- (1,2) -- (3,2) -- (1,0) -- (0,0);
    \draw [thick] (0,0) -- (0,1) -- (1,2) -- (3,2) -- (1,0) -- (0,0);
    \draw [thick,-stealth] (0,0) -- (3.75,0) node[below=0pt]{$X_{1,3}$};
	\draw [thick,-stealth] (0,0) -- (0,3.75) node[left=0pt]{$X_{3,5}$};
    \filldraw (0,1) circle (1pt) node[left=0pt]{$c_{3,5}$};
    \filldraw (1,0) circle (1pt) node[below=0pt]{$c_{2,5}$};
    \filldraw (0,0) circle (1pt) (1,2) circle (1pt) (3,2) circle (1pt);
    \draw [dashed] (1,2) -- (0,2) node[left=0pt]{$c_{2,4,5,}{+}c_{3,5}$};
	\end{scope}
	\end{tikzpicture}
	\caption{The left geometry corresponds to the $A_5[1,2,3,(4,5)]$ sub-amplitude that appears on the $X_{2,4}$ facet of $A_6[1,(2,3),4,(5,6)]$. The right geometry is the geometry corresponding to $A_5[1,2,3,(4,5)]$ derived from the inverse soft construction in Section \ref{sec:invsoftconstruction}. }
	\label{fig:extendedequivalenclass}
\end{figure}
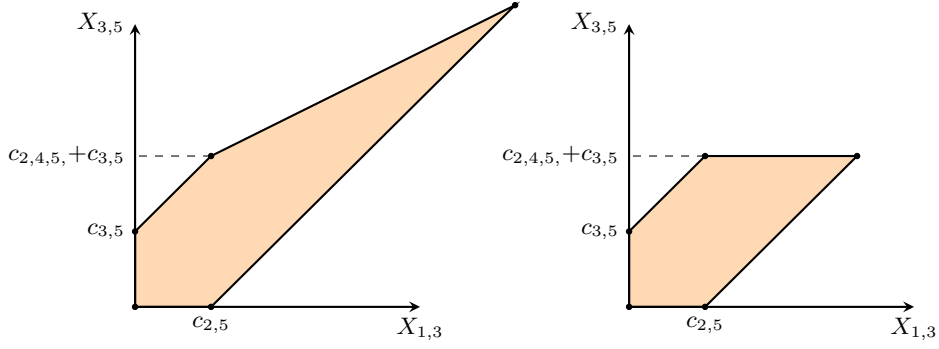
\indent Motivated by the appearance of deformed geometries, we now turn to a direct construction of this larger equivalence class of subspaces for open associahedra at $n=6$. Consider the open associahedra $\mathcal{A}_6[(1,2),(3,4),(5,6)]$. The restriction equations from section~\ref{sec:invsoftconstruction} are 

\begin{equation}
\pmb{H}_6[(1,2),(3,4),(5,6)]=\left\{\begin{array}{c}
s_{2,5,6}, \quad s_{3,5,6}, \quad s_{3,4,6}\\
\text{ equal to negative constants}
\end{array}\right\}\,  \ .
\end{equation}

\noindent In $Y=(1,X_{1,3},X_{3,5},X_{1,5})$ coordinates, the non-trivial facet vectors are 

\begin{align}\label{orfgoe}
& W_{1,3}=(0,1,0,0) \,,& &   W_{1,4}=(c_{3,5,6},1,-1,1)\,, \nonumber\\
& W_{3,5}=(0,0,1,0) \,,& &   W_{3,6}=(c_{3,4,6},1,1,-1)\,, \\
& W_{1,5}=(0,0,0,1) \,,& &   W_{2,5}=(c_{2,5,6},-1,1,1)\,. \nonumber
\end{align}

\noindent The canonical form associated with the above facet vectors can be calculated using the vertex expansion of the rational function given in appendix~\ref{sec:dualpolytope},

\begin{equation}
\underline{\Omega}=\sum_{v\in \textrm{vertices}}\frac{\langle W^{\star} \Pi_{I\in v}W_{I}\rangle}{(Y\cdot W^{\star})\Pi_{I\in v}  (Y\cdot W_{I}) } \ .  
\label{fess}    
\end{equation}

\noindent If $W^{\star}=(1,0,\ldots,0)$, one finds $\langle \Pi_{I\in v}W_{I}\rangle=\pm 1$ and that this expansion is equivalent to the Feynman diagram expansion of the canonical rational function. To show that eq. (\ref{fess}) is still equivalent to a Feynamn diagram expansion after some deformation, one must show that $\langle \Pi_{I\in v}W_{I}\rangle=\pm 1$ after the deformation. A class of continuous deformations of eq.~\eqref{orfgoe} which preserve the Feynman diagram expansion is

\begin{figure}
\centering
  \includegraphics[scale=0.5]{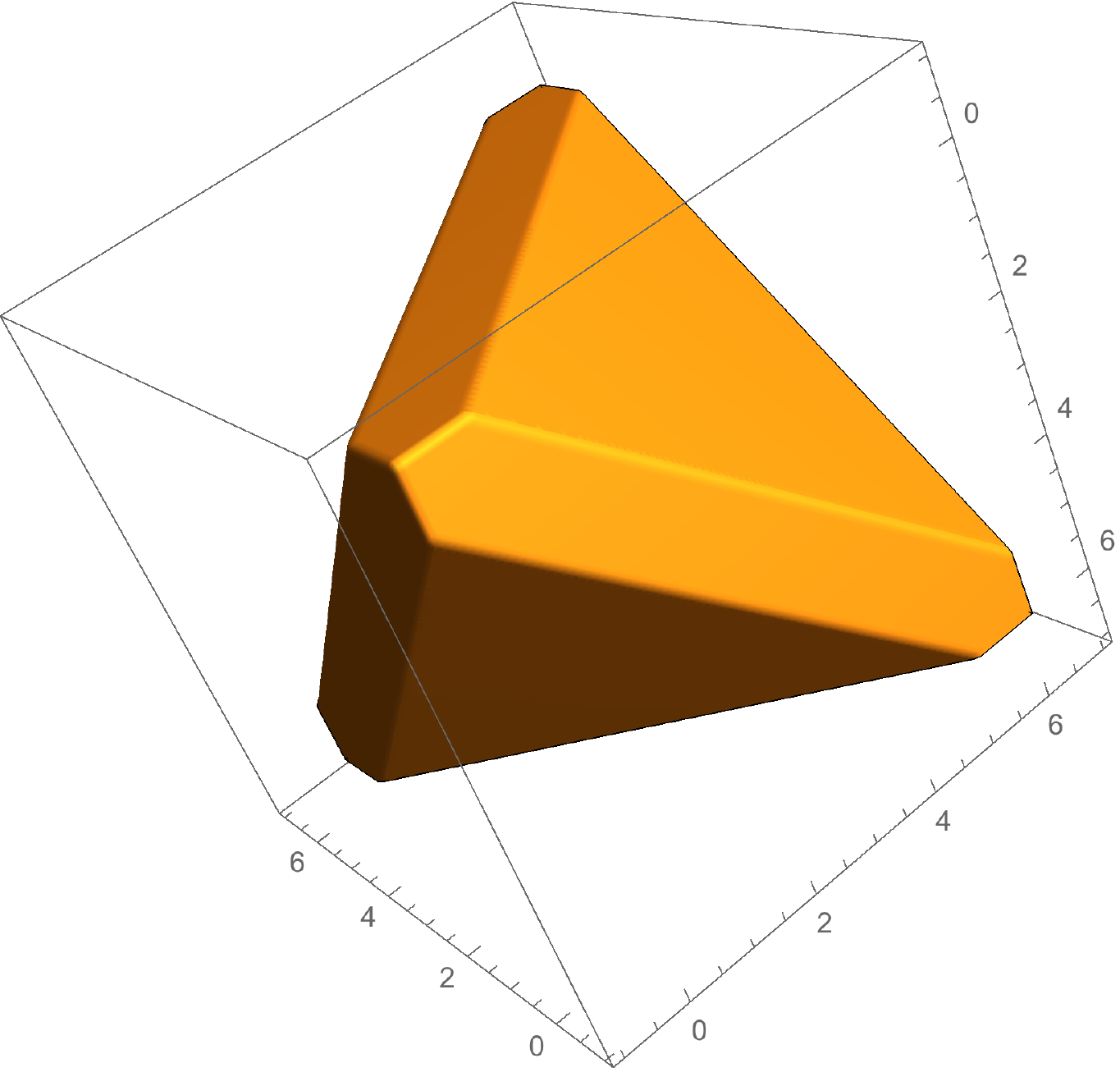} \ \
  \includegraphics[scale=0.5]{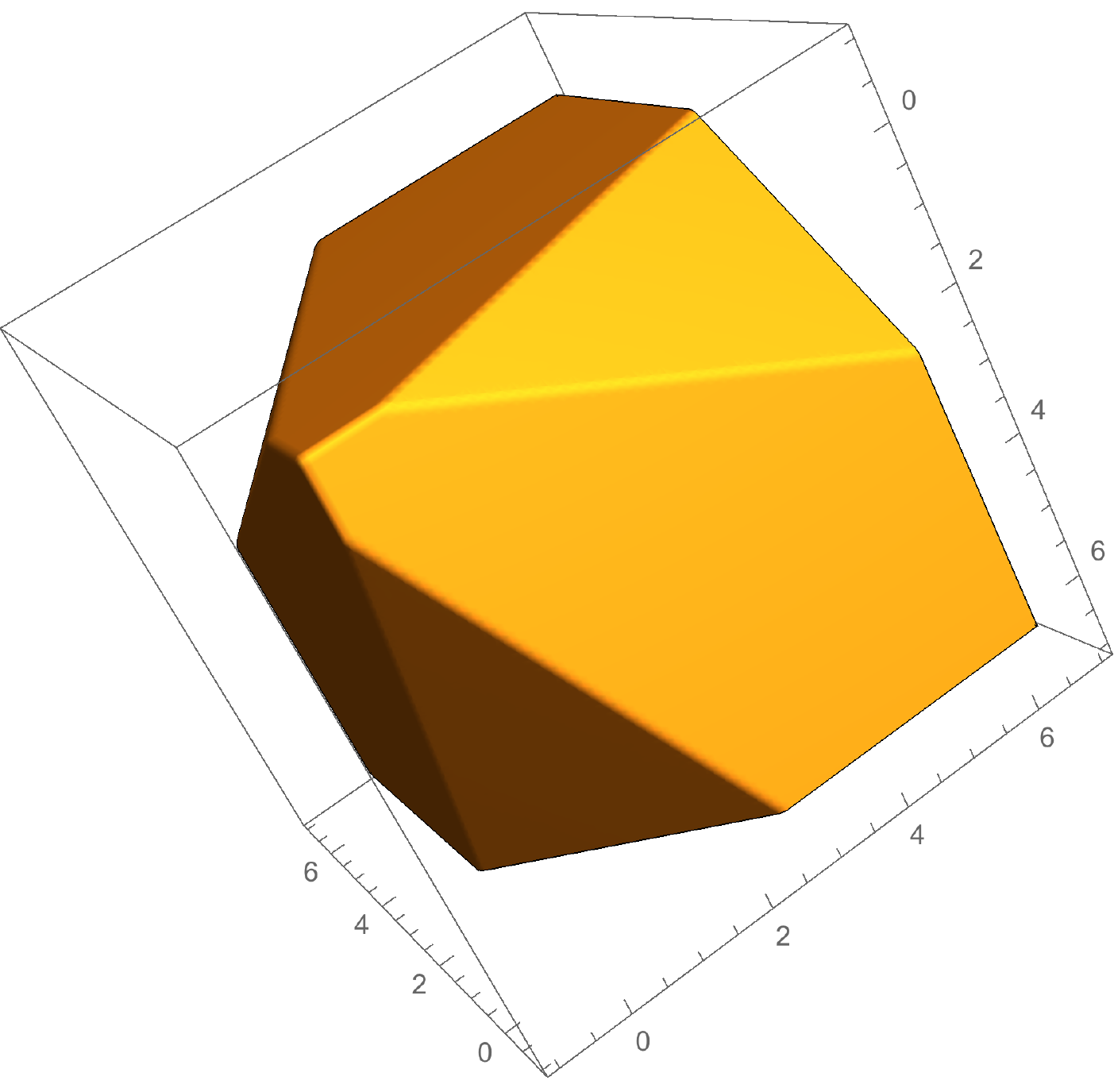} 
  \caption{The left figure corresponds to the undeformed geometry, $\mathcal{A}_6[(1,2),(3,4),(5,6)]$, given by restriction equations in section~\ref{sec:invsoftconstruction}. The right figure corresponds to the deformed geometry, $\mathcal{A}'_6[(1,2),(3,4),(5,6)]$, with $\alpha_{i,j}=2$. Note that the geometry is unbounded in both images and can be visualized as a cone with additional structure. }
  \label{boat2}
\end{figure}

\begin{align}\label{explictdefor}
& W_{1,3}'=(0,1,0,0)\,,& &   W_{1,4}'=(c_{3,5,6},\alpha_{1,1},-1,\alpha_{1,2})\,,  \nonumber\\
& W_{3,5}'=(0,0,1,0)\,,& &   W_{3,6}'=(c_{3,4,6},\alpha_{2,1},\alpha_{2,2},-1)\,,  \\
& W_{1,5}'=(0,0,0,1)\,,& &   W_{2,5}'=(c_{2,5,6},-1,\alpha_{3,1},\alpha_{3,2})\,. \nonumber
\end{align}

\noindent where all $\alpha_{i,j}\geqslant 1$. The additional restriction that $\alpha_{i,j}\geqslant 1$ in eq. (\ref{explictdefor}) is so the geometry has the correct facet and vertex structure. One can explicitly check that the canonical form of the geometry corresponding to the facets in eq. (\ref{explictdefor}) is functionally equivalent to eq.~\eqref{orfgoe}. Although it is easier to see the validity of the deformation at the level of facets, we can also interpret the deformation at the level of the constraints as 

\begin{equation}
\begin{split}
\hat{c}_{3,5,6}&=-s_{2,5,6}-(\alpha_{1,1}-1)X_{1,3}-(\alpha_{1,2}-1)X_{1,5}  \ , \\
\hat{c}_{3,4,6}&=-s_{2,5,6}-(\alpha_{2,1}-1)X_{1,3}-(\alpha_{2,2}-1)X_{3,5}  \ , \\
\hat{c}_{2,5,6}&=-s_{2,5,6}-(\alpha_{3,1}-1)X_{3,5}-(\alpha_{3,2}-1)X_{1,5} \ . \\
\end{split} 
\label{newdeforss}
\end{equation}

\noindent Therefore, there is a continuous class of valid subspaces which yield the canonical forms of the same functional form. \\

\indent For general bi-color amplitudes, the general equivalence class of subspaces seems too large to study. For instance, the number of valid deformations increases drastically from $k=3$ to $k=4$. We find that for the constraints of the form $s_{\mathsf{B}_{i},\mathsf{B}_{j}}=-c_{\mathsf{B}_{i},\mathsf{B}_{j}}$, we have a class of continuous deformations,\footnote{We note that $\mathsf{B}_i$ and $\mathsf{B}_j$ can be blocks at different levels. According to the recursive construction given in section~\ref{sec:invsoftconstruction}, if $\mathsf{B}_i$ and $\mathsf{B}_j$ are blocks at the same level, then one of them must be an $\mathfrak{adj}$ state.}

\begin{equation}
c_{\mathsf{B}_{i},\mathsf{B}_{j}}=-s_{\mathsf{B}_{i},\mathsf{B}_{j}}-\alpha_{i,j}s_{\mathsf{B}_{i}}-\beta_{i,j}s_{\mathsf{B}_{j}} \, , 
\label{clsoedeforss}
\end{equation}

\noindent where $\alpha_{i,j}\geqslant 0$ and $\beta_{i,j}\geqslant 0$, that preserve the Feynman diagram expansion. Although we do not have a proof that eq.~\eqref{clsoedeforss} is a valid class of deformations, we have performed checks for a large number of amplitudes and believe that the deformation can be derived from a generalized inverse-soft construction where we loosen the constraint that the restriction equations must correspond to generalized mandelstam variables. Note that this class of deformations is smaller than the general class of deformations. For example, when applying eq.~\eqref{clsoedeforss} to $\mathcal{A}_6[(1,2),(3,4),(5,6)]$, we get a special class of deformations given by eq.~\eqref{newdeforss}. In appendix~\ref{sec:factorizationchannel}, we give some explicit examples on how the fiber geometries and deformed constraints appear on the facets of the open associahedra.

\indent Given the existence of this large class of subspace, the most naive question one can ask is: what is the set of all geometries associated with a given canonical form? However, as we saw from direct computation, this question is hard to answer at large $n$, even for closed associahedra. For example, we just provided an extended equivalence class of constraints for closed associahedra that is much larger than the set of constraints given in~\cite{Arkani-Hamed:2017mur}. A more interesting problem is finding an extended equivalence class of open associahedra that is \textit{closed}. By closed, we refer to sets of geometries whose facet geometries are fiber products of geometries within the set. For example, the equivalence class of geometries for associahedra constructed in~\cite{Arkani-Hamed:2017mur} are closed.
It is still an open question whether there exists a closed equivalence class of geometries for open associahedra. Computational analysis shows that even after including the deformation provided by eq.~\eqref{clsoedeforss}, the geometries constructed in section~\ref{sec:invsoftconstruction} are still not closed for $n>10$. \\

\section{Recursion for open associahedra}
\label{recursionsection}

\noindent In this section, we give a recursion procedure for open associahedra, generalizing the recursion procedures of~\cite{Salvatori:2019phs, Arkani-Hamed:2019vag, Yang:2019esm}. This recursion is BCFW-like and could possibly offer insight into more complex BCFW-like recursions for theories with non-trivial flavor structure \cite{Britto:2005fq}. There are a number of challenges to finding a recursion for the open associahedra compared to the all $\mathfrak{adj}$ case. The most obvious problem is that the corresponding positive geometry is unbounded. However, applying some extra-steps, we generalize the proof~\cite{Yang:2019esm} for open associahedra carved out by the constraints $\pmb{H}_n[\alpha]$. For the purposes of this section, we restrict our recursion to the diagonal elements of double partial amplitudes, $m_{n}[\alpha|\alpha]$, which we denote simply as $A_n[\ldots]$.\footnote{Depending on the contexts, $\ldots$ will refer to the ordering $\alpha$ or to kinematic data, such as $X_{i,j}$.} Any off diagonal $m_{n}[\alpha|\beta]$ can be written as the product of lower point diagonal partial amplitudes~\cite{Cachazo:2013iea, Arkani-Hamed:2017mur}. \\
\indent We begin with a quick review of the recursion for ABHY closed associahedra, following Ref.~\cite{Yang:2019esm}, before giving the open-polytope generalization of the recursion for increasingly complex orderings.

\subsection{Review of recursion for closed associahedra}
\label{closedassrec}

We first review the recursion for double partial bi-adjoint amplitudes, $A_n[\alpha]$, with ordering $\alpha=[1,2,3,\ldots,n]$, following the field theoretic derivation of Ref.~\cite{Yang:2019esm} but focusing on the one-variable shift only. We also provide an example of the recursion for $A_5[1,2,3,4,5]$ at the end of the section. \\
\indent For bi-adjoint associahedra, the restriction equations are

\begin{equation}
\pmb{H}_{n}=\{ s_{i,j}=-c_{i,j}\,|\,1<i<j-1<n \}  \ .  
\label{biadjointrestr}
\end{equation}

\noindent Using these constraints, we can find a closed form solution to any $X_{i,j}$ variable in terms of the planar variables $X_{2,i}$ and the $c$-constants,\footnote{For this expansion to hold, we require that $3\leqslant i\leqslant n-1$ and $j\neq 2,3$. We further require $i<j$ except for $j=1$, for which we define $j-1=n$ in $\sum_b$.}

\begin{equation}
X_{i,j}=X_{2,j}-X_{2,i+1}+\sum_{a=2}^{i-1}\sum_{b=i+1}^{j-1}c_{a,b} \ . 
\label{closedformequt}
\end{equation}

\noindent Therefore, the natural basis for any recursion of the amplitude is 

\begin{equation}
Y=(1,X_{2,4},X_{2,5},\ldots,X_{2,n}) \ . 
\label{adjbaissssss}
\end{equation}

\noindent All planar variables not in the chosen basis can be written as a linear combination of $X_{2,a}$ and $c_{i,j}$. Now consider the contour integral

\begin{equation}
A_n=\oint_{|z-1|=\epsilon}\frac{dz}{2\pi i} \frac{z}{z-1}A_n[zX_{2,4}] \ ,
\label{orgcontour}
\end{equation}

\noindent where only $X_{2,4}$ is shifted: $X_{2,4}\rightarrow zX_{2,4}$. The contour of integration in the complex plane is a small circle around $z=1$ and the original amplitude is the residue of the integrand at $z=1$. The residue theorem leads to

\begin{equation}
\oint_{|z-1|=\epsilon} \frac{dz}{2\pi i}\frac{z}{z-1}A_n[zX_{2,4}] = -\Bigg ( \textrm{Res}_{\infty}+\sum_{z_i\in\,\textrm{finite poles}} \textrm{Res}_{z_{i}}\Bigg ) \frac{z}{z-1}A_n[zX_{2,4}] \,.
\label{summation}
\end{equation}

\noindent There are two contributions in eq.~\eqref{summation}: residues at finite poles, which correspond to intermediate states going on-shell, and at the pole at infinity, which is zero as we will see. 

\paragraph{Finite poles} The residues at finite poles correspond to deformed planar variables going to zero, $\hat{X}_{i,j}\rightarrow 0$. Since the only shift variable is $X_{2,4}$, only the planar variables that depend on $X_{2,4}$ in the basis $Y$ can contribute, which are all of the form $X_{3,i}$. Furthermore, the dependence of $X_{3,i}$ on $X_{2,4}$ always takes the form 

\begin{equation}
X_{3,i} = -X_{2,4}+\ldots \ ,
\label{propto}
\end{equation}

\noindent such that after the shift $X_{2,4}\rightarrow zX_{2,4}$,

\begin{equation}
\hat{X}_{3,i}(z)=(1-z)X_{2,4}+X_{3,i} \ .
\label{genericshift}
\end{equation}

\noindent Solving for $\hat{X}_{3,i}=0$ gives 

\begin{equation}
z_{i}=1+\frac{X_{3,i}}{X_{2,4}}, \quad \text{and}\quad \hat{X}_{3,j}(z_i)=X_{3,j}-X_{3,i} \ .
\label{zival}
\end{equation}

\noindent Plugging eq. (\ref{zival}) into eq. (\ref{summation}), we can write the finite pole contributions as

\begin{equation}
\begin{split}
\textrm{Res}_{z_{i}}\left ( \frac{z}{z-1}A[zX_{2,4}] \right )&=-\left ( \frac{1}{X_{3,i}}+\frac{1}{X_{2,4}}\right )\hat{A}_{L}(z_i)\times \hat{A}_{R}(z_i) \ .
\label{rescalc}
\end{split}
\end{equation}

\noindent where the $\hat{A}_{L/R}(z_i)$ are sub-amplitudes evaluated under the shift $\hat{X}_{3,j}(z_i) = X_{3,j}-X_{3,i}$. \\ 

\paragraph{No Pole at Infinity} To show that the residue at infinity vanishes, we only need to analyze that diagrams that scale as $\mathcal{O}(1/z)$ or worse in the large $z$ limit. First, no diagrams can scale as $O(1)$ in the large $z$ limit, as each Feynman diagram corresponds to a vertex with $n-3$ intersecting facets of planar variables. Since the open associahedron is $(n-3)$ dimensional, there cannot be $(n-3)$ intersecting facets that are all perpendicular to the $X_{2,4}$ facet. Next, we consider the diagrams that scale as $\mathcal{O}(1/z)$. They contain exactly $n{-}4$ propagators/facets that are independent of $X_{2,4}$. We denote them as $X_{v_{i}}$ and group the $\mathcal{O}(1/z)$ diagrams by shared $X_{v_{i}}$ propagators. The coefficient is nothing but the canonical rational function of a line segment, which is the intersection of these $n{-}4$ facets,
\begin{equation}\label{sumfeydiacontr}
    \left(\frac{1}{\hat{X}}+\frac{1}{C-\hat{X}}\right)\frac{1}{\prod^{n-4}X_{v_i}}\equiv A_4[\hat{X},C]\frac{1}{\prod^{n-4}X_{v_i}}\,.
\end{equation}
We denote this one dimensional canonical rational function as $A_{4}[\hat{X},C]$ since it corresponds to a four-point sub-amplitude. We note that in principle $C$ is a linear combination of the $c_{i,j}$ constants and unshifted basis variables. When $z\rightarrow \infty$, we find that

\begin{equation}
A_{4}[\hat{X},C]\sim \frac{1}{z}A_4[\hat{X},0]\ .
\label{deformationss}
\end{equation}

\noindent From eq.~\eqref{sumfeydiacontr}, we have $A_4[\hat{X},0]=0$. In general, for a bounded geometry, its canonical rational function vanishes when a boundary component becomes degenerate. Since all the one dimensional boundary components of a closed associahedron are bounded, all terms of the form~\eqref{sumfeydiacontr} must vanish at $z\rightarrow\infty$ and there is no pole at infinity. In general, one can shift $k$ variables in the basis $Y$~\cite{Yang:2019esm}, 
\begin{align}
    A_n=-\sum_{z_i\in\,\text{finite poles}}\textrm{Res}_{z_i}\frac{z^k}{z-1}A_n[zX]\,,
\end{align}
and the pole at infinity vanishes due to the \emph{soft condition} $A_{k+3}[X,0]=0$ that holds by the sub-geometries appearing at the $\mathcal{O}(1/z^k)$ order. In fact, the soft condition is satisfied by the ABHY associahedra at arbitrary dimensions~\cite{Arkani-Hamed:2017mur}. 

\paragraph{Example:}  We consider the lowest-point, non-trivial example: $A_{5}[1,2,3,4,5]$. We work in the basis, $Y=(1,X_{2,4},X_{2,5})$, and shift 

\begin{equation}
X_{2,4}\rightarrow zX_{2,4}
\end{equation}

\noindent The only variables which depend on $X_{2,4}$ in this basis are $X_{1,3}$ and $X_{3,5}$. Applying eq.~\eqref{genericshift} and~\eqref{rescalc}, and summing over the contributions from these two facets, we find

\begin{align}\label{examplefinite5point}
A_{5}&=\sum_{i=1,5} \left ( \frac{1}{X_{3,i}}+\frac{1}{X_{2,4}}\right )\hat{A}_{L}(z_i)\times \hat{A}_{R}(z_i)\nonumber\\
&=\left ( \frac{1}{X_{1,3}}+\frac{1}{X_{2,4}} \right )\left (\frac{1}{X_{1,4}}+\frac{1}{\hat{X}_{3,5}(z_1)} \right )+\left ( \frac{1}{X_{3,5}}+\frac{1}{X_{2,4}} \right )\left (\frac{1}{\hat{X}_{1,3}(z_5)}+\frac{1}{X_{2,5}} \right ) \nonumber\\
&=\left ( \frac{1}{X_{1,3}}+\frac{1}{X_{2,4}} \right )\left (\frac{1}{X_{1,4}}+\frac{1}{X_{3,5}-X_{1,3}} \right )+\left ( \frac{1}{X_{3,5}}+\frac{1}{X_{2,4}} \right )\left (\frac{1}{X_{1,3}-X_{3,5}}+\frac{1}{X_{2,5}} \right ) \nonumber\\
&=\frac{1}{X_{1,4}X_{1,3}}+\frac{1}{X_{2,4}X_{1,4}}+\frac{1}{X_{3,5}X_{2,5}}+\frac{1}{X_{2,4}X_{2,5}}+\frac{1}{X_{1,3}X_{3,5}} \ .
\end{align}

\noindent The intermediate steps correspond to a partial triangulation of the polytope $\mathcal{A}_5$ by an unphysical boundary $X_{3,5}-X_{1,3}=0$, as visualized in figure~\ref{fig:5pointclosedexp}.

\begin{figure}[t]
	\centering
	\begin{tikzpicture}[every node/.style={font=\footnotesize}]
	\filldraw [color=orange!30] (0,0) -- (2,0) -- (4,2) -- (4,4) -- (0,4) -- (0,0);
	\draw [thick] (0,0) -- (2,0) -- (4,2) -- (4,4) -- (0,4) -- (0,0);
	\draw [thick] (0,2) -- (4,2);
	\filldraw (1,0) circle (0pt) node[below=0pt]{\large $X_{2,5}$};
    \filldraw (0,2) circle (0pt) node[left=0pt]{\large $X_{2,4}$};
    \filldraw (2,4) circle (0pt) node[above=0pt]{\large $X_{1,4}$};
    \filldraw (4,3) circle (0pt) node[right=0pt]{\large $X_{1,3}$};
    \filldraw (2.75,0.75) circle (0pt) node[right=0pt]{\large $X_{3,5}$};
    \filldraw (2,1.9) circle (0pt) node[above=0pt]{\large $X_{3,5}-X_{1,3}$};
	\end{tikzpicture}
	\caption{A visualization of the (partial) triangulation of $\mathcal{A}_5[1,2,3,4,5]$ corresponding to the single variable shift. }
	\label{fig:5pointclosedexp}
\end{figure}
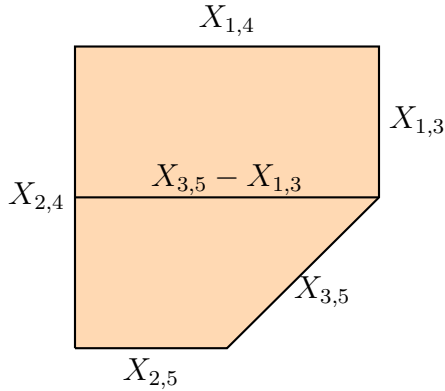

\subsection{Recursion for adjacent $\mathfrak{f}$-$\mathfrak{af}$ pair amplitudes}
\label{allfermpairsadj}

We want to consider a generalization of the recursion in section~\ref{closedassrec} to open associahedra. One might not expect such an recursion to exist since the soft condition $A_n[X,0]=0$ no longer holds. In fact, the derivation of the bi-adjoint recursion shows that for a particular $k$-variable shift, a much weaker condition is sufficient: 
only the $z$ dependent $k$ dimensional sub-geometries that appear at the $\mathcal{O}(1/z^k)$ order need to satisfy the soft condition. Although it is not difficult to construct an open associahedron in which all the two and higher dimensional sub-geometries are unbounded (see, for example, figure~\ref{fig:6pointexample}) that breaks the soft condition, there always exist bounded one dimensional sub-geometries (i.e., line segments) that respect it. Therefore, we can try to find a one-variable-shift scheme in which the $z$ dependence at the $\mathcal{O}(1/z)$ order appears in bounded one-dimensional sub-geometries only.

\indent As a warm-up example, consider the $A_{6}[(1,2),(3,4),(5,6)]$ amplitude, for which there are six $X_{i,j}$ allowed by flavor conservation. The restriction equations for this geometry were  given in eq.~\eqref{allrestrictionss}. We find that a convenient basis choice and shift scheme are
\begin{equation}
Y=(1,X_{1,3},X_{1,5},X_{3,5})\,,\quad X_{3,5}\rightarrow\hat{X}_{3,5}=zX_{3,5} \,.
\label{chosenbasisss}
\end{equation}
The three remaining $X_{i,j}$ can then be written in terms of $Y$ and $c$ constants,
\begin{equation}
\begin{split}
\hat{X}_{3,6}&=c_{3,4,6}-X_{1,5}+X_{1,3}+\hat{X}_{3,5}\, ,\\
\hat{X}_{2,5}&=c_{2,5,6}-X_{1,3}+\hat{X}_{3,5}+X_{1,5}\, , \\
\hat{X}_{1,4}&=c_{3,5,6}-\hat{X}_{3,5}+X_{1,5}+X_{1,3}\, , \\
\end{split}   
\label{allrestrictionsstwo}
\end{equation}
where the hat indicates that the variable is shifted.
There are two Feynman diagrams that scale as $\mathcal{O}(1/z)$, and they form a bounded one dimensional sub-geometry,

\begin{equation}
A_{6}\big|_{\mathcal{O}(1/z)}=\left ( \frac{1}{\hat{X}_{3,5}}+\frac{1}{\hat{X}_{1,4}} \right ) \frac{1}{X_{1,5}X_{1,3}} \,.
\label{feynmandiagramsss}
\end{equation}

\noindent It is easy to see that the large $z$ fall-off is actually $1/z^2$ for eq.~\eqref{feynmandiagramsss} such that there is no pole at infinity.
We now see that the basis and shift given in eq.~\eqref{chosenbasisss} form a valid recursion for $A_{6}[(1,2),(3,4),(5,6)]$. An example of the actual recursion procedure is given at the end of this section. In contrary, should we choose the basis as $Y=(1,X_{1,5},X_{2,5},X_{3,5})$ and still shift $X_{3,5}$, the $\mathcal{O}(1/z)$ order would become
\begin{align}
    A_6\big|_{\mathcal{O}(1/z)}=\frac{1}{\hat{X}_{1,3}X_{1,4}X_{1,5}}+\frac{1}{\hat{X}_{3,5}X_{1,5}X_{2,5}}\,,
\end{align}
where $\hat{X}_{1,3}=c_{2,5,6}+\hat{X}_{3,5}+X_{1,5}-X_{2,5}$. Now the $\mathcal{O}(1/z)$ order corresponds to two unbounded one dimensional geometries. There are no cancellations at the $\mathcal{O}(1/z)$ order, 
\begin{align}\label{eq:badscheme}
    A_6\big|_{\mathcal{O}(1/z)}\sim \frac{1}{z}\frac{1}{X_{3,5}X_{1,5}}\left(\frac{1}{X_{1,4}}+\frac{1}{X_{2,5}}\right),
\end{align}
which leads to a pole at infinity.

The example above shows that a good basis and shift choice is crucial for the recursion. Before moving to the most general cases, we first consider amplitudes with only $(\mathfrak{a})\mathfrak{f}$ states where all  $\mathfrak{f}$-$\mathfrak{af}$ pairs are adjacent,
\begin{align}
    \alpha=[(1,2),(3,4),\ldots,(n-1,n)]\,.
\end{align}
The preferred basis choice is
\begin{align}\label{eq:adjacentBasis}
    Y&=\big(1,\{X_{2i-1,2i+1}\},\{X_{3,2j+1}\}\big) \qquad (1\leqslant i\leqslant n/2\,,\quad 3\leqslant j\leqslant n/2-1)\nonumber\\
    &\equiv \big(1,X_{1,3},X_{3,5},\ldots,X_{n-1,1},X_{3,7},X_{3,9},\ldots,X_{3,n-1}\big)\,.
\end{align}
Physically, $\{X_{2i-1,2i+1}\}$ corresponds to a partial triangulation that gives an $(n/2)$-point $\mathfrak{adj}$ sub-amplitude. Applying the basis in section~\ref{closedassrec} to this sub-amplitude gives $\{X_{3,2i+1}\}$, the second part of $Y$. Such a basis at eight points is illustrated in figure~\ref{fig:triangulation8p}. The exact linear dependence of the other planar variables on this basis is given in eq.~\eqref{eq:adjacentConstraints}.

\begin{figure}[t]
\centering
\begin{tikzpicture}[every node/.style={font=\footnotesize,},dir/.style={decoration={markings, mark=at position \halfway with {\arrow{Latex}}},postaction={decorate}}]
\coordinate (p1) at (0:2);
\coordinate (p2) at (-45:2);
\coordinate (p3) at (-90:2);
\coordinate (p4) at (-135:2);
\coordinate (p5) at (-180:2);
\coordinate (p6) at (-225:2);
\coordinate (p7) at (-270:2);
\coordinate (p8) at (-315:2);
\node (p1) at (0:2) [vertex,label={[label distance=-2pt]0:{$1$}}] {};
\node (p2) at (-45:2) [vertex,label={[label distance=-2pt]-45:{$2$}}] {};
\node (p3) at (-90:2) [vertex,label={[label distance=-2pt]-90:{$3$}}] {};
\node (p4) at (-135:2) [vertex,label={[label distance=-2pt]-135:{$4$}}] {};
\node (p5) at (-180:2) [vertex,label={[label distance=-2pt]-180:{$5$}}] {};
\node (p6) at (-225:2) [vertex,label={[label distance=-2pt]-225:{$6$}}] {};
\node (p7) at (-270:2) [vertex,label={[label distance=-2pt]-270:{$7$}}] {};
\node (p8) at (-315:2) [vertex,label={[label distance=-2pt]-315:{$8$}}] {};
\path (p1.center) -- (p2.center) node (f1) [pos=0.5] {};
\path (p2.center) -- (p3.center) node (f2) [pos=0.5] {};
\path (p3.center) -- (p4.center) node (f3) [pos=0.5] {};
\path (p4.center) -- (p5.center) node (f4) [pos=0.5] {};
\path (p5.center) -- (p6.center) node (f5) [pos=0.5] {};
\path (p6.center) -- (p7.center) node (f6) [pos=0.5] {};
\path (p7.center) -- (p8.center) node (f7) [pos=0.5] {};
\path (p8.center) -- (p1.center) node (f8) [pos=0.5] {};
\draw [thick,gray] (f1.center) to [bend right=20] (f2.center) (f3.center) to [bend right=20] (f4.center) (f5.center) to [bend right=20] (f6.center) (f7.center) to [bend right=20] (f8.center);
\draw [thick] (p1.center) -- (p2.center) -- (p3.center) -- (p4.center) -- (p5.center) -- (p6.center) -- (p7.center) -- (p8.center) -- cycle;
\draw [thick,blue] (p1.center) -- (p3.center) ;
\draw [thick,blue] (p3.center) -- (p5.center) ;
\draw [thick,blue] (p5.center) -- (p7.center) ;
\draw [thick,blue] (p1.center) -- (p7.center) ;
\draw [thick,red] (p3.center) -- (p7.center) ;
\end{tikzpicture}
\caption{An example of the basis $Y$ for $\alpha=[(1,2),(3,4),(5,6),(7,8)]$. The gray curves are flavor lines. The blue diagonals correspond to $\{X_{2i-1,2i+1}\}$, which gives the partial triangulation that isolates the maximal $\mathfrak{adj}$ sub-amplitude. The red diagonals correspond to $\{X_{3,2j+1}\}$, which gives the triangulation of the resultant sub-polygon associated with the maximal $\mathfrak{adj}$ sub-amplitude. }
\label{fig:triangulation8p}
\end{figure}
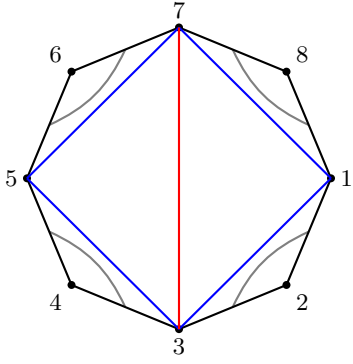

\indent We now repeat the derivation of section~\ref{closedassrec}, except using the planar basis in eq.~\eqref{eq:adjacentBasis} and the shift
\begin{equation}
    X_{3,5}\rightarrow\hat{X}_{3,5}=zX_{3,5}\,.
\end{equation}
We denote the set of deformed propagators under the shift $X_{3,5}\rightarrow zX_{3,5}$ as $F$, and the propagators themselves as $X_{F_{i}}$. For amplitudes with adjacent $\mathfrak{f}$-$\mathfrak{af}$ pairs only, $F$ is given by 
\begin{align}\label{eq:Fspecial}
    F=F^{(1)}\cup F^{(2)}\,,\quad \text{where} & & F^{(1)}=\{X_{3,5},X_{2,5}\}\bigcup_{i=3}^{n/2}\{X_{3,2i}\}\,,& & F^{(2)}=\bigcup_{j=1}^{n/2}\{X_{4,2i-1}\}\,.
\end{align}
The set $F$ for $n=8$ is shown in figure~\ref{fig:examplerestrict} as an example. Unlike the bi-adjoint case, not all variables depend on $X_{3,5}$ in the same way. We denote the prefactor of $X_{3,5}$ as $\lambda_{F_i}$,
\begin{equation}
    X_{F_i}=\lambda_{F_i}X_{3,5}+\ldots\,.
\end{equation}
In particular, we have
\begin{equation}\label{eq:lambdaadjacent}
    \lambda_{F_i}=\left\{\begin{array}{lll}
    1 & \quad\quad & X_{F_i}\in F^{(1)} \\
    -1 & \qquad & X_{F_i}\in F^{(2)}
    \end{array}\right.\,,
\end{equation}
which can be read off from eq.~\eqref{eq:adjacentConstraints}. The counterparts of eq.~\eqref{genericshift} and~\eqref{zival} are
\begin{gather}
    \hat{X}_{F_i}(z)=-\lambda_{F_i}(1-z)X_{3,5}+X_{F_i}\,, \nonumber\\
    z_{F_i}=1-\frac{X_{F_i}}{\lambda_{F_i}X_{3,5}}\quad\text{and}\quad \hat{X}_{F_j}(z_{F_i})=X_{F_j}-\frac{\lambda_{F_j}}{\lambda_{F_i}}X_{F_i}\,.
\end{gather}
Now we consider the contour integral that reproduces the amplitude,
\begin{align}
    A_n[(1,2),(3,4),\ldots,(n-1,n)]&=\oint_{|z-1|=\epsilon} \frac{dz}{2\pi i} \frac{z}{z-1}A_n[zX_{3,5}] \nonumber\\
    &= -\Bigg ( \textrm{Res}_{\infty}+\sum_{z_{F_i}\in\,\textrm{finite poles}} \textrm{Res}_{z_{F_i}}\Bigg ) \frac{z}{z-1}A_n[zX_{3,5}] \,.
\end{align}
The finite residue contribution is now
\begin{equation}
\begin{split}
-\,\textrm{Res}_{z_{i}}\left ( \frac{z}{z-1}A[zX_{3,5}] \right )&=\left ( \frac{1}{X_{F_{i}}}-\frac{1}{\lambda_{F_i}X_{3,5}}\right )\hat{A}_{L}(z_{F_i})\times \hat{A}_{R}(z_{F_i}) \,.
\end{split}
\label{contrifinitres}
\end{equation}
Very importantly, since the $X_{F_i}\in F$ are $(\mathfrak{a})\mathfrak{f}$ factorization channels, both $A_L$ and $A_R$ are lower points $(\mathfrak{a})\mathfrak{f}$ amplitudes with adjacent $\mathfrak{f}$-$\mathfrak{af}$ pairs. The recursion scheme is thus closed.

\paragraph{No Pole at Infinity} 
The goal is to show that under the particular basis choice~\eqref{eq:adjacentBasis} and the shift $X_{3,5}\rightarrow zX_{3,5}$, every Feynman diagram that scales as $\mathcal{O}(1/z)$ has a companion Feynman diagram which shares the same $(n{-}4)$ unshifted propagators and allowed by flavor symmetry. These two diagrams will cancel each other in the large-$z$ limit, as sketched in eq.~\eqref{sumfeydiacontr}. Geometrically, this is equivalent to proving that all edges created by the intersections of $(n{-}4)$ unshifted facets are bounded. Unlike the bi-adjoint case, this cancellation does not necessarily occur, as the companion Feynman diagrams could be forbidden by flavor conservation. This would happen had we used a ``bad'' basis. Eq.~\eqref{eq:badscheme} is such an example.
At low multiplicities, we can check explicitly the residue at infinity. For our eight-point example shown in figure~\ref{fig:examplerestrict}, the Feynman diagrams at $\mathcal{O}(1/z)$ are proportional to either $\frac{1}{\hat{X}_{3,5}}+\frac{1}{\hat{X}_{1,4}}$ or $\frac{1}{\hat{X}_{3,5}}+\frac{1}{\hat{X}_{4,7}}$, which in fact vanishes as $1/z^2$ at $z\rightarrow\infty$.
We defer a rigorous proof that actually works for a more generic color ordering to section~\ref{almostgenericorderingss}.

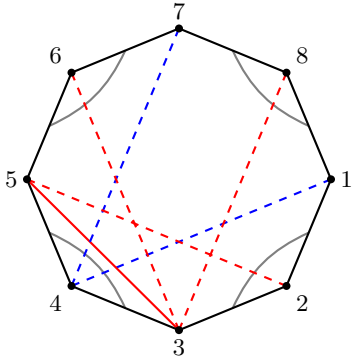
\begin{figure}[t]
\centering
\begin{tikzpicture}[every node/.style={font=\footnotesize,},dir/.style={decoration={markings, mark=at position \halfway with {\arrow{Latex}}},postaction={decorate}}]
\coordinate (p1) at (0:2);
\coordinate (p2) at (-45:2);
\coordinate (p3) at (-90:2);
\coordinate (p4) at (-135:2);
\coordinate (p5) at (-180:2);
\coordinate (p6) at (-225:2);
\coordinate (p7) at (-270:2);
\coordinate (p8) at (-315:2);
\path (p1.center) -- (p2.center) node (f1) [pos=0.5] {};
\path (p2.center) -- (p3.center) node (f2) [pos=0.5] {};
\path (p3.center) -- (p4.center) node (f3) [pos=0.5] {};
\path (p4.center) -- (p5.center) node (f4) [pos=0.5] {};
\path (p5.center) -- (p6.center) node (f5) [pos=0.5] {};
\path (p6.center) -- (p7.center) node (f6) [pos=0.5] {};
\path (p7.center) -- (p8.center) node (f7) [pos=0.5] {};
\path (p8.center) -- (p1.center) node (f8) [pos=0.5] {};
\draw [thick,gray] (f1.center) to [bend right=20] (f2.center) (f3.center) to [bend right=20] (f4.center) (f5.center) to [bend right=20] (f6.center) (f7.center) to [bend right=20] (f8.center);
\draw [thick,dashed,blue] (p1.center) -- (p4.center) ;
\draw [thick,red] (p3.center) -- (p5.center) ;
\draw [thick,dashed,blue] (p4.center) -- (p7.center) ;
\draw [thick,red,dashed] (p3.center) -- (p6.center) (p3.center) -- (p8.center) (p2.center) -- (p5.center);
\node (p1) at (0:2) [vertex,label={[label distance=-2pt]0:{$1$}}] {};
\node (p2) at (-45:2) [vertex,label={[label distance=-2pt]-45:{$2$}}] {};
\node (p3) at (-90:2) [vertex,label={[label distance=-2pt]-90:{$3$}}] {};
\node (p4) at (-135:2) [vertex,label={[label distance=-2pt]-135:{$4$}}] {};
\node (p5) at (-180:2) [vertex,label={[label distance=-2pt]-180:{$5$}}] {};
\node (p6) at (-225:2) [vertex,label={[label distance=-2pt]-225:{$6$}}] {};
\node (p7) at (-270:2) [vertex,label={[label distance=-2pt]-270:{$7$}}] {};
\node (p8) at (-315:2) [vertex,label={[label distance=-2pt]-315:{$8$}}] {};
\draw [thick] (p1.center) -- (p2.center) -- (p3.center) -- (p4.center) -- (p5.center) -- (p6.center) -- (p7.center) -- (p8.center) -- cycle;
\end{tikzpicture}
\caption{The set $F$ for $\alpha=[(1,2),(3,4),(5,6),(7,8)]$. The shifted propagator $X_{3,5}$ is represented by a red solid line and the other elements in $F^{(1)}$ are given in red dashed lines. The elements of $F^{(2)}$ are shown with blue dashed lines.}
\label{fig:examplerestrict}
\end{figure}

\paragraph{Example:} 
We now finish the full recursion of $A_{6}[(1,2)(3,4)(5,6)]$. The basis and shifted variable are given in eq.~\eqref{chosenbasisss}. According to eq.~\eqref{eq:Fspecial} and~\eqref{eq:lambdaadjacent}, the shifted variables and the $\lambda$-list are  

\begin{equation}
F=\{ X_{3,5}, X_{2,5},X_{3,6},X_{1,4}\} \,, \qquad \lambda=\{1,1,1,-1\}\,.   
\label{expshiftssixpoint}
\end{equation}

\noindent We sum over the contribution of each term in eq. (\ref{expshiftssixpoint}), finding

\begin{align}
A_{6}&=\sum_{X_{F_{i}}\in F} \left ( \frac{1}{X_{F_{i}}}-\frac{1}{\lambda_{F_i}X_{3,5}}\right )\hat{A}_{L}(z_i)\times \hat{A}_{R}(z_i) \nonumber\\
&=\left (\frac{1}{X_{2,5}}-\frac{1}{X_{3,5}}\right ) \frac{1}{\hat{X}_{3,5}(z_{{2,5}})X_{1,5}}+\left (\frac{1}{X_{3,6}}-\frac{1}{X_{3,5}}\right )\frac{1}{X_{1,3}\hat{X}_{3,5}(z_{{3,6}})} \nonumber\\
&\quad +\left (\frac{1}{X_{1,4}}+\frac{1}{X_{3,5}}\right )\frac{1}{X_{1,5}X_{1,3}} \nonumber \\
&=\left (\frac{1}{X_{2,5}}-\frac{1}{X_{3,5}}\right ) \frac{1}{(X_{3,5}-X_{2,5})X_{1,5}}+\left (\frac{1}{X_{3,6}}-\frac{1}{X_{3,5}}\right )\frac{1}{(X_{3,5}-X_{3,6})X_{1,3}}   \nonumber\\
&\quad +\left (\frac{1}{X_{1,4}}+\frac{1}{X_{3,5}}\right )\frac{1}{X_{1,5}X_{1,3}}  \nonumber \\
&=\frac{1}{X_{2,5}X_{3,5}X_{1,5}}+\frac{1}{X_{3,6}X_{3,5}X_{1,3}}+\frac{1}{X_{1,3}X_{1,4}X_{1,5}}+\frac{1}{X_{1,3}X_{3,5}X_{1,5}}
\label{exmaple6solution}
\end{align}

\noindent Note that the contribution from the $X_{3,5}$ facet vanishes trivially. In the intermediate steps, the spurious boundary components $X_{3,5}-X_{2,5}=0$ and $X_{3,5}-X_{3,6}=0$ introduce a partial triangulation to the open associahedron, and they cancel in the final result as expected.

\subsection{Preferred planar bases for open associahedra}
\label{basisarborder}

We now turn to deriving a generalization of the recursion in section~\ref{closedassrec} to generic open associahedra. The first step to deriving a recursion relation for any positive geometries is choosing a basis $Y$ and shift variables. As we have seen before, a suitable basis is crucial for the success of the recursion. 

Consider an ordering $\alpha=[\mathsf{B}_1,\mathsf{B}_2,\mathsf{B}_3,\ldots,\mathsf{B}_m]$ in which the first block $\mathsf{B}_1=(l_1,r_1)$ while the rest are generic and possibly contain sub-blocks, $\mathsf{B}_i=(l_i,\mathsf{B}_{i_1},\mathsf{B}_{i_2},\ldots,\mathsf{B}_{i_s},r_i)$. The partial triangulation by the diagonals
$\bigcup_{i=1}^{m}\{X_{l_i,r_i}\}$ 
gives us a $2m$-point adjacent pair $\mathfrak{f}$-$\mathfrak{af}$ sub-amplitude of the ordering $[(l_1,r_1),(l_2,r_2),\ldots,(l_m,r_m)]$. Following the strategy in section~\ref{allfermpairsadj}, we can write down the preferred basis of this sub-amplitude,
\begin{align}\label{eq:triangulation1}
    \bigcup_{i=1}^{m}\{X_{l_i,l_{i+1}}\}\bigcup_{i=4}^{m}\{X_{l_2,l_i}\}=\{X_{l_1,l_2},X_{l_2,l_3},\ldots,X_{l_m,l_1}\}\cup\{X_{l_2,l_4},X_{l_2,l_5},\ldots,X_{l_2,l_m}\}\,,
\end{align}
cf. eq.~\eqref{eq:adjacentBasis}. Under this partial triangulation, now each block $\mathsf{B}_i$ looks locally like a lower-point ordering $[(r_i,l_i),\mathsf{B}_{i_1},\mathsf{B}_{i_2},\ldots,\mathsf{B}_{i_s}]$. We can thus obtain the full basis recursively as $Y=\big(1,f[(l_1,r_1),\mathsf{B}_2,\mathsf{B}_3,\ldots,\mathsf{B}_m]\big)$, where
\begin{align}
    f[(l_1,r_1),\mathsf{B}_2,\mathsf{B}_3,\ldots,\mathsf{B}_m]=\bigcup_{i=1}^{m}\{X_{l_i,l_{i+1}},X_{l_i,r_i}\}\bigcup_{i=4}^{m}\{X_{l_2,l_i}\}\bigcup_{i=2}^{m}f[(r_i,l_i),\mathsf{B}_{i_1},\mathsf{B}_{i_2},\ldots]\,.
    \label{finalrequaiotss}
\end{align}
The recursion terminates at $\mathfrak{adj}$ states and adjacent $\mathfrak{f}$-$\mathfrak{af}$ pairs, on which the $f$ function returns the empty set. 

\begin{figure}[t]
\centering
\begin{tikzpicture}[every node/.style={font=\footnotesize,},dir/.style={decoration={markings, mark=at position \halfway with {\arrow{Latex}}},postaction={decorate}}]
\node (p1) at (0:2) [vertex,label={[label distance=-2pt]0:{$1$}}] {};
\node (p2) at (-36:2) [vertex,label={[label distance=-2pt]-36:{$2$}}] {};
\node (p3) at (-72:2) [vertex,label={[label distance=-2pt]-72:{$3$}}] {};
\node (p4) at (-108:2) [vertex,label={[label distance=-2pt]-108:{$4$}}] {};
\node (p5) at (-144:2) [vertex,label={[label distance=-2pt]-144:{$5$}}] {};
\node (p6) at (-180:2) [vertex,label={[label distance=-2pt]-180:{$6$}}] {};
\node (p7) at (-216:2) [vertex,label={[label distance=-2pt]-216:{$7$}}] {};
\node (p8) at (-252:2) [vertex,label={[label distance=-2pt]-252:{$8$}}] {};
\node (p9) at (-288:2) [vertex,label={[label distance=-2pt]-288:{$9$}}] {};
\node (p10) at (-324:2) [vertex,label={[label distance=-2pt]-324:{$10$}}] {};
\path (p1.center) -- (p2.center) node (f1) [pos=0.5] {};
\path (p2.center) -- (p3.center) node (f2) [pos=0.5] {};
\path (p3.center) -- (p4.center) node (f3) [pos=0.5] {};
\path (p4.center) -- (p5.center) node (f4) [pos=0.5] {};
\path (p5.center) -- (p6.center) node (f5) [pos=0.5] {};
\path (p6.center) -- (p7.center) node (f7) [pos=0.5] {};
\path (p7.center) -- (p8.center) node (f8) [pos=0.5] {};
\path (p8.center) -- (p9.center) node (f9) [pos=0.5] {};
\path (p9.center) -- (p10.center) node (f10) [pos=0.5] {};
\path (p10.center) -- (p1.center) node (f11) [pos=0.5] {};
\draw [thick,gray] (f1.center) to [bend right=20] (f2.center) (f3.center) to [bend right=20] (f4.center) (f8.center) to [bend right=20] (f9.center) (f7.center) to [bend right=0] (f10.center) (f5.center) to [bend right=0] (f11.center);
\draw [thick,blue] (p1.center) -- (p3.center) ;
\draw [thick,blue] (p3.center) -- (p5.center) ;
\draw [thick,blue] (p5.center) -- (p1.center) ;
\draw [thick,blue] (p9.center) -- (p7.center) ;
\draw [thick,red] (p5.center) -- (p10.center) ;
\draw [thick,blue] (p6.center) -- (p10.center) ;
\draw [thick,red] (p9.center) -- (p6.center) ;
\draw [thick] (p1.center) -- (p2.center) -- (p3.center) -- (p4.center) -- (p5.center) -- (p6.center) -- (p7.center) -- (p8.center) -- (p9.center) -- (p10.center) -- cycle;
\end{tikzpicture}
\caption{An example of the preferred basis $Y$ for $\alpha=[(1,2),(3,4),(5,(6,(7,8),9),10)]$. The blue and red lines collectively correspond to the chosen basis. While the blue lines correspond to the partial triangulation which yields sub-polygons associated with purely $\mathfrak{adj}$ sub-amplitudes, the red lines correspond to the triangulations of the remaining four-gons.}
\label{fig:exmaplestriangulationss}
\end{figure}
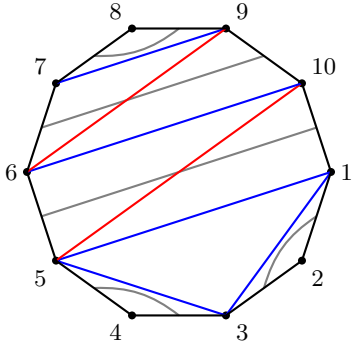

We now apply this basis recursion to some simple but nontrivial examples. First, we consider $\alpha=[(1,2),(3,4),(5,(6,7),8)]$. Applying eq.~\eqref{finalrequaiotss} to the out-most level of blocks, we get
\begin{align}
    f[(1,2),(3,4),(5,(6,7),8)]=\{X_{1,3},X_{3,5},X_{1,5},X_{5,8}\}\cup f[(8,5),(6,7)]\,.
\end{align}
Then $f[(8,5),(6,7)]=\{X_{6,8}\}$ gives the triangulation of the block $(5,(6,7),8)$. Together, we get the preferred basis for $\alpha=[(1,2),(3,4),(5,(6,7),8)]$,
\begin{align}
    Y=(1,X_{1,3},X_{3,5},X_{1,5},X_{5,8},X_{6,8}\}\,.
\end{align}
As our second example, we consider $\alpha=[(1,2),(3,4),(5,(6,(7,8),9),10)]$, for which the basis recursion gives
\begin{gather}
    f[(1,2),(3,4),(5,(6,(7,8),9),10)]=\{X_{1,3},X_{3,5},X_{1,5},X_{5,10}\}\cup f[(10,5),(6,(7,8),9)]\,,\nonumber\\
    f[(10,5),(6,(7,8),9)]=\{X_{6,10},X_{6,9}\}\cup f[(9,6),(7,8)]=\{X_{6,10},X_{6,9},X_{7,9}\}\,.
\end{gather}
The corresponding triangulation is shown in figure~\ref{fig:exmaplestriangulationss}.

\indent Given the above basis, one can find the dependence of all $X_{i,j}$ variables on the chosen shift variable. 
If there are at least three blocks at the out-most level ($m\geqslant 3$), we choose $X_{l_2,l_3}=X_{3,l_3}$ as the shift variable.
If $m{=}2$, we choose $X_{3,r_{2}}=X_{3,n}$ as the shift variable. 
With the preferred basis $Y$ and shift variable, we now give a complete recursion for generic bi-color amplitudes. 

\subsection{Recursion for $A_n[(1,2),\mathsf{B}_{2},\mathsf{B}_{3},\ldots, \mathsf{B}_{m}]$ with $m\geqslant 3$}
\label{almostgenericorderingss}

\noindent We now consider $\alpha=[(1,2),\mathsf{B}_{2},\mathsf{B}_{3},\ldots, \mathsf{B}_{m}]$ with $m\geqslant 3$. The case of $m=2$ is somewhat special and is covered in the next sub-section. The shift variable is $X_{3,l_{3}}$. For convenience, we separate the set of deformed propagators $F$ into three parts, $F=F^{(1)}\cup F^{(2)}\cup F^{(3)}$. The first two parts directly generalize eq.~\eqref{eq:Fspecial},
\begin{gather}
    F^{(1)}=\{ X_{3,l_{3}}\,,X_{2,l_{3}} \} \bigcup_{i=3}^{m} \{X_{3,r_{i}} \}\,,\quad F^{(2)}= \bigcup_{\substack{i=1 \\ i\neq 2}}^{m} \{ X_{r_{2},l_{i}} \}\,, \nonumber\\
    \lambda_{F_i}=\left\{\begin{array}{lll}
    1 & \quad\quad & X_{F_i}\in F^{(1)} \\
    -1 & \qquad & X_{F_i}\in F^{(2)}
    \end{array}\right.\,.
    \label{F1variabless}
\end{gather}
The $F^{(3)}$ part depends on the internal structure of the blocks,
\begin{gather}
    F^{(3)}=\bigcup_{\mathsf{I}\in \textrm{sub}[\mathsf{B}_{2}] } \{ X_{l_{\mathsf{I}},l_{3}} \}\bigcup_{i=3}^{m}\bigcup_{\mathsf{I}\in \textrm{sub}[\mathsf{B}_{i}] }\{ X_{3,l_{\mathsf{I}}} \}\,,\nonumber\\
    \lambda_{F_i}=\left\{\begin{array}{lll}
    1 & \quad\quad & X_{F_i}\in \{X_{l_{\mathsf{I}},l_3}\}\text{ and }\mathsf{I}\in\textrm{sub}[\mathsf{B}_2] \\
    \textrm{Loc}(\mathsf{I}) & \qquad & X_{F_i}\in \{X_{3,l_{\mathsf{I}}}\}\text{ and }\mathsf{I}\in\textrm{sub}[\mathsf{B}_{i\geqslant 3}]
    \end{array}\right.\,,
    \label{F3variabless}
\end{gather}
where $l_{\mathsf{I}}$ is again the left ($\mathfrak{a}$)$\mathfrak{f}$ state of the sub-block $\mathsf{I}\in \textrm{sub}[\mathsf{B}_{i}]$. The prefactor $\textrm{Loc}(\mathsf{I})$ corresponds to the location of $\mathsf{I}$ in $\mathsf{B}_{i}$ going from \emph{right to left},
\begin{align}
    \textrm{Loc}(\mathsf{B}_{i_\ell})=s-\ell+2\quad\text{ for }\quad\mathsf{B}_i=(l_i,\mathsf{B}_{i_1},\mathsf{B}_{i_2},\ldots,\mathsf{B}_{i_s},r_i)\,.
\end{align}
For example, given $\mathsf{B}_{i}=(l_{i},\mathsf{B}_{i_1},\mathsf{B}_{i_2},\mathsf{B}_{i_3},r_{i})$, we have $\textrm{Loc}(\mathsf{B}_{i_3})=2$, $\textrm{Loc}(\mathsf{B}_{i_2})=3$, and $\textrm{Loc}(\mathsf{B}_{i_1})=4$. We note that except for $X_{3,l_3}$, all the $X_{F_i}\in F$ crosses exactly one flavor line. This property will play an important role in proving that there is no pole at infinity. We can then write the recursion as
\begin{equation}
A_n=\sum_{X_{F_{i}}\in F}\left ( \frac{1}{X_{F_{i}}}-\frac{1}{\lambda_{F_i}X_{3,l_3}}\right )\hat{A}_{L}(z_{F_i})\times \hat{A}_{R}(z_{F_i}) \,,
\end{equation}
where $\hat{A}_{L/R}$ are evaluated with the deformed propagators
\begin{align}
    z_{F_i}=1-\frac{X_{F_i}}{\lambda_{F_i}X_{3,l_3}}\,,\qquad \hat{X}_{F_j}(z_{F_i})=X_{F_j}-\frac{\lambda_{F_j}}{\lambda_{F_i}}X_{F_i}\,.
\end{align}

\noindent We will first give an example, before giving a rigorous proof for the lack of pole at infinity.

\paragraph{Example: } 
We apply the recursion to $A_{8}[(1,2),(3,4),(5,(6,7),8)]$. The shift variable is $X_{3,5}\rightarrow zX_{3,5}$. According to eq.~\eqref{F1variabless} and~\eqref{F3variabless}, the set of deformed propagators and the corresponding $\lambda$-prefactors are

\begin{equation}
F=\{X_{3,5},\ X_{2,5},\ X_{3,8},\ X_{3,6},\ X_{1,4}\} \,,\quad
\lambda=\{ 1,\ 1, \ 1,\ 2,\ -1 \}\,.
\end{equation}

\noindent We apply eq. (\ref{contrifinitres}), finding 

\begin{align}
A_8&=\left (\frac{1}{X_{2,5}}-\frac{1}{X_{3,5}} \right ) \frac{1}{X_{1,5}X_{6,8}}\left(\frac{1}{X_{1,6}}+\frac{1}{X_{5,8}} \right )  \frac{1}{\hat{X}_{3,5}(z_{X_{2,5}})}   \nonumber \\
&\quad +\left (\frac{1}{X_{3,8}}-\frac{1}{X_{3,5}} \right )  \frac{1}{X_{1,3}\hat{X}_{3,5}(z_{X_{3,8}})X_{6,8}} \left ( \frac{1}{X_{5,8}}+\frac{1}{\hat{X}_{3,6}(z_{X_{3,8}})}\right )  \nonumber  \\
&\quad +\left (\frac{1}{X_{3,6}}-\frac{1}{2X_{3,5}} \right )  \frac{1}{\hat{X}_{3,5}(z_{X_{3,6}})X_{1,3}X_{6,8}}\left ( \frac{1}{\hat{X}_{3,8}(z_{X_{3,6}})}+\frac{1}{X_{1,6}}\right )   \nonumber  \\
&\quad +\left (\frac{1}{X_{1,4}}+\frac{1}{X_{3,5}} \right )  \frac{1}{X_{1,3}X_{6,8}X_{1,5}} \left ( \frac{1}{X_{5,8}}+\frac{1}{X_{1,6}}\right )     \,.
\label{8pointexamplesss}
\end{align}
The factor of $1/2$ in the third line is due to $\lambda_{X_{3,6}}=2$. After simplifying eq.~\eqref{8pointexamplesss}, one finds the nine Feynman diagrams that contribute to the amplitude,
\begin{align}
    A_8&=\frac{1}{X_{1, 3} X_{1, 4} X_{1, 5} X_{1, 6} X_{6, 8}} + \frac{1}{
 X_{1, 3} X_{1, 5} X_{1, 6} X_{3, 5} X_{6, 8}} + \frac{1}{
 X_{1, 5} X_{1, 6} X_{2, 5} X_{3, 5} X_{6, 8}} \nonumber\\
 &\quad + \frac{1}{
 X_{1, 3} X_{1, 6} X_{3, 5} X_{3, 6} X_{6, 8}} + \frac{1}{
 X_{1, 3} X_{3, 5} X_{3, 6} X_{3, 8} X_{6, 8}} + \frac{1}{
 X_{1, 3} X_{1, 4} X_{1, 5} X_{5, 8} X_{6, 8}} \nonumber\\
 &\quad + \frac{1}{
 X_{1, 3} X_{1, 5} X_{3, 5} X_{5, 8} X_{6, 8}} + \frac{1}{
 X_{1, 5} X_{2, 5} X_{3, 5} X_{5, 8} X_{6, 8}} + \frac{1}{
 X_{1, 3} X_{3, 5} X_{3, 8} X_{5, 8} X_{6, 8}}\,.
\end{align}
Without working out all the algebra, we can see the cancellation of spurious poles characteristic of BCFW-like recursions. For example, the spurious pole $X_{3,8}-\frac{1}{2}X_{3,6}$ appears in two terms,
\begin{equation}
\frac{1}{X_{1,3}X_{3,5}X_{6,8}}\left[\frac{1}{X_{3,8}(X_{3,6}-2X_{3,8})}+\frac{1}{X_{3,6}(X_{3,8}-\frac{1}{2}X_{3,6})} \right]=\frac{1}{X_{1,3}X_{3,5}X_{6,8}X_{3,8}X_{3,6}}\,.  
\end{equation}
It indeed cancels upon summing both terms, as expected from a BCFW-like recursion.

\paragraph{No Pole at Infinity} 
We will give a proof by contradiction. For there to be a pole at infinity, there must be an unbounded edge defined by the intersection of $(n-4)$ undeformed  facets. We denote $X_{j,l}$ and $X_{i,k}$ as the two planar variables associated with the edge that is compatible with such $(n-4)$ facets.\footnote{We define for convenience $i<j<k<l$ in the cyclic sense.} For this edge to be unbounded, $X_{j,l}$ must be forbidden by flavor conservation (namely, cross two flavor lines) and $X_{i,k}$ must be an element of $F$ (or vice versa). If $X_{i,k}=X_{3,l_3}$, then we must have $j=r_2$ or $j=l_{\mathsf{I}}$ with $\mathsf{I}\in\textrm{sub}[\mathsf{B}_2]$ such that $X_{i,j}$ and $X_{j,k}$ cross only one flavor line. For $l$, we have the following possibilities,
\begin{itemize}
    \item if $l=r_1=2$, then $X_{k,l}=X_{2,l_3}\in F$;
    \item if $l=r_i$ with $3\leqslant i\leqslant m$, then $X_{i,l}=X_{3,r_i}\in F$;
    \item if $l=l_{\mathsf{I}}$ with $\mathsf{I}\in\textrm{sub}[\mathsf{B}_{i}]$ and $3\leqslant i\leqslant m$, then $X_{i,l}=X_{3,l_{\mathsf{I}}}\in F$;
    \item for all the other choices, $X_{j,l}$ would not be forbidden by flavor conservation. 
\end{itemize}
Therefore, at least one of $X_{i,l}$ and $X_{k,l}$ must be an element of $F$, which contradicts with the assumption that both belong to the intersecting $(n-4)$ undeformed facets. Next, we consider the case that $X_{i,k}$ is an element of $F$ other than $X_{3,l_3}$, then $X_{i,k}$ must cross one flavor line according to the form of $F$. On the other hand, $X_{j,l}$ should cross two flavor lines. This leads to the contradiction that either $X_{k,l}$ or $X_{i,l}$ must cross two flavor lines. Therefore, we can now conclude that the intersection of $(n-4)$ undeformed facets being an unbounded edge cannot happen, and thus there is no pole at infinity.

\subsection{Recursion for $A_n[(1,2),\mathsf{B}_{2}]$}
\label{recurmionmone}

The final type of ordering to consider is  $\alpha=[(1,2),\mathsf{B}_{2}]$, so there is only one $\mathfrak{f}$-$\mathfrak{af}$ block other than $(1,2)$. Now by cyclicity we can define $l_3=l_1=1$, but if we still shift $X_{3,l_3}=X_{1,3}$, there will be a pole at infinity. For this case we instead choose the shift variable as $X_{3,r_{2}}=X_{3,n}$. The set of deformed variables is simply

\begin{equation}
F=\{ X_{3,n} \} \bigcup_{\mathsf{I}\in \textrm{sub}[\mathsf{B}_{2}]} \{ X_{1,l_{\mathsf{I}}}, \ X_{3,l_{\mathsf{I}}} \} \,,\quad
\lambda_{F_i}=\left\{\begin{array}{lll}
 -1 & \quad & X_{F_i} \in \{ X_{1,l_{\mathsf{I}}} \} \\ 
 1 & & X_{F_i} \in \{X_{3,n}, X_{3,l_{\mathsf{I}}} \ \}
\end{array}\right.\,.
\label{shiftedvariablessspeciacase}
\end{equation}

\noindent All $X_{F_i}\in F$, including the shift variable $X_{3,n}$, cross a flavor line so that the proof in section~\ref{allfermpairsadj} for no pole at infinity applies here. 

\paragraph{Example: } 
We apply the recursion procedure to $A_{6}[(1,2),(3,(4,5),6)]$. The shift variable is $X_{3,6}$. From eq.~\eqref{shiftedvariablessspeciacase}, the deformed propagators are 
\begin{equation}
\begin{split}
F&=\{ X_{1,4},X_{3,6} \} \ , \quad \lambda=\{ -1,1\} \ .
\end{split}
\end{equation}
Then we apply eq. (\ref{contrifinitres}), which gives 
\begin{align}
A_6&=\sum_{X_{F_{i}}\in F} \left ( \frac{1}{X_{F_{i}}}-\frac{1}{\lambda_{F_i}X_{3,6}}\right )\hat{A}_{L}(z_{F_i})\times \hat{A}_{R}(z_{F_i})=\left (\frac{1}{X_{1,4}}+\frac{1}{X_{3,6}} \right ) \frac{1}{X_{1,3}X_{4,6}} \,.
\end{align}
Here, the result is solely contributed by the $X_{1,4}$ channel.

\section{Duality between pullback and Melia decomposition}
\label{sec:proofpullback}

Finally, we move on to color-dressed amplitudes and study the color-kinematics duality in the geometric picture. We consider an arbitrary color-dressed amplitude in a gauge theory with $k$ pairs of $(\mathfrak{a})\mathfrak{f}$ states:
\begin{equation}
\textbf{M}_{n,k}=\sum_{\text{cubic } g}C(g|\alpha_{g})N(g|\alpha_{g})\prod_{I\in g}\frac{\vartheta_{I}}{s_{I}}\,,
\label{colordressedamplitude}
\end{equation}
where $C(g|\alpha_g)$ and $N(g|\alpha_g)$ are the color factor and kinematic numerator associated with the cubic graph $g$. Both of them are written with respect to an ordering $\alpha_g$ of external particles that is compatible with the graph. Different choices will give rise to additional signs to both quantities but the product is unchanged. Note that the function $\vartheta_I$, defined in eq.~\eqref{eq:thetaI}, removes all graphs that violates flavor conservation. The color factor $C(g|\alpha_{g})$ is a product of structure constants $f^{abc}$ and fundamental generators $(T^{a})_{\underline{i}}^{\bar{j}}$ at each vertex of $g$, depending on the interaction type. Specifying $\alpha_g$ fixes the ordering of external particles and thus determines $C(g|\alpha_g)$ without ambiguities.
We can write the color factors as linear combinations of those in the Melia basis, and the kinematic coefficients now become the color-ordered amplitudes $M_n[N,\beta]$,
\begin{equation}
\begin{split}
\mathbf{M}_{n,k}&=\sum_{\beta\in \text{Melia}} K[\beta] M_{n}[N;\beta] \ , \\
 M_{n}[N;\beta]&=\sum_{\beta\textrm{-planar } g}N(g|\beta)\prod_{I\in g} \frac{\vartheta_{I}}{s_{I}} \,.
\end{split}
\end{equation}
The exact form of $K[\beta]$ is unimportant to us and we refer interested readers to Ref.~\cite{Johansson:2015oia} for more explicit definitions, and Ref.~\cite{Ochirov:2019mtf} for a factorization based recursive construction. In the bi-color scalar theory, the kinematic numerator is replaced by another color factor $\tilde{K}[\beta]$. Therefore, the color decomposition gives the double partial amplitudes $m_n[\alpha|\beta]$,
\begin{equation}
\begin{split}
\textbf{M}_{n,k}^{\textrm{bi-color}}&=\sum_{\alpha,\beta\in \text{Melia}}K[\alpha]\tilde{K}[\beta]m_n[\alpha|\beta] \ , \\
m_n[\alpha|\beta]&=\sum_{\alpha\textrm{ and }\beta\textrm{-planar }g}\prod_{I\in g}\frac{\vartheta_{I}}{s_{I}} \ ,
\end{split}
\end{equation}
which consist of the Feynman diagrams that are compatible with both $\alpha$ and $\beta$. While we have studied the diagonal components in the previous sections, the off-diagonal $m[\alpha|\beta]$ is simply a product of lower points diagonal sub-amplitudes.\footnote{This is the generalization of the same statement for bi-adjoint double partial amplitudes~\cite{Cachazo:2013iea}.} 

In this section, we study the color-kinematics duality in the context of amplitudes as canonical forms. In particular, we will prove the pullback conjecture first given in Ref.~\cite{Herderschee:2019wtl} using our results on the geometry of facets of open associahedra.

\subsection{Duality between color and kinematics}
By including the kinematic numerators, we can write down the full (anti-)fundamental scattering form dual to the amplitude~\eqref{colordressedamplitude},
\begin{equation}
\Omega_{n}^{k}[N]=
\sum_{\text{cubic }g}N(g|\alpha_{g})W(g|\alpha_{g}) \prod_{I\in g} \frac{\vartheta_{I}}{s_{I}}\,,
\label{scatterinfoss}
\end{equation}
where $W(g|\alpha_{g})$ is a differential form,
\begin{equation}
W(g|\alpha_g)=\textrm{sign}(g|\alpha_{g})\bigwedge_{I\in g}ds_{I} \,.
\label{dlogforgraph}
\end{equation} 
The function $\textrm{sign}(g|\alpha_{g})$ resolves the ambiguity in the ordering of the differentials in eq.~\eqref{dlogforgraph}. For two graph-ordering pairs $(g|\alpha_g)$ and $(g'|\alpha'_g)$ relative sign is given by
\begin{equation}
\textrm{sign}(g|\alpha_{g})=(-1)^{\textrm{mutate}(g|g')}(-1)^{\textrm{flip}(\alpha_{g},\alpha_{g}')}\textrm{sign}(g'|\alpha_{g}') \ ,
\end{equation}
where $\textrm{mutate}(g,g')$ is the number mutations one must perform to relate graphs $g$ and $g'$,\footnote{A mutation is defined as taking a four-point sub-graph and exchanging the $s$-channel with $t$-channel or vice-versa. Two graphs related by a mutation differ by only one propagator. If $S$ and $T$ are the two different propagators in $g$ and $g'$ related by a mutation, we define $\bigwedge_{I\in g'}ds_I\equiv \bigwedge_{I\in g}ds_I\big|_{dS\rightarrow dT}$.} while $\textrm{flip}(\alpha_{g},\alpha_{g}')$ is the number of vertex flips necessary to relate $\alpha_{g}$ to $\alpha_{g}'$. More detailed discussion of $\textrm{sign}(g|\alpha_{g})$ is given in Ref.~\cite{Arkani-Hamed:2017mur}.

We further restrict the scattering form~\eqref{scatterinfoss} to the (anti-)fundamental small kinematic space $\mathcal{K}_n^{k}$, which is spanned by $s_I$ variables that satisfy the following conditions, 
\begin{itemize}
    \item $I\subset \{1,2,\ldots,n \}$ is a subset of particles,
    \item $s_{\bar{I}}=s_{I}$ if $\bar{I}$ is the complement of $I$, 
    \item $s_{I}=0$ for $|I|=0,1,n-1,n$,
    \item $s_{I}=b_{I}$ if $\vartheta_I=0$, namely, forbidden by flavor symmetry,
    \item The seven-term identity $s_{I_{1},I_{2}}+s_{I_{2},I_{3}}+s_{I_{1},I_{3}}=s_{I_{1}}+s_{I_{2}}+s_{I_{3}}+s_{I_{4}}$ is applied to all the four-set partitions $I_1\cup I_2\cup I_3\cup I_4$ \emph{except for} the cases that all $I_{i}$ correspond to ($\mathfrak{a}$)$\mathfrak{f}$ states. \ .
\end{itemize}
The seven-term identity is crucial to show that the differential $W(g|\alpha_g)$ in the scattering form~\eqref{scatterinfoss} obeys the same Jacobi identities as the color factor $C(g|\alpha_g)$ in the color-dressed amplitude~\eqref{colordressedamplitude}, 
\begin{equation}
C(g|\alpha_{g}) \leftrightarrow W(g|\alpha_{g}) \,.
\end{equation}
Thus the duality between color-dressed amplitudes and scattering forms generalizes to the bi-color theory. The dimension of $\mathcal{K}_n^k$ is $\frac{n(n-3)}{2}-\frac{k(k-1)}{2}$.

In Ref.~\cite{Arkani-Hamed:2017mur}, the duality is further developed by showing that the color decomposition of bi-adjoint amplitudes is dual to the pullback of the scattering form to a subspace of small kinematic space. In Ref.~\cite{Herderschee:2019wtl}, this duality has been conjectured to the scattering forms with $(\mathfrak{a})\mathfrak{f}$ states. The goal is to properly pull $\Omega_n^k[N]$ back to the subspace $H_n[\alpha]$ and show that this process is dual to a Melia basis color decomposition. However, the constraints $\pmb{H}_n[\alpha]$ given in section~\ref{sec:invsoftconstruction} alone is inadequate for this purpose. The reason is that $\pmb{H}_n[\alpha]$ is imposed on the kinematic space $\mathcal{K}_n^k[\alpha]$, which is in general a subspace of $\mathcal{K}_n^k$. We call the complement space of $\mathcal{K}_n^k[\alpha]$ in $\mathcal{K}_n^k$ as $\mathcal{D}_n^k[\alpha]$, which consists of Mandelstam variables that respect flavor symmetry but are non-planar under $\alpha$. To get rid of these additional degrees of freedom, we can find a basis of $\mathcal{D}_n^k[\alpha]$ and set them to constants. However, this process is more complicated than it sounds, since $\mathcal{K}_n^k[\alpha]$ and $\mathcal{D}_n^k[\alpha]$ do not form a direct product. A change of basis in $\mathcal{D}_n^k[\alpha]$ in general will modify the differentials living in $\mathcal{K}_n^k[\alpha]$. Thus a successful pullback hinges on a good choice of the basis of $\mathcal{D}_n^k[\alpha]$.

We call the set of constraints in $\mathcal{D}_n^k[\alpha]$ as $\pmb{H}_n^A[\alpha]$ such that the full set of constraints is $\pmb{H}_n^T[\alpha]=\pmb{H}_n[\alpha]\cup\pmb{H}_n^A[\alpha]$. We expect that upon taking the pullback of any $W(g|\kappa)$ to the subspace $H_n^T[\alpha]$, we have 
\begin{equation}
W(g|\kappa)\big|_{H_{n}^{T}[\alpha]}=\left\{\begin{matrix}
(-1)^{\textrm{flip}(\kappa,\alpha)}dV[\alpha] &\hphantom{aa} & \textrm{if }g\textrm{ is compatible with $\alpha$} \\ 
0 &\quad & \textrm{otherwise}
\end{matrix}\right. \,,
\label{conjectdiffofr} 
\end{equation}
such that pulling back eq.~\eqref{scatterinfoss} to $H_{n}^{T}[\alpha]$ is equivalent to removing all the Feynman graphs incompatible with $\alpha$ ordering.
If $\pmb{H}_n^A[\alpha]$ totally localizes $\mathcal{D}_n^k[\alpha]$, then the corresponding subspace $H_n^T[\alpha]=H_n[\alpha]$. Ref.~\cite{Herderschee:2019wtl} proposed that $\pmb{H}_n^A[\alpha]$ is given by
\begin{align}\label{setrestriequss}
    \pmb{H}_n^A[\alpha]&=\left\{-s_{\mathsf{I},\mathsf{I}'}=c_{\mathsf{I},\mathsf{I}'}\middle|\begin{array}{cc}
    \mathsf{I}\cap\mathsf{I}'\neq\emptyset\text{ are (sub-)blocks of }\alpha\text{ other than }(1,2) \\
    \text{that the diagonal }(l_{\mathsf{I}},l_{\mathsf{I}'})\text{ crosses at least two flavor lines}
    \end{array}\right\}\nonumber\\
    &\quad\cup\left\{-s_{\mathsf{I},r_{\mathsf{I}'}}=c_{\mathsf{I},r_{\mathsf{I}'}}\middle|\begin{array}{cc}
    \mathsf{I}\cap\mathsf{I}'\neq\emptyset\text{ are (sub-)blocks of }\alpha\text{ other than }(1,2) \\
    \text{that the diagonal }(l_{\mathsf{I}},r_{\mathsf{I}'})\text{ crosses at least one flavor line} \\
    \text{other than the one associated with $r_{\mathsf{I}'}$}
    \end{array}\right\}\,.
\end{align}
We note that in case $\mathsf{I}$ and/or $\mathsf{I}'$ is an $\mathfrak{adj}$ state, $l_{\mathsf{I},\mathsf{I}'}$ is the $\mathfrak{adj}$ particle itself. In general, $H_n^T[\alpha]$ contains $H_n[\alpha]$ as a subspace, and $dV[\alpha]$ is not a top-form. However, it involves the differentials in $H_n[\alpha]$ only. Equivalently speaking, the appropriate bases of $\mathcal{D}_n^k[\alpha]$ for the pullback must contain the Mandelstam variables in $\pmb{H}_n^A[\alpha]$ as a subset. Meanwhile, setting those in $\pmb{H}_n^A[\alpha]$ to constants is already sufficient to remove all the unwanted differentials.

Let us consider two examples where we explicitly show that the pullback conjecture holds. First consider the ordering $\alpha=[(1,2),(3,(4,5),6)]$. We find that $\pmb{H}^{A}_{n}[\alpha]=\emptyset$, but the restriction equation from section~\ref{sec:invsoftconstruction} are 
\begin{equation}
\pmb{H}_{6}[(1,2),(3,(4,5),6)]=\{ s_{3,6}=-c_{3,6}\} \ .
\label{Hnrestrictionsssix}
\end{equation}
The $W(g|\alpha)$ of any incompatible diagram must vanish upon pullback to eq. (\ref{Hnrestrictionsssix}).  For example, the differential corresponding to one incompatible diagram is, 
\begin{equation}
dW=ds_{1,2}\wedge ds_{3,6} \wedge ds_{4,6} \ , 
\end{equation}
which is zero upon pullback to eq.~\eqref{Hnrestrictionsssix} as $ds_{3,6}=-d(c_{3,6})=0$. As a slightly more non-trivial example, consider the ordering $\alpha=[(1,2)(3,(4,(5,6),7),8)]$. \noindent Unlike the previous example, $\pmb{H}^{A}_{n}[\alpha]$ is no longer empty and equal to 
\begin{equation}
\pmb{H}_{8}^{A}[(1,2),(3,(4,(5,6),7),8)]=\{ s_{5,6,8}=-c_{5,6,8} \}\ .
\label{HnrestrictionsseightA}
\end{equation}
In addition, the restriction equations of section~\ref{sec:invsoftconstruction} take the form:
\begin{equation}
\pmb{H}_{8}[(1,2),(3,(4,(5,6),7),8)]=\{ s_{3,8}=-c_{3,8}, \quad s_{4,7}=-c_{4,7} \}\ .
\label{Hnrestrictionsseight}
\end{equation}
Upon imposing eq.~\eqref{Hnrestrictionsseight}, all incompatible differentials vanish except 
\begin{equation}
\begin{split}
dW_{1}&=ds_{1,2}\wedge ds_{1,2,4}\wedge ds_{1,2,4,7}\wedge ds_{1,2,3,4,7}\wedge ds_{1,2,3,4,7,8}\ ,\\
dW_{2}&=ds_{1,2}\wedge ds_{1,2,4}\wedge ds_{1,2,4,7}\wedge ds_{1,2,4,7,8}\wedge ds_{1,2,3,4,7,8} \ , \\
dW_{3}&=ds_{1,2}\wedge ds_{1,2,7}\wedge ds_{1,2,4,7}\wedge ds_{1,2,3,4,7}\wedge ds_{1,2,3,4,7,8}\ , \\
dW_{4}&=ds_{1,2}\wedge ds_{1,2,7}\wedge ds_{1,2,4,7}\wedge ds_{1,2,4,7,8}\wedge ds_{1,2,3,4,7,8} \ .
\end{split}
\label{8pointexmapleipullback}
\end{equation}
However, after we impose the additional constraints in eq.~\eqref{HnrestrictionsseightA}, all the differentials in eq.~\eqref{8pointexmapleipullback} vanish due to the linear relations among the Mandelstam variables. \\

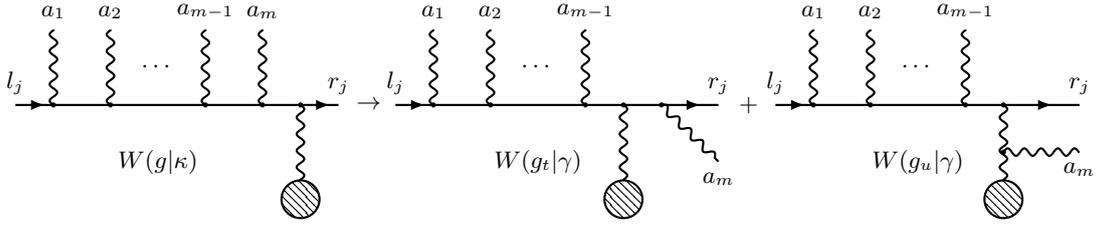
\begin{figure}
\centering
\begin{tikzpicture}[every node/.style={font=\footnotesize},dir/.style={decoration={markings, mark=at position \halfway with {\arrow[scale=0.8]{Latex}}},postaction={decorate}},glu/.style={decorate,decoration={coil, amplitude=4pt,segment length=5.15pt}},wavy/.style={decorate,decoration={coil,aspect=0, segment length=2.2mm, amplitude=0.5mm}}]
\draw [thick,dir] (0,0) node[above=0pt]{$l_j$} -- (0.5,0);
\draw [thick,dir] (3.75,0) -- (4.25,0) node[above=0pt]{$r_j$};
\draw [thick] (0.5,0) -- (3.75,0);
\draw [thick,wavy] (0.5,0) -- (0.5,1) node[above=0pt]{$a_1$};
\draw [thick,wavy] (1.25,0) -- (1.25,1) node[above=0pt]{$a_2$};
\draw [thick,wavy] (2.5,0) -- (2.5,1) node[above=0pt]{$a_{m-1}$};
\draw [thick,wavy] (3.25,0) -- (3.25,1) node[above=0pt]{$a_m$};
\draw [thick,wavy] (3.75,0) -- (3.75,-1);
\node at (1.875,0.5) {$\cdots$};
\filldraw (0.5,0) circle (0.75pt) (1.25,0) circle (0.75pt) (2.5,0) circle (0.75pt) (3.25,0) circle (0.75pt) (3.75,0) circle (0.75pt);
\draw [pattern=north west lines,thick] (3.75,-1.25) circle (0.25cm);
\node at (1.875,-0.75) {$W(g|\kappa)$};
\begin{scope}[xshift=5cm]
\draw [thick,dir] (0,0) node[above=0pt]{$l_j$} -- (0.5,0);
\draw [thick,dir] (3.5,0) -- (4.25,0) node[above=0pt]{$r_j$};
\draw [thick] (0.5,0) -- (3.5,0);
\draw [thick,wavy] (0.5,0) -- (0.5,1) node[above=0pt]{$a_1$};
\draw [thick,wavy] (1.25,0) -- (1.25,1) node[above=0pt]{$a_2$};
\draw [thick,wavy] (2.5,0) -- (2.5,1) node[above=0pt]{$a_{m-1}$};
\draw [thick,wavy] (3.5,0) -- (4.25,-0.75) node[below=0pt]{$a_m$};
\draw [thick,wavy] (3,0) -- (3,-1);
\node at (1.875,0.5) {$\cdots$};
\draw [pattern=north west lines,thick] (3,-1.25) circle (0.25cm);
\filldraw (0.5,0) circle (0.75pt) (1.25,0) circle (0.75pt) (2.5,0) circle (0.75pt) (3,0) circle (0.75pt) (3.5,0) circle (0.75pt);
\node at (1.875,-0.75) {$W(g_t|\gamma)$};
\end{scope}
\begin{scope}[xshift=10cm]
\draw [thick,dir] (0,0) node[above=0pt]{$l_j$} -- (0.5,0);
\draw [thick,dir] (3,0) -- (4,0) node[above=0pt]{$r_j$};
\draw [thick] (0.5,0) -- (3,0);
\draw [thick,wavy] (0.5,0) -- (0.5,1) node[above=0pt]{$a_1$};
\draw [thick,wavy] (1.25,0) -- (1.25,1) node[above=0pt]{$a_2$};
\draw [thick,wavy] (2.5,0) -- (2.5,1) node[above=0pt]{$a_{m-1}$};
\draw [thick,wavy] (3,-0.625) -- (4,-0.625) node[below=0pt]{$a_m$};
\draw [thick,wavy] (3,0) -- (3,-1);
\node at (1.875,0.5) {$\cdots$};
\filldraw (0.5,0) circle (0.75pt) (1.25,0) circle (0.75pt) (2.5,0) circle (0.75pt) (3,0) circle (0.75pt) (3,-0.625) circle (0.75pt);
\draw [pattern=north west lines,thick] (3,-1.25) circle (0.25cm);
\node at (1.875,-0.75) {$W(g_u|\gamma)$};
\end{scope}
\node at (9.65,0) {$+$};
\node at (4.65,0) {$\rightarrow$};
\end{tikzpicture}
    \caption{A visualization of eq.~(\ref{resttocon}), a key step to prove Lemma~\ref{lemma}.}
    \label{fig:feynmandiagrsum}
\end{figure}

\subsection{Proof of the pullback conjecture}\label{sec:pullbackconjproof}

We now give a complete proof of the pullback conjecture~\eqref{conjectdiffofr}. We first prove a lemma, based on which we given an inductive proof of eq.~\eqref{conjectdiffofr}.

\begin{lemma}\label{lemma}
Given a Feynman diagram $g$ that is incompatible with $\alpha$, the pullback of the differential $W(g|\kappa)$ to $H_n^T[\alpha]$ can be written as 
\begin{equation}
W(g|\kappa)\big|_{H_n^T[\alpha]}=\sum_{g_{l}} W(g_{l}|\kappa_{l})\big|_{H_n^T[\alpha]}, 
\label{wgkappasumansatz}
\end{equation}
where the sum is over the graphs $g_l$ that contain at least one propagator of the form $s_{a,b}$, $s_{a,r_j}$ or $s_{l_j,r_j}$, and incompatible with $\alpha$. Namely, the differential $W(g_l|\kappa_l)$ can be written as
\begin{align}
    W(g_l|\kappa_l)=ds_{I_l}\wedge W(g_l'|\kappa_l')\quad\text{ where }\quad s_{I_l}\in\{s_{a,b},s_{a,r_j},s_{l_j,r_j}\}\,,
\end{align}
where $a$, $b$ are $\mathfrak{adj}$ states and $l_i$, $r_i$ are $(\mathfrak{a})\mathfrak{f}$ states.
\end{lemma}

\begin{proof}
First, we note that any Feynman diagram must contain at least one propagator of the form
$\{s_{a,b},s_{a,r_j},s_{l_j,r_j},s_{l_j,a_1,a_2,\ldots,a_m}\}$. We can thus write $W(g|\kappa)$ as
\begin{equation}
W(g|\kappa)=ds_{I}\wedge W(g'|\kappa')\quad\text{ where }\quad s_I\in\{s_{a,b},s_{a,r_j},s_{l_j,r_j},s_{l_j,a_1,a_2,\ldots,a_m}\}\,.
\label{refereq}
\end{equation}
For the first three cases, eq.~\eqref{wgkappasumansatz} trivially holds, so we only need to prove it for the Feynman diagrams that contain the propagator $s_{l_j,a_1,a_2,\ldots,a_m}$ but none of the other three. In fact, $W(g|\kappa)$ for such graphs has to take the form
\begin{align}\label{eq:Wg7t}
    W(g|\kappa)=\Big(\bigwedge_{j=1}^{m-1}ds_{l_i,a_1,\ldots,a_j}\Big)\wedge ds_{l_i,a_1,\ldots a_{m-1},a_m}\wedge ds_{l_i,a_1,\ldots,a_m,r_i}\wedge W(g''|\kappa'')
\end{align}
since otherwise propagators of the first three types in eq.~\eqref{refereq} must appear. We use the seven-term identity to $s_{l_i,a_1,\ldots a_{m-1},a_m}$, which results in two new diagrams, as shown in figure~\ref{fig:feynmandiagrsum}. Thus we have $W(g|\kappa)=W(g_t|\gamma)+W(g_u|\gamma)$, where
\begin{align}\label{resttocon}
    W(g_t|\gamma)=ds_{a_m,r_j}\wedge W(g_t'|\gamma')\,,\qquad W(g_u|\gamma)=ds_{l_j,a_1,\ldots,a_{m-1}}\wedge W(g_u'|\gamma'')\,.
\end{align}
If both $g_t$ and $g_u$ are compatible with $\alpha$, namely, $\gamma=\alpha$, we recombine them back into eq.~\eqref{eq:Wg7t} and use
\begin{align}\label{eq:sla1am}
s_{l_i,a_1,\ldots,a_{m-1},a_m}&=s_{l_i,a_1,\ldots,a_{m-1}}+s_{l_i,a_m}+\sum_{j=1}^{m-1}s_{a_j,a_m}=s_{l_i,a_1,\ldots,a_{m-1}}-c_{l_i,a_m}-\sum_{j=1}^{m-1}c_{a_j,a_m}\,,
\end{align}
such that $ds_{l_i,a_1,\ldots,a_{m-1},a_m}=ds_{l_i,a_1,\ldots,a_{m-1}}$ and $W(g|\kappa)\big|_{H_n^T[\alpha]}=0$. We note that the first equality of eq.~\eqref{eq:sla1am} is due to the seven-term identities and the second equality is due to the constraints in $\pmb{H}_n[\alpha]$.
If both $g_t$ and $g_u$ are still incompatible with $\alpha$, then $W(g_t|\gamma)$ is already in the form of eq.~\eqref{wgkappasumansatz}. We can recursively apply the process to $ds_{l_j,a_1,\ldots,a_{m-1}}$ in $W(g_u|\gamma)$ until the differential vanishes or we will end up with a differential corresponding to a diagram with propagator $s_{l_j,r_j}$. 
\end{proof}

\indent From the explicit form of $\pmb{H}^{A}_{n}[\alpha]$ in eq. (\ref{setrestriequss}) and $\pmb{H}_{n}[\alpha]$ in section~\ref{sec:invsoftconstruction}, we know that the $s_{I}$ of the form:

\begin{equation}
 s_I\in\{s_{{l}_{j},{r}_{j}}, s_{a,b}, s_{a,{r}_{j}}\} \ ,
\label{equ:povesconsarb}
\end{equation}

\noindent are constants if incompatible with $\alpha$.\footnote{The constraints $\pmb{H}_n^A[\alpha]$ are needed because $(a,b)$ and $(a,\bar{r}_j)$ can be separated by more than one flavor lines and $\pmb{H}_n[\alpha]$ alone does not put constraints on them.} Therefore, the $W(g_{l}|\kappa_{l})$ in eq.~\eqref{wgkappasumansatz} is nonzero only if the $s_{I_{l}}$ in eq.~\eqref{wgkappasumansatz} is compatible with $\alpha$. For these nonzero differentials, the $W(g_{l}'|\kappa_{l}')$ in eq.~\eqref{wgkappasumansatz} corresponds to the differential of a reduced graph where the two states in $I_{l}$ are combined into a single state. Since $W(g_{l}'|\kappa_{l}')$ multiplies $ds_{I_{l}}$, we can effectively study $W(g_{l}'|\kappa_{l}')$ on the facet $s_{I_{l}}=0$,
\begin{align}
     W(g'_{l}|\kappa'_{l})\big|_{H_n^T[\alpha]}= W(g'_{l}|\kappa'_{l})\big|_{H_{n-1}^{T}[\alpha'_l]}+ds_{I_l}\wedge\omega\,.
\end{align}
The discrepancy is always of the form $ds_{I_l}\wedge \omega$ so it drops when multiplied with $ds_{I_l}$. We denote the $s_{I_l}=0$ facet of the subspace $H_n^T[\alpha]$ as $H_{n-1}^T[\alpha'_l]$, and it is given by the constraints $\pmb{H}_{n-1}^{T}[\alpha_l']$. Now we would like to recursively apply Lemma~\ref{lemma} but with $W(g|\kappa)$ replaced by $W(g_{l}'|\kappa_{l}')$ and $H_n^T[\alpha]$ replaced by $H_{n-1}^T[\alpha'_l]$. Therefore, for those $W(g|\kappa)$ incompatible with $\alpha$, at certain stage of this recursive process, we can factor out a one-form that belongs to eq.~\eqref{equ:povesconsarb} and at the same time is incompatible to $\alpha$ such that the entire differential $W(g|\kappa)$ vanishes.\footnote{The constraints on $s_{\mathsf{I},\mathsf{I}'}$ and $s_{\mathsf{I},r}$ in eq.~\eqref{setrestriequss} with $\mathsf{I}$ and/or $\mathsf{I}'$ being multi-particle blocks will appear as $s_{a,b}$ and $s_{a,r}$ on corresponding factorization channels.} This would prove the pullback conjecture~\eqref{conjectdiffofr}.

\indent However, as discussed in section~\ref{sec:extendeequiclass}, the facet constraints $\pmb{H}_{n-1}^{T}[\alpha_{l}']$ in general are not the same as those given in section~\ref{sec:invsoftconstruction} and eq.~\eqref{setrestriequss}. Fortunately, for the inductive proof to hold, we only need to show that the constraints imposed on the variables in eq.~\eqref{equ:povesconsarb} are not deformed under taking sequences of two-particle factorization of the form~\eqref{equ:povesconsarb}. Appendix~\ref{sec:factorizationadjacepart} shows that under such sequences of factorizations, the only possible type of deformation is

\begin{equation}
c_{i,\mathsf{B}}+Ks_{\mathsf{B}}=-s_{i,\mathsf{B}} \ .
\end{equation}

\noindent where $i$ is a single particle state, $\mathsf{B}$ is a block and $K$ is some non-negative integer constant. The possible constraints on $s_{I}$ in eq.~\eqref{equ:povesconsarb} are therefore unchanged since $\mathsf{B}$ is always a single particle state,

\begin{equation}
c_{i,j}=-s_{i,j}, \quad \textrm{ if  }\mathsf{B}\textrm{ is a single particle state }j \ ,
\end{equation}

\noindent so that our inductive proof holds. 

\section{Conclusion}
\label{sec:conclusion}

In this paper, we continue the study of positive geometries associated with bi-color theory, focusing in particular on the facet geometry of  open associahedra. We found many new features in the facet geometries, including a broader equivalence class of associahedra and a \textit{fiber} geometries. In addition, we found a highly efficient recursion for open associahedra which generalizes the recursions of Refs.~\cite{Salvatori:2019phs,Yang:2019esm} to open polytopes. We concluded with a proof of the pullback conjecture of Ref.~\cite{Herderschee:2019wtl}, which required a detailed analysis of the facet geometries. This work represents an important step in the study of positive geometry, with implications beyond the original bi-color theory. Here we outline a number of open questions worth pursuing:

\paragraph{Combinatorial approach to recursion:} Our recursion for bi-color theory follows the algebriac approach of Ref.~\cite{Yang:2019esm}. However, there is also the more combinatorial approach in Ref.~\cite{Salvatori:2019phs} that applies to simple polytopes. 
We suspect that it is possible to generalize the combinatorial approach for bi-color amplitudes using a limiting procedure. For example, we could approximate the open associahedra geometry as some closed associahedra with facets at infinity. If one can show that sending the facets to infinity does not break the recursion, then such a formula should hold for open associahedra.
\paragraph{Stringy canonical form and cluster algebras:} 
As mentioned in Ref.~\cite{Herderschee:2019wtl}, finding a stringy canonical form~\cite{Arkani-Hamed:2019mrd} for open associahedra would be very interesting, as it would provide a probe into stringy corrections for theories with non-trivial flavor structure. However, an immediate difficulty is the existence of many possible stringy deformations that yield in the same canonical form in the $\alpha'\rightarrow 0$ limit. In the case of stringy canonical forms corresponding to bi-adjoint/$Z$-theory, this ambiguity is removed by requiring the canonical form also have cluster algebra structure, where each facet of the closed associahedra is associated with g-vectors in a $A_{n}$ cluster algebra \cite{Arkani-Hamed:2020tuz,Arkani-Hamed:2019plo,Arkani-Hamed:2019mrd,baziermatte2018abhy,2019arXiv190606861P}. The cluster structure of the geometry picks out a unique stringy deformation of the canonical form. It would be interesting to see if this cluster algebra structure could be generalized to open associahedra, allowing one to pick out a stringy deformation (or class of stringy deformations) of flavored amplitudes. 
\paragraph{Positive geometry of one loop bi-color amplitudes:} 
In this paper, we only studied the positive geometries associated with tree-level amplitudes in bi-color theory. However, the positive geometries associated with 1-loop integrands of bi-adjoint amplitudes have been  identified \cite{Salvatori:2018aha, Arkani-Hamed:2019vag}. It would be interesting to see if the open associahedra story can be similarly generalized at 1-loop. There are two main approaches. The first is attempting a generalization of the inverse soft construction of Ref.~\cite{Herderschee:2019wtl} to one-loop amplitudes. Alternatively, one could attempt to develop a procedure for constructing closed cluster polytopes whose forbidden facets can be taken to infinity while retaining the correct edge-vertex structure.
\paragraph{Generalizing the inverse soft construction:} 
When finding the positive geometry associated with amplitudes, finding the associated subspace, $H_{n}$, of kinematic space is often the most difficult task as it encodes the combinatorial structure of the amplitude. The fact that such a subspace exists is one of the most non-trivial assumption of this program and finding more systematic methods to finding such subspaces remains a major open problem for the field. In constructing bi-color amplitudes, the inverse soft construction was crucial for finding the appropriate subspace $H_{n}[\alpha]$. It would be interesting to see if the inverse soft construction could be generalized to find the kinematic subspaces associated with amplitudes in other theories. For example, the positive geometry of the 1-loop momentum amplituhedron currently remains allusive \cite{Damgaard:2019ztj,Ferro:2020lgp}, although the canonical form gives hints to an inverse soft construction \cite{He:2018okq}. Alternatively, it would be interesting if the inverse soft construction could be used to identify the kinematic subspace associated with \textit{massive} amplitudes on the Coulomb branch of $\mathcal{N}=4$ SYM \cite{Craig:2011ws,Herderschee:2019dmc}. Perhaps the inverse soft construction is the missing tool necessary to find the kinematic subspaces associated with these amplitudes.

\acknowledgments
We would like to thank Song He, Giulio Salvatori and Yong Zhang for inspiring discussions and comments. AH would like to especially thank Henriette for continued support and comments. FT is supported by the Knut and Alice Wallenberg Foundation under grant
KAW 2013.0235, and the Ragnar S\"{o}derberg Foundation (Swedish Foundations’ Starting
Grant). AH is supported by the Leinweber Center for Theoretical Physics Graduate Fellowship. 

\appendix

\section{Canonical forms of polytopes}
\label{sec:reviewcanonicaofrm}

In this appendix, we review some useful computational techniques for the canonical forms of polytopes \cite{Arkani-Hamed:2017tmz}. A concise review of positive geometries is given in appendix A of Ref.~\cite{Arkani-Hamed:2017mur}. \\
\indent Loosely, a canonical form is a logarithmic volume form with simple poles that correspond to boundaries of the positive geometry. The residue of a canonical form on a pole is the canonical form of the boundary geometry associated with the pole. In the case of polytopes, boundary geometries are simply the facet geometries. Using this residue condition, one can recursively construct the canonical form a generic positive geometry using the canonical forms of its boundary geometries. Since the canonical form is a differential form, it is often easier to work with a scalar function called the \textit{canonical rational function}:

\begin{equation}
\underline{\Omega}(\mathcal{A})=\Omega(\mathcal{A})/\langle X d^{m}X\rangle     
\end{equation}

\noindent where $\langle X d^{m}X\rangle$ is the homogeneous differential on $\mathbb{P}^{N}$. For reasons we discuss below, the ambient space of positive geometries is almost always $\mathbb{P}^{N}$ instead of $\mathbb{R}^{N}$. \\
\indent Importantly, the canonical form of a positive geometry is \textit{unique} only if the ambient space obeys some minor restrictions. To see this, consider the canonical form of a bounded line in $[a,b] \subset \mathbb{R}$:

\begin{equation}
\Omega=\left ( \frac{1}{x-a}-\frac{1}{x-b}\right )dx \ .
\label{originaldiff}
\end{equation}

\noindent Since we assume that a canonical form is defined by its residue structure, an equally valid canonical form is 

\begin{equation}
\Omega=\left ( \frac{1}{x-a}-\frac{1}{x-b}+1/A\right )dx
\label{badcanoni}
\end{equation}

\noindent where $1/A$ is some constant. Since eq. (\ref{badcanoni}) has the same residue structure as eq. (\ref{originaldiff}), both eqs. (\ref{originaldiff}) and (\ref{badcanoni}) correspond to same bounded line. Therefore, there is naively an ambiguity in how we define the canonical form of a bounded line in $\mathbb{R}$. To remove this ambiguity, we impose that we are actually working in $\mathbb{P}$ with homogeneous coordinates, $Y=(1,x)$. If we embed eqs. (\ref{badcanoni}) and (\ref{originaldiff}) into $\mathbb{P}$, we find the canonical forms are now

\begin{equation}
\begin{split}
\left ( \frac{1}{x-a}-\frac{1}{x-b}\right )dx &\rightarrow \left ( \frac{1}{( Y \cdot W_{a} )}+\frac{1}{( Y \cdot W_{b} )} \right )\langle Y dY\rangle \\
\left ( \frac{1}{x-a}-\frac{1}{x-b}+1/A\right )dx &\rightarrow \left ( \frac{1}{( Y \cdot W_{a} )}+\frac{1}{( Y \cdot W_{b} )}+\frac{1}{( Y \cdot W_{A} )} \right )\langle Y dY\rangle \\
\end{split}
\label{badcanonieded}
\end{equation}

\noindent The last term in eq. (\ref{badcanoni}) now has non-zero residue and must correspond to a different geometry than eq. (\ref{originaldiff}). Therefore, embedding the positive geometry into $\mathbb{P}$ was enough to remove the ambiguity in the canonical form of a bounded line. In general, the possible ambiguities in $\Omega$ correspond to holomorphic top forms on the ambient space, such as the $1/A$ term in eq. (\ref{badcanoni}). To remove these ambiguities, we simply require the ambient space of the positive geometry to have no non-zero holomorphic top forms, such as $\mathbb{P}^{N}$. \\
\indent There are a number of operations we can perform on positive geometries. We can consider the addition of two positive geometries that live in the same space. Suppose we have two positive geometries in the same ambient space: $\mathcal{A}_{a}=(X,X_{\geq 0}^{a})$ and $\mathcal{A}_{b}=(X,X_{\geq 0}^{b})$, and define the sum of the geometries as $\mathcal{A}_{a}+\mathcal{A}_{b}=(X,X_{\geq}^{1}+X_{\geq}^{2})$. The canonical form of $\mathcal{A}_{a}+\mathcal{A}_{b}$ is

\begin{equation}
\Omega(\mathcal{A}_{a}+\mathcal{A}_{b})=\Omega(\mathcal{A}_{a})+\Omega(\mathcal{A}_{b}) \ .
\label{sumgeometries}
\end{equation}

\noindent We can also consider product geometries. Suppose we have two positive geometries in two different space, $\mathcal{A}_{X}=(X,X_{\geq 0})$ and $\mathcal{A}_{Y}=(Y,Y_{\geq 0})$. We then define the product geometry as $\mathcal{A}_{X\times Y}=(X\times Y, X_{\geq 0 }\times Y_{\geq 0})$. The canonical form of $\mathcal{A}_{X\times Y}$ is simply 

\begin{equation}
\Omega(\mathcal{A}_{X\times Y})=\Omega(\mathcal{A}_{X})\times \Omega(\mathcal{A}_{Y})\,.
\end{equation}

\indent Positive geometry is thus the natural mathematical structure that encodes locality and unitarity of scattering amplitudes. One interesting question is whether a certain physical scattering amplitude can be uplifted into the canonical form of some positive geometry. In the following, we will review a few approaches to compute the canonical form.




\subsection{Canonical forms and adjoints}\label{sec:adj}


Suppose we have a polytope in projective space $\mathbb{P}^{m}$ that is carved out by $M$ facets, $X\cdot W_{i}>0$. Its canonical form take the form\footnote{A very concise introduction of this representation of canonical forms can be found at \url{http://www.math.lsa.umich.edu/~tfylam/posgeom/gaetz_notes.pdf}, written by Christian Gaetz.}
\begin{equation}
\Omega[\mathcal{A}]=\frac{\text{adj}_{\mathcal{A}^*}(X)}{\prod_{i=1}^{M} (X\cdot W_{i})}\langle Xd^mX\rangle\,,
\label{ansatssz}
\end{equation} 
where $\text{adj}_{\mathcal{A}^*}(X)$ is the \emph{adjoint} of the dual polytope $\mathcal{A}^*$. It is the unique degree $M-m-1$ homogeneous polynomial that vanishes at all the spurious poles at the intersection of non-adjacent facets~\cite{kohn2019projective}. An explicit form of $\text{adj}_{\mathcal{A}^*}(X)$ will be given in eq.~\eqref{eq:adjpoly}.


\indent For instance, consider two facets, $W_{1}$ and $W_{2}$, whose associated boundaries, $\mathcal{B}_{1}$ and $\mathcal{B}_{2}$, do not touch inside $\mathcal{A}$. According to eq. (\ref{ansatssz}), there is naively a two dimensional singularity at $(X\cdot W_{1})=0$ and  $(X\cdot W_{2})=0$. However, since $\mathcal{B}_{1}$ and $\mathcal{B}_{2}$ do not intersect, such a singularity is unphysical. Therefore, $\text{adj}_{\mathcal{A}^*}(X)$ must vanish whenever both $(X\cdot W_{1})=0$ and  $(X\cdot W_{2})=0$ to cancel this two dimensional singularity. All $(W_{j},W_{i})$ pairs which do not touch in $\mathcal{A}$ impose such a restriction on $\text{adj}_{\mathcal{A}^*}(X)$, which can be used to restrict the form of $\text{adj}_{\mathcal{A}^*}(X)$. For example, such restrictions are enough to derive the unique canonical form associated with generic cyclic polytopes, as derived in~\cite{Arkani-Hamed:2014dca, Arkani-Hamed:2017tmz}.

\subsection{Disjoint triangulations}

An alternate method to calculating the canonial form associated with a polytope is to break apart $\mathcal{A}$ into a set of triangulations $\tau(\mathcal{A})$, such that 
\begin{equation}
\mathcal{A}=\bigcup_{\Delta\in\tau(\mathcal{A})} \Delta \,,
\end{equation}
where $\Delta$ is a simplex. From eq. (\ref{sumgeometries}), the canonical form is then the sum of those for the simplices,
\begin{equation}
\Omega[\mathcal{A}]=\sum_{\Delta\in\tau(\mathcal{A})} \Omega[\Delta] \ .
\label{tirangultio}
\end{equation}
The difficulty in applying eq. (\ref{tirangultio}) is finding a particular triangulation such that the $\Omega(\Delta)$ are easy to calculate. The recursion procedures in Section \ref{recursionsection} a partial triangulation of the generalized associahedra geometry, where the subdivision of $\tau(\mathcal{A})$ is not necessarily into simplices.\footnote{The definition of triangulation in Ref.~\cite{Arkani-Hamed:2017tmz,Arkani-Hamed:2017mur} does not make a distinction between triangulations and partial triangulations as the distinction breaks down for non-polytopal geometries.}

\subsection{Dual polytopes}
\label{sec:dualpolytope}

The final approach we will consider for calculating the canonical form is using the \textit{dual} polytope. Taking the dual of a polytope loosely corresponds to switching the role of facets and vertices. We generically find that the canonical form a polytope corresponds to the actual volume of the dual polytope (not just a volume form). \\
\indent An arbitrary polytope in $\mathbb{P}^{m}$ can be defined by its facets, $W_{i}$:

\begin{equation}
\mathcal{A}=\{\ X\ |\ \forall i: \ (X\cdot W_{i})>0 \}    
\end{equation}

\noindent or its vertices, $V_{i}$:

\begin{equation}
\mathcal{A}=\{\ \sum_{i} C_{i}V_{i}\ |\ \forall i: \ C_{i}>0 \} \ . \end{equation}

\noindent We can define the dual polytope by simply switching the role of $W_{i}$ and $V_{i}$. Facets of $\mathcal{A}$ are mapped to vertices of $\mathcal{A}^{\star}$ and vice-versa. Rather, we can define the dual polytope, $\mathcal{A}^{\star}$, using the facets of $\mathcal{A}$,

\begin{equation}
\mathcal{A}^{\star}=\{\ \sum_{i} C_{i}W_{i}\ |\ \forall i: \ C_{i}>0 \} \ , \end{equation}

\noindent or the vertices of $\mathcal{A}$,

\begin{equation}
\mathcal{A}^{\star}=\{\ X \ |\ \forall i: \ (X\cdot V_{i})>0 \}  \ .  
\end{equation}

\noindent $\mathcal{A}^{\star}$ is extremely useful from a computational standpoint as we can identify the canonical form of $\mathcal{A}$ with the actual volume of $\mathcal{A}^{\star}$:

\begin{equation}
\underline{\Omega}[\mathcal{A}]=\textrm{Vol}(\mathcal{A}^{\star}) \ .
\end{equation}

\noindent The above identity is especially useful if $\mathcal{A}$ corresponds to a simple polytope, a polytope where each vertex intersects the minimal number of facets. Such vertices are mapped to facets in the dual polytope, $\mathcal{A}^{\star}$. We now define a reference point in the dual polytope, $W^{\star}\in \mathcal{A}^{\star}$. For each facet of the dual polytope $\mathcal{A}^{\star}$, we define a simplex bounded by the facet and the reference point, $W^{\star}$. The collection of all such simplexes defines a disjoint triangulation of the dual polytope. Furthermore, the volume of each simplex is very compact:

\begin{equation}
\frac{\langle W^{*}\prod_{i\in F(v)}W_i\rangle}{(X\cdot W^{*})\prod_{i\in F(v)}(X\cdot W_i)}   \ , 
\end{equation}

\noindent where $W_{i}$ correspond to the vertices that bound the corresponding facet of the dual polytope. Summing over all such simplices in the dual polytope, which corresponds to a vertex expansion of the original polytope, we find the canonical form is 

\begin{equation}
\Omega[\mathcal{A}] =\sum_{v\in\textrm{vertices}}\frac{\langle W^{*}\prod_{i\in F(v)}W_i\rangle}{(X\cdot W^{*})\prod_{i\in F(v)}(X\cdot W_i)}\langle Xd^m X\rangle  \ , 
\label{vertextriangulation}
\end{equation}

\noindent where $\Omega[\mathcal{A}]$ is independent of $W^{\star}$. Importantly, eq. (\ref{vertextriangulation}) does not correspond to a disjoint triangulation of $\mathcal{A}$. 

From eq.~\eqref{vertextriangulation}, one can take the common denominator and reach the form of eq.~\eqref{ansatssz} for the canonical form. This provides an explicit formula for the adjoint~\cite{Warren1996},
\begin{align}\label{eq:adjpoly}
    \text{adj}_{\mathcal{A}^*}(X)&=\sum_{\Delta^*\in\tau(\mathcal{A}^*)}\textrm{vol}(\Delta^*)\prod_{i\in V(\mathcal{A}^*)\backslash V(\Delta^*)}(X\cdot W_i) \nonumber\\
    &=\sum_{\Delta^*\in\tau(\mathcal{A}^*)}\Big\langle{ \prod_{j\in V(\Delta^*)}}W_j\Big\rangle \prod_{i\in V(\mathcal{A}^*)\backslash V(\Delta^*)}(X\cdot W_i)\,,
\end{align}
where the summation is over all the simplices $\Delta^*$ of a particular triangulation $\tau(\mathcal{A}^*)$ of the dual polytope, and the second factor is a product of vertices in $\mathcal{A}^*$ that are not in the simplex $\Delta^*$. The final result is independent of the triangulations~\cite{Warren1996}. This formula also provides a direct counting of the degree of the adjoint. For a polytope in $\mathbb{P}^m$ that consists of $M$ facets, we can choose a triangulation in which the simplices are all made from vertices of $\mathcal{A}^*$ such that for each term in the summation there are exactly $M-m-1$ powers of $X$. Therefore, we have
\begin{align}
    \text{deg}\,[\text{adj}_{\mathcal{A}^*}(X)]=M-m-1\,.
\end{align}

\section{Semi-direct products}
\label{sec:semidirpro}

In this appendix, we further study the relation between semi-direct product and fiber geometries. We start with a simple example. Consider a generic quadrilateral bounded by
\begin{align}
    X_1=x=0,\quad X_2=ax+y-1=0, \quad X_3=x+by-1=0, \quad X_4=y=0\,,
\end{align}
where $a$ and $b$ are arbitrary parameters as long as the geometry is closed and non-singular. We note that arbitrary quadrilateral can be brought into this form by a general coordinate transformation. The canonical form is given by
\begin{align}
    \Omega[\mathcal{A}]=\frac{X_5}{X_1X_2X_3X_4}dx\wedge dy\,,
\end{align}
where $X_5=1-ax-by$ cancels the spurious singularity introduced by the denominator. We would like to view the quadrilateral as a semi-direct product, $\mathcal{A}=\mathcal{A}_L\bowtie\mathcal{A}_R$, where $\mathcal{A}_L$ and $\mathcal{A}_R$ are two line segments. Their canonical forms are
\begin{align}
    \Omega[\mathcal{A}_L]=\left(\frac{1}{X_1}-\frac{1}{X_3}\right)dx\,,\qquad \Omega[\mathcal{A}_R]=\left(\frac{1}{X_2}-\frac{1}{X_4}\right)dy\,.
\end{align}
The canonical form of the quadrilateral factorizes as $\Omega(\mathcal{A})=\Omega(\mathcal{A}_L)\wedge\Omega(\mathcal{A}_R)$
if and only if $ab=0$. This makes either $X_1 \parallel X_3$ or $X_2 \parallel X_4$. Namely, the quadrilateral is actually a trapezoid, which is a fiber-product geometry. In other words, we have shown that all the semi-direct product quadrilaterals are fiber-product geometries.

Next, we consider the semi-direct product between two generic closed simple polytopes. It is also necessary to assume that none of them are not direct products of lower dimensional polytopes. Suppose we have two polytopes $\mathcal{A}_{L/R}$ in the projective space $\mathbb{P}^m$ and $\mathbb{P}^n$ that are carved out by $M$ and $N$ facets respectively, their canonical forms are given by
\begin{align}\label{eq:canform}
    \Omega[\mathcal{A}_L(X)]=\frac{\text{adj}_{\mathcal{A}_L^*}(X)}{\prod_{i=1}^M(X\cdot W_i^L)}\langle X d^m X\rangle\,,& &\Omega[\mathcal{A}_R(Y)]=\frac{\text{adj}_{\mathcal{A}_R^*}(Y)}{\prod_{i=1}^N(Y\cdot W_i^R)}\langle Y d^n Y\rangle
\end{align}
where $X=(1,x_1,x_2,\ldots,x_m)$ and $Y=(1,y_1,y_2,\ldots,y_n)$ are the homogeneous coordinates. 
The facets $W_i^{L/R}$ are given by vectors in the dual space, see eq.~\eqref{eq:ALAR}.
The degrees of $\text{adj}_{\mathcal{A}_{L/R}^*}$ are $M{-}m{-}1$ and $N{-}n{-}1$ respectively, as given in appendix~\ref{sec:adj}.


Now we consider the linear deformation of the facets. The polytope $\mathcal{A}_L[X;Y]$ is obtained from $\mathcal{A}_L$ by linearly shifting the $C^L$ in the facet vectors by the coordinates of $\mathcal{A}_R$, and vice versa, see eq.~\eqref{eq:Cshift}.
The canonical rational functions become
\begin{align}
    \underline{\Omega}^L(X;Y)=\frac{Q(X;Y)}{\prod_{i=1}^{M}(X\cdot\widetilde{W}_i^L)}\,,& &\underline{\Omega}^R(Y;X)=\frac{P(Y;X)}{\prod_{i=1}^{N}(Y\cdot\widetilde{W}_i^R)}\,,
\end{align}
where the polynomials $Q(X;Y)$ and $P(Y;X)$ are obtained from $\text{adj}_{\mathcal{A}_L^*}(X)$ and $\text{adj}_{\mathcal{A}_R^*}(Y)$ through the same shift of $C^{L/R}$,
\begin{align}
    Q(X;Y)=\text{adj}_{\mathcal{A}_L^*}(X)\Big|_{C_i^L\rightarrow C_i^L+\sum_{j=1}^{n}\alpha_jy_j}\,, & & P(Y;X)=\text{adj}_{\mathcal{A}_R^*}(Y)\Big|_{C_i^R\rightarrow C_i^R+\sum_{j=1}^{m}\beta_jx_j}\,.
\end{align}

We can embed the two deformed polytopes into $\mathbb{P}^{m+n}$. We choose the homogeneous coordinate of $\mathbb{P}^{m+n}$ as $Z=(1,x_1,x_2,\ldots,x_m,y_1,y_2,\ldots,y_n)$, such that  
\begin{align}
    \langle Xd^mX\rangle\wedge\langle Yd^nY\rangle=dx_1\wedge\ldots\wedge dx_m\wedge dy_1\wedge\ldots\wedge dy_n=\langle Zd^{m+n}Z\rangle\,.
\end{align}
The facets are embedded through
\begin{align}
    & \widetilde{W}_i^L\rightarrow\mathcal{W}_i=(C_i^L,w_1^i,w_2^i,\ldots,w_m^i,\alpha_1,\alpha_2,\ldots,\alpha_n)\,,\nonumber\\
    & \widetilde{W}_i^R\rightarrow\mathcal{W}_{M+i}=(C_i^R,\beta_1,\beta_2,\ldots,\beta_m,u_1^i,u_2^i,\ldots,u_n^i)\,,
\end{align}
such that we have $X\cdot\widetilde{W}_{i}^L=Z\cdot\mathcal{W}_i$ and $Y\cdot\widetilde{W}_{i}^R=Z\cdot\mathcal{W}_{M+i}$. The geometry $\mathcal{A}_L\bowtie\mathcal{A}_R$ is thus defined as the convex hull of the $M+N$ facets $\mathcal{W}_i$ in $\mathbb{P}^{m+n}$. Meanwhile, the adjoints after the shift become
\begin{align}
    Q(X;Y)\rightarrow q(Z)\,,& & P(Y;X)\rightarrow p(Z)\,.
\end{align}
According to eq.~\eqref{eq:adjpoly}, the degrees of $q(Z)$ and $p(Z)$ are increased by one since the volume of dual simplices becomes a linear function in $Z$ after \emph{nonzero shifts}. 
Now we need to investigate whether the product $q(Z)p(Z)$ gives the adjoint of $(\mathcal{A}_L\bowtie\mathcal{A}_R)^*$, the numerator of the canonical rational function $\underline{\Omega}(\mathcal{A}_L\bowtie\mathcal{A}_R)$. By construction the polynomial $q(Z)p(Z)$ vanishes at all the spurious vertices of $\mathcal{A}_L\bowtie\mathcal{A}_R$, such that due to the uniqueness of the adjoint, we only need to check its degree,
\begin{itemize}
    \item If all $\alpha_i=\beta_i=0$, then we fall back onto the direct product geometry. For this case,
    \begin{align}
        q(Z)=\textrm{adj}_{\mathcal{A}_L^*}(X)\Big|_{X\cdot W_i^L\rightarrow Z\cdot\mathcal{W}_i}\,,& & p(Z)=\textrm{adj}_{\mathcal{A}_R^*}(Y)\Big|_{Y\cdot W_i^R\rightarrow Z\cdot\mathcal{W}_{M+i}}\,,
    \end{align}
    and thus the degrees are unchanged. The product $q(Z)p(Z)$ thus gives the correct adjoint after including a facet at infinity,
    \begin{align}
        \underline{\Omega}(\mathcal{A}_L\times\mathcal{A}_R)=(Z\cdot\mathcal{W}_{\infty})\,\underline{\Omega}^L(X;Y)\,\underline{\Omega}^R(Y;X)=\frac{(Z\cdot \mathcal{W}_{\infty})\,q(Z)\,p(Z)}{\prod_{i=1}^{M+N}(Z\cdot\mathcal{W}_i)}\,,
    \end{align}
    where $\mathcal{W}_{\infty}=(1,0,\ldots,0)$.
    \item If all $\beta_i=0$ but some of the $\alpha_i$'s are nonzero, we fall back onto the fiber-product geometry. For this case, we have
    \begin{align}
        \text{deg}[q(Z)]=M-m\,, & &\text{deg}[p(Z)]=N-n-1\,.
    \end{align}
    The product $q(Z)p(Z)$ thus has exactly the correct degree $(M+N)-(m+n)-1$ as the adjoint of $\mathcal{A}_L\ltimes\mathcal{A}_R$. Thus we have
    \begin{align}
    \underline{\Omega}(\mathcal{A}_L\ltimes\mathcal{A}_R)=\underline{\Omega}^L(X;Y)\,\underline{\Omega}^R(Y;X)=\frac{q(Z)\,p(Z)}{\prod_{i=1}^{M+N}(Z\cdot\mathcal{W}_i)}\,.
    \end{align}
    The same analysis applies for $\alpha_i=0$ but some $\beta_i$'s being nonzero.
    \item Finally, if some of the $\alpha_i$ and $\beta_i$ are nonzero at the same time, then the degree of $q(Z)p(Z)$ will be $(M+N)-(m+n)$, which is too high for an adjoint. Such a numerator will at least result in additional poles at infinity.
\end{itemize}
Thus we have shown that the canonical form of a closed semi-direct product geometry factorizes if and only if it is a fiber-product geometry. The same conclusion applies to those unbounded open semi-direct product geometries that are obtained from closed geometries by sending certain facets to infinity.

\section{Resolving the constraints for amplitudes with adjacent $\mathfrak{f}$-$\mathfrak{af}$ pairs}
\label{sec:relationplanarvariabless}

For closed associahedra, there is a closed form solution to any $X_{i,j}$ in terms of variables $X_{2,i}$ and the $c$-constants, given by eq.~\eqref{closedformequt}. The generalization of eq.~\eqref{closedformequt} to open associahedra is difficult due to the complexity of $\pmb{H}_{n}[\alpha]$ for generic $\alpha$. \\
\indent We will consider the case where all $\mathfrak{f}$-$\mathfrak{af}$ pairs are adjacent. Our preferred basis is 

\begin{equation}
\begin{split}
Y= \{1\,, X_{1,3}\,, X_{3,5}\,, X_{5,7}\,,\ldots, X_{n-1,1}\,, X_{3,7}\,, X_{3,9}\, ,\ldots, X_{3,n-1} \} \ .
\end{split}
\label{bigbasisalladjacent}
\end{equation}

\noindent The restriction equations are given in eq.~\eqref{eq:adjacentGeneric}, but reproduced here for convenience,

\begin{align}
    \pmb{H}_n[(1,2),(3,4),\ldots,(n-1,n)]=\left\{
    \begin{array}{c}
    s_{j,2i-1,2i}=-c_{j,2i-1,2i} \\
    s_{2i-1,2i,2k}=-c_{2i-1,2i,2k}
    \end{array} \middle|
    \begin{array}{c}
    2\leqslant i\leqslant n/2 \\
    2\leqslant j\leqslant 2i-3 \\
    i+1\leqslant k\leqslant n/2
    \end{array}\right\}\,.
\end{align}

\noindent In the basis given by eq. (\ref{bigbasisalladjacent}), we find that any $X_{i,j}$ can be written in a closed form, 

\begin{align}\label{eq:adjacentConstraints}
    X_{2i-1,2m+1} &=-X_{3,2i+1}+X_{3,2m+1}+X_{2i-1,2i+1}\nonumber\\
    &\quad +(2i-4)\sum_{k=i+1}^{m}X_{2k-1,2k+1}+\sum_{j=3}^{2i-2}\sum_{k=i+1}^{m}c_{j,2k-1,2k}\,, \nonumber\\
    X_{2i-1,2m} &= -X_{3,2i+1}+X_{3,2m+1}+2X_{2i-1,2i+1}+(2i-5)X_{2m-1,2m+1}\nonumber\\
    &\quad +(2i-3)\sum_{k=i+1}^{m-1}X_{2k-1,2k+1}+\sum_{j=3}^{2i-2}\sum_{k=i+1}^{m}c_{j,2k-1,2k}+\sum_{k=i}^{m-1}c_{2k-1,2k,2m}\,, \nonumber\\
    X_{2i,2m+1}&=-X_{3,2i+1}+X_{3,2m+1}+(2i-3)\sum_{k=i+1}^{m}X_{2k-1,2k+1}+\sum_{j=3}^{2i-1}\sum_{k=i+1}^{m}c_{j,2k-1,2k}\,, \nonumber\\
    X_{2,2m+1} &= -X_{1,3}+X_{3,2m+1}+\sum_{k=2m+1}^{n/2}X_{2k-1,2k+1}+\sum_{k=m+1}^{n/2}c_{2,2k-1,2k}\,,
\end{align}

\noindent where $i<m$ and $n+1\equiv 1$. Eq.~\eqref{eq:adjacentConstraints} can be checked by making the substitutions,

\begin{equation}
\begin{split}
X_{i,j}&\rightarrow \Big( \sum_{a=i}^{j-1}\ p_{a} \Big)^{2} \ ,  \\
c_{I}&\rightarrow -\Big(\sum_{a\in I}\ p_{a} \Big)^{2} \ ,
\end{split}
\end{equation}

\noindent and showing that each equation reduces to some form of momentum conservation. 

\section{Specific factorization channels}
\label{sec:factorizationchannel}

\noindent This section contains the main technical detail necessary for the proof of the pullback conjecture in section~\ref{sec:proofpullback}. We begin with a very quick study of the factorization channels of two adjacent states to build intuition. We then move onto a more detailed analysis of the factorization channel associated with an arbitrary fermionic block. 

\subsection{Factorization of two adjacent particles}
\label{sec:factorizationadjacepart}


The first factorization example we consider is of the form
\begin{equation}\label{eq:factorizationchanleresitrtwo}
    s_{i,i+1}=q^2\rightarrow 0\,.
\end{equation}
This channel is always allowed \emph{except for} the case that they are $\mathfrak{f}$ and $\mathfrak{af}$ states with different flavors. On the factorization channel, the constraints that involve both particles are trivially inherited as
\begin{align}
    s_{i,i+1,\ldots}=-c_{i,i+1,\ldots}\;\rightarrow\;s_{q,\ldots}=-c_{q,\ldots}\,,
\end{align}
where $c_{q,\ldots}=c_{i,i+1,\ldots}$. Next, we consider the constraints $s_{i,\mathsf{B}}=-c_{i,\mathsf{B}}$ and $s_{\mathsf{C},i+1}=-c_{\mathsf{C},i+1}$, where the block $\mathsf{B}$ and $\mathsf{C}$ are adjacent to $i+1$ and $i$ respectively.\footnote{We note that these constraints appear both in $\pmb{H}_n[\alpha]$ defined in section~\ref{sec:invsoftconstruction} and $\pmb{H}_n^A[\alpha]$ defined in section~\ref{sec:proofpullback}.} One can easily show that they are automatically satisfied on the factorization channel due to the strict positive-ness of the planar variables incompatible to the factorization channel.
Finally, we are left with the constraints that take the form
\begin{align}\label{sumc}
    s_{i,\mathsf{B}}=-c_{i,\mathsf{B}}\quad\text{ and }\quad s_{i+1,\mathsf{B}}=-c_{i+1,\mathsf{B}}\,.
\end{align}
Here, $\mathsf{B}$ is a block that is \emph{not} adjacent to either $i$ or $i+1$. It is not hard to see that the construction given in section~\ref{sec:invsoftconstruction} always produces them in pairs. We can take the sum of the constraints in eq.~\eqref{sumc}, which gives
\begin{align}\label{equ:deformationappendixone}
    s_{q,\mathsf{B}}=s_{i,\mathsf{B}}+s_{i+1,\mathsf{B}}-s_{\mathsf{B}}=-c_{i,\mathsf{B}}-c_{i+1,\mathsf{B}}-s_{\mathsf{B}}\;\rightarrow\;s_{q,\mathsf{B}}=-c_{q,\mathsf{B}}-s_{\mathsf{B}}\,,
\end{align}
where $c_{q,\mathsf{B}}=c_{i,\mathsf{B}}+c_{i+1,\mathsf{B}}$. If we apply the construction of section~\ref{sec:invsoftconstruction} to the sub-amplitude, we will get instead $s_{q,\mathsf{B}}=-c_{q,\mathsf{B}}$. Therefore, this is a deformed constraint that follows the pattern of eq.~\eqref{clsoedeforss}. The $s_{i,i+1}=0$ facet is thus an example of the geometries given by the deformed constraints as discussed in section~\ref{sec:extendeequiclass}. 
The same calculation shows that the deformation~\eqref{equ:deformationappendixone} is closed under taking subsequent adjacent-particle factorizations that probe lower dimensional boundaries.

\subsection{Factorization of an arbitrary fermionic block}

As our second example, we study the facet geometries associated with the factorization channel $s_{\mathsf{B}_{j}}=X_{l_j,l_{j+1}}=p^2\rightarrow 0$, where $\mathsf{B}_{j}$ is an arbitrary fermionic block. The color ordering factorizes into

\begin{align}
& \alpha_L=(\mathsf{B}_1,\ldots,\mathsf{B}_{j-1},p,\mathsf{B}_{j+1},\ldots,\mathsf{B}_{m}) \,,\quad  \alpha_R=(-p,\mathsf{B}_{j}) \ .
\end{align} 

\noindent We suppose that the block $\mathsf{B}_{j}$ has $s$ sub-blocks $\mathsf{B}_{j}=\big(l_j,\mathsf{B}_{j_1},\ldots,\mathsf{B}_{j_s},r_j\big)$. We also define for convenience the set $\pmb{\mathcal{C}}$ as the set of new constraints introduced by adding a new block~$\mathsf{B}_i$,
\begin{align}
    \pmb{H}[\alpha,\mathsf{B}_i]=\pmb{H}[\alpha]\cup\pmb{\mathcal{C}}[\alpha,\mathsf{B}_i]\,.
\end{align}
It is nothing but the union of the constraints $\pmb{C}_{1,2,3}$ and $\pmb{H}_{|\mathsf{B}_i|}$ defined in section~\ref{sec:invsoftconstruction}, which can be easily read off from eq.~\eqref{eq:Hngeneric}.
The constraints relevant to this factorization channel are contained in
\begin{subequations}\label{eq:allequationss}
\begin{align}\label{eq:parta}
& \pmb{\mathcal{C}}\big[(1,2),\ldots,\mathsf{B}_{j-1},\mathsf{B}_{j}\big]\,, \\ \label{eq:partb}
& \pmb{\mathcal{C}}\big[(1,2),\ldots,\mathsf{B}_{j-1},\mathsf{B}_{j},\mathsf{B}_{{j+1}}\big]\,, \\ \label{eq:partc}
& \pmb{\mathcal{C}}\big[(1,2),\ldots,\mathsf{B}_{j-1},\mathsf{B}_j,\mathsf{B}_{j+1},\ldots,\mathsf{B}_{i}\big]\,,\quad \text{ where }\quad j+2\leqslant i\leqslant m\,.
\end{align}
\end{subequations}
We now systematically consider how the constraints in eq.~\eqref{eq:allequationss} translate into $\pmb{H}^{L/R}$ for the facet geometry. Some of these constraints, or certain linear combinations of them, will become elements of $\pmb{H}^{L/R}$, while the rest will be automatically satisfied in the limit $X_{l_{j},l_{j+1}}\rightarrow 0$, so do not impose any constraints on the facet geometry. \\

\begin{itemize}
\item In eq.~\eqref{eq:parta}, the $\pmb{C}_1$ part translates into $\mathcal{A}_L$ under the replacement $s_{\mathsf{B}_{j},\ldots}\rightarrow s_{q,\ldots}$, while the $\pmb{C}_3$ and $\pmb{H}_{|\mathsf{B}_j|}$ part go into $\mathcal{A}_R$. On the other hand, the $\pmb{C}_2$ part are automatically satisfied on the facet.
\item Next, we consider the set in eq.~\eqref{eq:partb}. The $\pmb{C}_3$ and $\pmb{H}_{|\mathsf{B}_{j+1}|}$ part go trivially into $\mathcal{A}_L$, and so does the $\pmb{C}_2$ part after the replacement $s_{\mathsf{B}_{j},\ldots}\rightarrow s_{q,\ldots}$. The $\pmb{C}_1$ part is automatically satisfied on the facet.
\item For the set in eq.~\eqref{eq:partc}, the $\pmb{C}_2$ and $\pmb{C}_3$ part are exactly as before: they translate into $\mathcal{A}_L$ after the replacement $s_{\mathsf{B}_{j},\ldots}\rightarrow s_{q,\ldots}$. 
\end{itemize}
The $\pmb{C}_1$ part of eq.~\eqref{eq:partc} needs further analysis. In particular, the relevant piece is
\begin{align}
\big\{s_{\mathsf{B}_{i},l_j},s_{\mathsf{B}_{i},r_j}\big\}\bigcup_{k=1}^{s}\big\{s_{\mathsf{B}_{i},\mathsf{B}_{j_k}}\big\}\,,\quad (j+2\leqslant i\leqslant m)\,.
\end{align}
The constraints $s_{\mathsf{B}_{i},l_j}$ and $s_{\mathsf{B}_{i},\mathsf{B}_{j_k}}$ can be written as
\begin{align}
X_{l_{i+1},l_{j_1}}&=c_{\mathsf{B}_{i},l_j}+X_{l_i,l_{i+1}}+X_{l_{i+1},l_j}+X_{l_i,l_{j_1}}-X_{l_i,l_j}\nonumber \ , \\
X_{l_{i+1},l_{j_{k+1}}}&=c_{\mathsf{B}_{i},\mathsf{B}_{j_k}}+X_{l_i,l_{i+1}}+X_{l_{j_k},l_{j_{k+1}}}+X_{l_i,l_{j_{k+1}}}+X_{l_{i+1},l_{j_k}}-X_{l_i,l_{j_k}}\,.
\end{align}
We can combine them into 
\begin{align}\label{eq:for3}
X_{l_{i+1},l_{j_k+1}}+X_{l_{j},l_{j+1}}&=X_{l_{j+1},l_{j_{k+1}}}+X_{l_{i+1},l_j}+\sum_{p=j+1}^{i}\sum_{\ell=1}^{k}\left(c_{\mathsf{B}_{p},\mathsf{B}_{j_\ell}}+X_{l_p,l_{p+1}}+X_{l_{j_\ell},l_{j_{\ell+1}}}\right)\nonumber\\
&\quad+\sum_{p=j+1}^{i}\left(c_{\mathsf{B}_{l_pr_p},l_j}+X_{l_p,l_{p+1}}\right)\,,
\end{align}
where $j+2\leqslant i\leqslant m$, $0\leqslant k\leqslant s$ and $l_{j_{s+1}}\equiv r_j$. Next, we use the seven-term identity to expand $s_{\mathsf{B}_i,\mathsf{B}_j}$ and $s_{\mathsf{B}_i,r_j}$ on the factoriztion channel $X_{l_j,l_{j+1}}=0$,
\begin{align}
-s_{\mathsf{B}_{i},\mathsf{B}_{j}}&=X_{l_i,l_j}+X_{l_{i+1},l_{j+1}}-X_{l_i,l_{j+1}}-X_{l_i,l_{i+1}}-X_{l_{i+1},l_j}\nonumber\\
-s_{\mathsf{B}_{i},r_j}&=c_{\mathsf{B}_{i},r_j}=X_{l_i,r_j}+X_{l_{i+1},l_{j+1}}-X_{l_i,l_{j+1}}-X_{l_i,l_{i+1}}-X_{l_{i+1},r_j}\,.
\end{align}
Add up the above two equations, we get
\begin{align}
-s_{\mathsf{B}_{i},\mathsf{B}_{j}}=c_{\mathsf{B}_{i},r_j}+X_{l_i,l_j}+X_{l_{i+1},r_j}-X_{l_{i+1},l_j}-X_{l_i,r_j}\,.
\label{eq:mnorss}
\end{align}
Combining eq.~\eqref{eq:mnorss} and~\eqref{eq:for3}, we find
\begin{align}
-s_{\mathsf{B}_{i},\mathsf{B}_{j}}&=c_{\mathsf{B}_{i},r_j}+\sum_{\ell=1}^{s}\left(c_{\mathsf{B}_{i},\mathsf{B}_{j_\ell}}+X_{l_i,l_{i+1}}+X_{l_{j_\ell},l_{j_{\ell+1}}}\right)+\left(c_{\mathsf{B}_{i},l_j}+X_{l_i,l_{i+1}}\right)\nonumber\\
&\equiv c_{{q},\mathsf{B}_i}+(s+1)X_{l_i,l_{i+1}}+\sum_{\ell=1}^{s}X_{l_{j_{\ell},l_{j_{\ell+1}}}}\,.
\label{finalequation}
\end{align}
Now we can use the replacement $s_{\mathsf{B}_i,\mathsf{B}_j}\rightarrow s_{q,\mathsf{B}_j}$ and find the constraint
\begin{align}
     -s_{q,\mathsf{B}_j}=c_{{q},\mathsf{B}_i}+(s+1)X_{l_i,l_{i+1}}+\sum_{\ell=1}^{s}X_{l_{j_{\ell},l_{j_{\ell+1}}}}
\end{align}
that belongs to $\pmb{H}^L$. Comparing with $-s_{q,\mathsf{B}_j}=c_{q,\mathsf{B}_j}$ that is given by the recursion in section~\ref{sec:invsoftconstruction}, we find that it follows the deformation pattern~\eqref{clsoedeforss} and at the same time the constant $c$ is linearly shifted by the variables $X_{l_{j_{\ell},l_{j_{\ell+1}}}}$ that belong to $\mathcal{A}_R$. Therefore, the facet displays a fiber product $\mathcal{A}'_L\ltimes\mathcal{A}_R$ that involves a deformed geometry $\mathcal{A}'_L$.

\bibliographystyle{JHEP}
\bibliography{reference}

\end{document}